\documentclass[trans]{IEEEtran}
\usepackage{hyperref}
\usepackage{lettrine}
\usepackage{epsfig}
\usepackage{footnote}
\usepackage{cite}
\usepackage{float}
\usepackage{amssymb}
\usepackage{amsmath}
\usepackage{amsthm}
\usepackage{lscape}
\usepackage{amsmath,mathabx}
\usepackage{color}
\usepackage{amsfonts}
\usepackage{epstopdf}
\usepackage{tabularx,tabulary}
\usepackage[font=footnotesize]{caption}
\usepackage{graphicx}
\usepackage{dblfloatfix}    
\usepackage{bbm}
\usepackage{algorithm,algorithmic}
\usepackage[usestackEOL]{stackengine}
\usepackage{color}
\usepackage{multicol}
\usepackage{multirow}
\usepackage{caption}
\usepackage{subcaption}
\usepackage[T1]{fontenc}

\newcommand{\E}{\mathbb{E}}
\newcommand{\vect}[1]{\boldsymbol{#1}}

\DeclareMathOperator{\diag}{diag}

\DeclareMathOperator{\argmin}{argmin}
\DeclareMathOperator{\argmax}{argmax}

\newcommand{\bigtriangleq}{\mathbin{\setstackgap{S}{0pt}\stackMath\Shortstack{\smalltriangleup\\ =}}}

\newcommand*\rfrac[2]{{}^{#1}\!/_{#2}}
\newtheorem{theorem}{Theorem}
\newtheorem{lemma}[theorem]{Lemma}

\begin{document}
\title{Distributed User Clustering and  Resource Allocation for Imperfect NOMA in Heterogeneous Networks }
\author{
Abdulkadir~Celik,~\IEEEmembership{Member,~IEEE,} Ming-Cheng~Tsai,~\IEEEmembership{Student Member,~IEEE,} \\ Redha~M.~Radaydeh,~\IEEEmembership{Senior Member,~IEEE,} Fawaz~S.~Al-Qahtani,~\IEEEmembership{Member,~IEEE,} \\ and~Mohamed-Slim~Alouini,~\IEEEmembership{Fellow,~IEEE}

\thanks{ 
A. Celik, Ming-Cheng~Tsai, and M-S. Alouini are with Computer, Electrical, and Mathematical Sciences and Engineering Division at King Abdullah University of Science and Technology (KAUST), Thuwal, KSA. Corresponding author: Abdulkadir Celik (abdulkadir.celik@kaust.edu.sa).

R. M. Radaydeh is with the Electrical Engineering Program, Department of Engineering and Technology, Texas A\& M University-Commerce, , Commerce,
TX 75428 USA.

F. S. Al-Qahtani is with Research, Development, Innovation (RDI) Division of Qatar Foundation, Doha, Qatar.  

This work was supported by NPRP from the Qatar National Research Fund under grant no. 8-1545-2-657. A part of this paper has been presented at IEEE GLOBECOM 2018, Singapore \cite{celik2017ULNOMA}. 
}}

\markboth{ IEEE Transactions on Communications}{Celik \MakeLowercase{\textit{et al.}}: Distributed Cluster Formation and Power-Bandwidth Allocation for Imperfect NOMA in DL-HetNets}
\maketitle
\begin{abstract}
In this paper, we propose a distributed cluster formation (CF) and resource allocation (RA) framework for non-ideal non-orthogonal multiple access (NOMA) schemes in heterogeneous networks. The imperfection of the underlying NOMA scheme is due to the receiver sensitivity and interference residue from non-ideal successive interference cancellation (SIC), which is generally characterized by a fractional error factor (FEF). Our analytical findings first show that several factors have a significant impact on the achievable NOMA gain. Then, we investigate fundamental limits on NOMA cluster size as a function of FEF levels, cluster bandwidth, and quality of service (QoS) demands of user equipments (UEs). Thereafter, a clustering algorithm is developed by taking feasible cluster size and channel gain disparity of UEs into account. Finally, we develop a distributed $\alpha$-fair RA framework where $\alpha$ governs the trade-off between maximum throughput and proportional fairness objectives. Based on the derived closed-form optimal power levels, the proposed distributed solution iteratively updates bandwidths, clusters, and UEs' transmission powers. Numerical results demonstrate that proposed solutions deliver a higher spectral and energy efficiency than traditionally adopted basic NOMA cluster size of two. We also show that an imperfect NOMA cannot always provide better performance than orthogonal multiple access under certain conditions. Finally, our numerical investigations reveal that NOMA gain is maximized under  downlink/uplink decoupled (DUDe) UE association.

\end{abstract}

\IEEEpeerreviewmaketitle

\section{Introduction}
\label{sec:intro}
\lettrine{E}{ver} increasing number of communications devices with the ambitious quality of service (QoS) demands puts forward challenging goals for fifth-generation (5G) networks such as massive connectivity, enhanced mobile broadband, ultra-reliability, low-latency,  etc. To fulfill such demands,  ultra-dense heterogeneous networks (HetNets) have already been considered as a promising solution since densification of the network has the ability to boost the network coverage and capacity while reducing the operational and capital expenditures of mobile operators  \cite{celik2017ULNOMA}. However, traditional orthogonal multiple access (OMA) schemes employed by today's HetNets dedicate radio resources to a certain user either in time, frequency, or code domains, which is not adequately spectral efficient for the expected massive number of users. 
Having its root in multi-user detection theory, non-orthogonal multiple access (NOMA) can momentarily serve multiple users on the same radio resource by multiplexing them either in power or code domain \cite{Dai15}. As a result, it has recently gained attention with its ability to serve toward higher spectral efficiency and massive connectivity goals of the next generation networks. In particular, power domain NOMA ensures a certain reception power for each user such that some users operate in low power levels in order to cancel dominant interference using successive interference cancellation (SIC) while some others transmit at high power levels at the expense of limited interference cancellation (IC) opportunity. Even though such a strategy paves the way for a notion of fairness embedded in NOMA, the impacts of fair bandwidth scheduling on the network performance is still an interesting phenomenon to be investigated.

In order to extract the desired signal, the SIC receiver first decodes the strongest interference, then re-generates the transmitted signal, and finally subtracts it from the received composite signal, which is repeated for succeeding interference components. However, a real-life NOMA system is required to account for the following challenges of a practical SIC receiver: First, a more complicated power control policy is necessary since decoder needs to observe a minimum SINR at each cancellation stage, which is mainly characterized by the receiver sensitivity \cite{Andrews2003}. Thus, the optimal power control strategy must comply with the resulting disparity of the received power levels, which is also referred to as \textit{SIC constraints}. Second, system performance can substantially be deteriorated due to the amplitude and phase estimation errors which determine the residual interference after SIC and is often quantified by a \textit{fractional error factor} (FEF) \cite{SIC_OFDMA}. Therefore, it is necessary to develop an optimal power and bandwidth allocation scheme which accounts for these practical challenges.   

Furthermore, cluster formation (CF) is of the utmost importance to maximize the gain achieved by a NOMA scheme. Ongoing research efforts typically consider a perfect NOMA scheme for basic cluster of size two, which directly reduces the clustering to a pairing problem. However, clustering more user equipments (UEs) to share the same bandwidth provides an improved spectral efficiency at the cost of SIC delay which linearly increases with the cluster size. Therefore, CF problem involves two main tasks: 1) Determining the optimal cluster sizes and 2) Grouping UEs to maximize the overall network performance. Accordingly, this paper addresses a distributed framework for cluster formation (CF) and $\alpha$-fair resource allocation for UL-HetNets under imperfect NOMA scheme.

\subsection{Related Works}
Related works on NOMA can be exemplified as follows: The impact of UE grouping is investigated in \cite{Ding2015TVT} for a two-UE DL-NOMA system with fixed and cognitive radio inspired power allocation schemes. The work in \cite{Liu2016fairness} proposed three different sub-optimal approaches for max-min fair UE clustering problem. Authors of \cite{Liu2017Joint} iteratively built clusters where each iteration jointly optimizes beam-forming and power allocation for given clusters. User pairing for UL-NOMA is investigated in \cite{Sedaghat2018} which divides the set of UEs into disjunct pairs and assigns the available resources to these pairs by considering some predefined power allocation schemes. In \cite{Kang2018}, authors study the optimal user pairing for the NOMA system in the sense of maximizing the total sum rate. In \cite{Elbamby2017}, authors study the problem of resource optimization, mode selection, and power allocation subject to queue stability constraints under the assumption of in-band full duplex base stations (BSs). Beam-forming and power allocation of a multiple-input-multiple-output (MIMO) NOMA system are investigated in \cite{Ali2017beam} where two-UE clusters are formed from high and low channel gain UEs with the consideration of channel gain correlations. The work in \cite{Ali2016} proposed a UL power allocation scheme by first grouping users into a single cluster and then optimizing the power allocation. 

In \cite{Ding2019Improved}, Ding et. al. proposed a cluster beamforming strategy to jointly optimize beamforming vectors and power allocation coefficients for MIMO-NOMA clustering with the purpose of energy efficiency. The sum rate maximization problem of mm-wave-NOMA systems is investigated in \cite{Cui2018Unsupervised} where authors also develop a K-means-based machine learning algorithm for user clustering. In \cite{Tsai2018Quality}, a suboptimal quality-balanced clustering approach is proposed to optimize the total sum rate in a system. The impacts of channel state information (CSI) imperfections on the NOMA performance are investigated in  \cite{Fang2017Joint, Wei2017Optimal, Xiao2018Reinforcement}: Energy efficient resource scheduling is addressed in \cite{Fang2017Joint} where authors also account for imperfect CSI for a generic cluster size. In \cite{Wei2017Optimal}, power-efficient resource allocation is studied for multicarrier NOMA systems.  Accounting for the imperfection of CSI at transmitter side, a solution is proposed to  jointly design the power-rate allocation, user scheduling, and SIC decoding policy for minimizing the total transmit power.  An interesting problem is investigated in \cite{Xiao2018Reinforcement} for NOMA systems vulnerable to jamming attacks. Authors proposed a reinforcement learning-based power control scheme without being aware of the jamming and CSI. 
Except \cite{Ali2016,Ali2017beam, Fang2017Joint,Ding2019Improved,Cui2018Unsupervised,Tsai2018Quality}, proposed methods in these works mostly focus on the basic form of a NOMA (i.e., pairing only two UEs where power allocation is analytically more tractable) by ignoring the benefits of incorporating more UEs. 

Considering the massive connectivity goal of the future networks, it is important to investigate NOMA schemes that allow larger cluster sizes for the sake of spectral efficiency and increased connectivity. Since it is possible to employ sophisticated SIC receivers at BSs with high computational power, possible IC delay of larger cluster sizes can be mitigated in order to enhance UL-NOMA performance.

In \cite{Tabassum2017modeling}, multi-cell uplink NOMA systems are analyzed using the theory of Poisson cluster process. The impact of channel gain disparity on DL-NOMA is investigated in \cite{Ding2015TVT} for a two-user system with fixed  and cognitive radio inspired power allocation. A near optimal solution was proposed by combining Lagrangian duality and dynamic programming for joint power and channel allocation in \cite{Lei2016}. In \cite{Zhang2016}, the authors derived closed-form expressions for the outage probability of two-user UL-NOMA assuming fixed powers of different users.  A simple power and rate allocation scheme for UL-NOMA is developed for a multicarrier system in \cite{Choi2018} where a practical modulation and coding scheme is employed at each UE. In \cite{Pappi2017},  a distributed UL-NOMA scheme is proposed for cloud radio access networks, which can offer substantial improvement over benchmark schemes. In these works, authors mostly consider a basic NOMA cluster of size two  except \cite{Ali2016} where authors develop a general DL and UL power control framework for a generic cluster size.

In \cite{Yang2016}, a dynamic power allocation scheme is proposed for both DL and UL NOMA scenarios with two users with various QoS requirements. User clustering and power-bandwidth allocation of HetNets is studied in \cite{Celik2017DLNOMA, celik2017ULNOMA} where clusters are formed according to different objectives such as maximum sum-rate, max-min fairness, and energy-spectrum cost minimization. The work in \cite{Celik2017DLNOMA} is further extended for for DL-HetNets in \cite{Celik2019Distributed} where a user clustering and power-bandwidth allocation is proposed. The main limitation of this work is treating the cluster size as a given design parameter. However, it is necessary to analyze the maximum permissible cluster size for UL-HetNets since the next-generation networks are expected to accommodate the massive connectivity. Therefore, allowing more low-power users on the same resource block is desirable for serving a large number of users and increasing the spectral efficiency of the NOMA scheme. In this regard, the proposed user clustering in this paper is apart from that of \cite{Celik2019Distributed} such that we analytically derive the maximum cluster size as a function of QoS demands, channel quality, cluster bandwidth, and SIC efficiency. Accordingly, the proposed clustering algorithm forms the clusters jointly with bandwidth allocation and ensures the QoS demand of each cluster is satisfied. Noting that \cite{Celik2019Distributed} is not involved with resource allocation fairness, our closed-form power allocations are also different since UL-UEs do not compete for a common power source. Excluding \cite{Celik2017DLNOMA, celik2017ULNOMA,Celik2019Distributed}, these works also do not consider the residual interference caused by the error propagation during the IC process.

\subsection{Main Contributions}
Our main contributions can be summarized as follows: 
\begin{itemize}
\item 
An imperfect NOMA scheme is investigated in order to account for practical SIC constraints due to the receiver sensitivity and residual interference. Our analytical findings show that decoding order, SIC constraints, residual interference, and channel gain disparity of cluster members have a significant impact on the achievable NOMA gain. These findings are then supported with numerical results which clearly demonstrate that NOMA cannot always provide a better performance than OMA depending upon the SINR requirements and FEF levels (i.e, residual interference).  
\item
Existing works on NOMA typically assume a basic cluster size of two and simply pair UEs to form clusters. However, cluster sum-rates and spectral efficiency increase with the cluster size as more UEs share the same bandwidth.  Hence, the largest feasible cluster size is first analytically obtained as a function of FEF levels, cluster bandwidth, and QoS requirements of cluster members. Then, we show that a given bandwidth can accommodate more UEs as the SINR requirements and FEF levels decrease, which is especially beneficial to outline capabilities of NOMA to serve for the massive connectivity. Thereafter, we propose a cluster formation method where each BS first determines the largest feasible cluster sizes and then iteratively matches its UEs with clusters to maximize the channel gain disparity for an improved NOMA gain. 
\item 
Based on the developed cluster formation method, we develop a distributed $\alpha$-fair resource allocation methodology where $\alpha \in [0,1]$ manages the balance between maximum throughput and proportional fairness. Resource allocation is decomposed into power control and bandwidth allocation problems. Based on the derived closed-form power control expression of the considered imperfect NOMA scheme, the proposed algorithm iteratively updates bandwidth allocations, cluster sizes, and transmission powers to maximize the $\alpha$-fair network objective. Finally, the performance gain of the developed algorithm is compared with OMA and basic NOMA schemes under different system network parameters such as BS/UE density, traffic offloading factor, FEF, and QoS requirements.   
\end{itemize}

\subsection{Notations and Paper Organization}
Throughout the paper, sets and their cardinality are denoted with calligraphic and regular uppercase letters (e.g., $| \mathcal{A}|=A$), respectively. Vectors and matrices are represented in lowercase and uppercase boldfaces (e.g., $\vect{a}$ and $\vect{A}$), respectively. Superscripts $b$, $c$, and $i$ are used for indexing BSs/cells, clusters, and UEs, respectively. 
The remainder of the paper is organized as follows: Section \ref{sec:sys_mod} introduces the network model and UE association schemes. Section \ref{sec:SIC} addresses constraints, decoding order, and imperfections of practical SIC receivers and their impacts on achievable NOMA gain. Section \ref{sec:CF-PBA} first provides the problem statement and then gives an overview of proposed solution methodology. Section \ref{sec:CF} analyses the feasible cluster size and develops a clustering algorithm.  Section \ref{sec:PBA} addresses proposed $\alpha$-fair distributed power and bandwidth allocation along with the algorithm of overall solution. Numerical results are presented in Section \ref{sec:res} and Section \ref{sec:conc} concludes the paper with a few remarks.

\section{Network Model}
\label{sec:sys_mod}

We consider UL transmission of a 2-tiered HetNet that operates on a single-input and single-output NOMA scheme. Each tier represents a particular cell class, i.e., tier-1 consists of a single macrocell and tier-2 comprises of smallcells. The index set of all BSs is denoted by $\mathcal{B}=\{b \: | \: 0 \leq b \leq B \}$ where $B$ denotes the number of small base staions (SBSs), $b=0$ is for the macro base station (MBS) index, and $1 \leq b \leq B$ are indices for SBSs, respectively\footnote{The terms BS, cell, and their indices are used interchangeably throughout the paper.}.  Maximum transmission powers of UEs and BSs are denoted as  $\bar{P}_u$ and $\bar{P}_c$, respectively, where $\bar{P}_c$ equals to $\bar{P}_m$ and $\bar{P}_s$ for the MBS and SBSs, respectively.

Association of UEs with the BS can be done either in a DL/UL coupled (DUCo) or decoupled (DUDe) fashion. Conventional DUCo scheme associates UEs with the same BS for both DL and UL transmission based on received signal strength information (RSSI), which yields a significant traffic load on macrocells due to MBS's high transmission power. Therefore, UE association is typically done by introducing a bias factor, $0 \leq \flat \leq 1$, in order to offload DL traffic from MBSs to SBSs. Nonetheless, requiring UEs to follow the same association in both UL and DL may always not yield a desirable performance. While eeping DL association method the same as in DUCo, DUDe scheme alternatively determines the UL association based on channel gain such that a UE can be associated with a nearby SBS in the UL even if it is associated with the MBS in the DL \cite{boccardi2016decouple,celik2017joint}.

Contingent upon the user associations, index set of all $U \triangleq \sum_b U_b$ UEs is given as $\mathcal{U}\triangleq\bigcup_b \mathcal{U}_b$ where $\mathcal{U}_b$ is the set of $U_b$ UEs associated with BS$_b$. Each BS partitions $\mathcal{U}_b$ into disjoint $C_b$ clusters such that $\mathcal{K}_b^c$ symbolizes the set of $K_b^c$ UEs within cluster $c$, i.e., $U_b=\sum_{c\in \mathcal{C}_b} K_b^c$. Similarly, the set of all $C$ clusters are denoted as $\mathcal{C}\triangleq \bigcup_b \mathcal{C}_b$ where $\mathcal{C}_b$ is the set of $C_b$ clusters of BS$_b$. Entire UL bandwidth is divided into $\Theta$ resource blocks (RBs) each of which has a bandwidth of $W$ Hz. The available set of RBs can be exploited by $C$ clusters based on an $\alpha$-fair resource allocation policy. The number of RBs allocated to $\mathcal{K}_b^c$ is denoted as $\theta_b^c \in [0, \Theta]$, $\sum_{b,c}\theta_b^c \leq \Theta$. For the remainder of the paper, we assume that a UE can be associated with exactly one cluster at a time and allocated RBs are dedicated to the corresponding clusters.

\section{Impacts of Constraints on NOMA Gain}
\label{sec:SIC}
In this section, we first introduce the constraints and imperfections of a practical SIC receiver, then analyze the impacts of decoding order and receiver sensitivity on the achievable NOMA gain. 

\subsection{Constraints and Imperfections of SIC Receivers} 
  Let us now focus on a generic cluster of BS$_b$ $\mathcal{K}_b^c = \{i| \: i \in \mathcal{U}_b, h_{i-1}^b \geq h^b_i \geq h^b_{i+1}, \delta_{b,c}^i=1 \}$ where $\delta_{b,c}^i \in \{ 0,1 \}$ is a binary indicator for cluster membership. For the UL-NOMA transmission, we consider the following decoding order 
\begin{equation}
\label{eq:order}
\underbrace{\omega_{b,c}^{K_b^c} h_{K_b^c}^b < ... < }_{\substack{\text{Lower Rank Decoding Order} \\ \mathcal{O}_i^\ell\text{ that can not be cancelled}}} \omega_{b,c}^i h_i^b\underbrace{< ... < \omega_{b,c}^1 h_1^b}_{\substack{\text{Higher Rank Decoding Order} \\ \mathcal{O}_i^h \text{ that can be cancelled}}},
\end{equation}
where $h_i^b $ is the composite channel gain from UE$_i$ to BS$_b$, $\omega_{b,c}^i h_i^b $ is the power received from UE$_i$ which is normalized by the maximum transmission power $\bar{P}_u$, $\omega_{b,c}^i $ is the power allocation weight, $\mathcal{O}_i^\ell=\{i+1, \ldots, K_b^c \}$ is the lower rank decoding order set, and $\mathcal{O}_i^h=\{1, \ldots, i-1\}$ is the higher rank decoding order set for UE$_i$. Notice that the UL decoding order is the reverse of DL order considered in the literature, which is addressed in the next section. Accordingly, a generic SINR representation of the imperfect SIC receiver can be given by
\begin{align}
\label{eq:sinr_c_dl}
\Gamma_{b,c}^{i}& =\frac{\delta_{b,c}^{i} \omega_{b,c}^{i} h_i^b}{\epsilon_b  \sum_{\substack{j =1 \\j \in \mathcal{K}_b^c}}^{i-1} \delta_{b,c}^{j}  \omega_{b,c}^{j} h_j^b
+ \sum_{\substack{k =i+1 \\j \in \mathcal{K}_b^c}}^{K_b^c} \delta_{b,c}^{k}  \omega_{b,c}^{k} h_k^b
+\varrho_b^c},
\end{align} 
where $0 \leq \epsilon_b \leq 1$ is the FEF of BS$_b$ which characterizes the residual interference, $\varrho_b^c \bigtriangleq \sigma \theta_b^c /\bar{P}_u$, $\sigma=N_0 \theta_b^c W$ is the thermal receiver noise power, and $N_0$ is the noise power spectral density. The first term of the denominator represents the residual interference after cancellation which can indeed be linked to the SIC efficiency, i.e., $(1-\epsilon_b)$. On the other hand, the second term of denominator represents the uncancelled lower rank interference. 

The first term of the denominator represent the residual interference after cancellation which can indeed be linked to the SIC efficiency, i.e., $(1-\epsilon_b)$. On the other hand, the second term of denominator represents the uncancelled lower rank interference. The residual interference is primarily caused by amplitude, phase, and channel estimation errors, which lead to imperfect regeneration of the received signals. Another source of SIC imperfections is erroneous bit decisions in the previously decoded users. Under low bit-error rate requirements ($<10^{-5}$), error propagation of bit decisions is also a result of the imperfect estimations \cite{Hasan2004cancellation}. Multistage detection, error correction coding, iterative detection, enhanced channel estimation are among the key techniques to ameliorate FEF levels of SIC receivers\cite{Andrews2005SIC}. That is, the FEF is an important hardware parameter to be taken into account in power allocation strategy because being FEF agnostic can substantially deteriorate the NOMA performance as $\epsilon_b \rightarrow 1$.

For a desirable performance, a cluster member should be able to cancel the dominant interference while tolerating the SIC imperfection and interference induced from lower rank UEs.  Following from \eqref{eq:sinr_c_dl}, the achievable data rate of UE$_i$ is given by 
\begin{equation}
\label{eq:capacity}
\mathrm{R}_i=W \theta_b^c \log_2(1+ \Gamma_{b,c}^i), \forall i \in \mathcal{K}_b^c.
\end{equation}

$\mathrm{R}_i$ is generally required to be higher than a certain service rate agreement, $\mathrm{R}_i \geq \overline{\mathrm{R}}_i, \forall i$, which is referred to as \textit{QoS constraint}\footnote{Instead of the inelastic traffic conditions where users require a minimum instantaneous throughput requirements, we are interested in elastic users with average QoS demands over a long time period.} and given by 

\begin{equation}
\label{eq:QoSC}
\Gamma_{b,c}^i \geq 2^{\frac{\bar{\mathrm{R}}_i}{\theta_b^cW}}-1, \forall i \in \mathcal{K}_b^c, \forall b, \forall c.
\end{equation}
On the other hand, SIC constraints are given by
\begin{equation}
\label{eq:SICconst}
\Gamma_{b,c}^{i} \geq 10^{\frac{\S_b}{10}}, \forall i \in \mathcal{K}_b^c, \forall b, \forall c,
\end{equation}
where  $\S_b$ is the receiver sensitivity of BS$_b$ which is often given in units of dB. These two constraints can be combined and projected onto SINRs as a unified constraint as follows
\begin{equation}
\label{eq:Gammabar}
\Gamma_{b,c}^i  \geq \bar{\Gamma}_{b,c}^i (\theta_b^c)=\max \left( {10^{\frac{\S_b}{10}},2^{\frac{\bar{\mathrm{R}}_i}{\theta_b^cW}}-1} \right), \forall i \in \mathcal{K}_b^c,
\end{equation}
which is referred to as \textit{composite SINR constraints} (CSCs) in the remainder of paper. 

\subsection{Impacts of SIC Constraints and Imperfections} 
\label{sec:order_analysis}
Even though DL-NOMA decodes UE signals in descending order of their channel gains, employing the same order in UL-NOMA may not give the desired performance under CSCs. To be more specific, let us consider a basic NOMA cluster of UE$_k$ and UE$_\ell$ with channel gains $h_k$ and $h_\ell$, $h_k \geq h_\ell$, and composite SINR demands $\bar{\Gamma}_k$ and $\bar{\Gamma}_\ell$, respectively. 
\subsubsection{Descending Order}
Employing the descending decoding order as in the DL case (i.e., UE$_k$ cancel the interference of UE$_\ell$), OMA and NOMA sum-rates can be respectively given as
\begin{align}
\label{eq:CO}
&\mathrm{R}_{\downarrow}^O= \rfrac{1}{2}\left\{ \log_2 \left( 1+ \rho h_k \right) + \log_2 \left( 1+ \rho h_\ell \right)\right \}, \\
\label{eq:CN}
& \mathrm{R}_{\downarrow}^N = \log_2 \left( 1+ \frac{\omega_k h_k }{ \epsilon \omega_\ell h_\ell + \rfrac{1}{\rho}}\right) + \log_2 \left( 1+ \frac{\omega_\ell h_\ell }{ \omega_k h_k + \rfrac{1}{\rho}}\right),
\end{align}
where $0 \leq \omega_k,  \omega_\ell  \leq 1$ are power weights and $\rho=\rfrac{\bar{P}_u}{\sigma}$. As $\rho \rightarrow \infty$ and $\epsilon \rightarrow 0$, asymptotic capacity of OMA and perfect NOMA can be respectively expressed as
\begin{align}
\label{eq:CO_asymp}
  \lim_{\substack{\rho \rightarrow \infty \\ \epsilon \rightarrow 0 }}  \mathrm{R}_{\downarrow}^O & \simeq \rfrac{1}{2} \{ \log_2 \left( \rho h_k \right) + \log_2 \left( \rho h_\ell \right) \}= \rfrac{1}{2} \log_2 \left( \rho^2 h_k h_\ell \right),\\
\label{eq:CN_asymp}
\lim_{\substack{\rho \rightarrow \infty \\ \epsilon \rightarrow 0 }} \mathrm{R}_{\downarrow}^N & \simeq  \log_2 \left(\rho \omega_k h_k \right) + \log_2 \left( 1+ \frac{\omega_\ell h_\ell }{ \omega_k h_k}\right) \simeq   \log_2 \left(\rho \omega_k h_k \right),
\end{align}
where \eqref{eq:CO_asymp} and  \eqref{eq:CN_asymp}  follow from the facts that $(1+ \rho h_k) \simeq \rho h_k$ as $\rho \rightarrow \infty$ and the second term of \eqref{eq:CN} becomes negligible as $\rho \rightarrow \infty$, respectively. Accordingly, asymptotic gain of NOMA scheme can be given by
\begin{align}
 \nonumber {\Delta}_{\downarrow} &\bigtriangleq  \lim_{\substack{\rho \rightarrow \infty \\ \epsilon \rightarrow 0 }}  \left( \mathrm{R}_{\downarrow}^N -\mathrm{R}_{\downarrow}^O \right) = \log_2 \left(\rho \omega_k h_k \right)  - \rfrac{1}{2} \log_2 \left( \rho^2 h_k h_\ell \right)\\
& = \log_2 \left( \frac{\rho \omega_k h_k}{\rho \sqrt{h_k h_\ell }}\right) = \log_2 \left( \omega_k \sqrt{\frac{h_k}{h_\ell}} \right)
\end{align}
In the descending order, the SIC constraint requires $   \lim_{\rho \rightarrow \infty } \frac{\omega_\ell h_\ell }{ \omega_k h_k + \rfrac{1}{\rho}} \underset{}{\geq} 10^{\frac{\S_b}{10}}$ that reduces to a power disparity constraint, i.e., $\frac{\omega_\ell h_\ell }{ \omega_k h_k} \underset{}{\geq} 10^{\frac{\S_b}{10}}$. Even for a SIC receiver with perfect sensitivity, i.e., $\S_b \rightarrow 0$, power disparity constraint constitutes $\frac{\omega_\ell  h_\ell}{ h_k } \geq \omega_k$, thus the upper bound on ${\Delta}_{\downarrow}$ is given by
\begin{equation}
\label{eq:delta_down_asymp2}
{\Delta}_{\downarrow} \leq \log_2 \left( \omega_\ell \sqrt{\frac{h_\ell}{h_k}} \right) \leq \rfrac{1}{2} \left\{ \log_2 (h_\ell) -\log_2 (h_k)\right\} ,
\end{equation}
which is always non-positive due to $\rfrac{h_\ell}{h_k}\leq 1$. That is, sum-rate of UL-OMA and the descending ordered UL-NOMA perform the same for users with equal channel gains. For non-equal channel gain cases, UL-NOMA provides a worse performance which is deteriorated even further for imperfect NOMA case as $\epsilon \rightarrow 1$.

\subsubsection{Ascending Order} Following the similar steps in (\ref{eq:CO})-(\ref{eq:CN_asymp}), asymptotic NOMA gain for the ascending order case (i.e., UE$_\ell$ cancel the interference of UE$_k$) can be obtained as 
\begin{align}
\nonumber {\Delta}_{\uparrow} &\bigtriangleq \lim_{\substack{\rho \rightarrow \infty \\ \epsilon \rightarrow 0 }} \left( \mathrm{R}_{\uparrow}^N -\mathrm{R}_{\uparrow}^O \right)= \log_2 \left(\rho \omega_\ell h_\ell \right)  - \rfrac{1}{2} \log_2 \left( \rho^2 h_k h_\ell \right)\\
\label{eq:delta_up_asymp}
& = \log_2 \left( \frac{\rho \omega_\ell h_\ell}{\rho \sqrt{h_k h_\ell }}\right) = \log_2 \left( \omega_\ell \sqrt{\frac{h_\ell}{h_k}} \right).
\end{align}
In the ascending order, the SIC constraint requires $   \lim_{\rho \rightarrow \infty } \frac{\omega_k h_k }{ \omega_\ell h_\ell + \rfrac{1}{\rho}} \underset{}{\geq} 10^{\frac{\S_b}{10}}$ that reduces to a power disparity constraint, i.e., $\frac{ \omega_k h_k}{\omega_\ell h_\ell } \underset{}{\geq} 10^{\frac{\S_b}{10}}$. For a SIC receiver with perfect sensitivity, i.e., $\S_b \rightarrow 0$, power disparity constraint constitutes $\frac{\omega_k h_k}{h_\ell} \geq \omega_\ell  $, thus the upper bound on ${\Delta}_{\downarrow}$ is given by
\begin{equation}
\label{eq:delta_up_asymp2}
{\Delta}_{\uparrow} \leq \log_2 \left( \omega_k \sqrt{\frac{h_k}{h_\ell}} \right) \leq \rfrac{1}{2} \left\{ \log_2 (h_k) -\log_2 (h_\ell)\right\},
\end{equation}
which is always non-negative due to  $\rfrac{h_k}{h_\ell} \geq 1$. That is, sum-rate of UL-OMA and the descending ordered UL-NOMA perform the same for users with equal channel gains whereas UL-NOMA provides a superior performance proportional to the channel gain disparity of users. Unfortunately, this desirable performance gain obtained by channel gain disparity of users  naturally diminishes as $\epsilon$ increases and NOMA yields a worse performance than OMA after a certain point, which is investigated in the remainder of the paper.

\section{Cluster formation and  Resource Allocation}
\label{sec:CF-PBA}
Centralized CF and RA is a combinatorial problem whose solution requires impractical time complexity even for moderate size of HetNets. Since a fast yet high performance solution is of the essence to employ NOMA in large-scale HetNets, this section first makes a problem statement by formulating a centralized problem then outlines the proposed distributed solution methodology to mitigate the high communication and computational overhead of centralized solutions.

\subsection{Centralized Problem Formulation}
\label{sec:optimal}
In order to investigate fair power and bandwidth allocation schemes, we adopt a generalized throughput formulation that has been proposed by the nominal work in \cite{origin_alpha} where the degree of fairness is adjusted by a single parameter $\alpha \in [0,1]$. In other words,  $\alpha$ manages the compromise between throughput maximization and fairness by means of the generalized $\alpha$-fair function which can be expressed as 
\begin{align}
\pi_{b,c}^i=\begin{cases}
\frac{1}{1-\alpha}R_i^{1-\alpha}\left(\delta_{b,c}^{i}, \theta_b^c, \omega_{b,c}^{i} \right), &\text{ for } 0\leq \alpha <1 \\
\log\left[R_i \left(\delta_{b,c}^{i}, \theta_b^c, \omega_{b,c}^{i} \right)\right], & \text{ for } \alpha = 1 
\end{cases},
\end{align}
which corresponds to the throughput maximization if $\alpha=0$ and proportional fairness if $\alpha \rightarrow1$. For the sake of a unified and continuous form of the fairness function, we exploit the following $\alpha$-fair objective function  \cite{general_alpha}
\begin{align}
\nonumber \Pi \left( \vect{\delta}, \vect{\theta}, \vect{\omega} \right)&=\sum_{\forall (b,c,i)} \pi_{b,c}^i \left(\delta_{b,c}^{i}, \theta_b^c, \omega_{b,c}^{i} \right)\\
&\hspace{-3pt} =\sum_{\forall (b,c,i)}\frac{1}{1-\alpha}\left(R_i^{1-\alpha}\left(\delta_{b,c}^{i}, \theta_b^c, \omega_{b,c}^{i} \right)-1\right).
\end{align}

Accordingly, a centralized CF and RA problem can be formulated as in $\mathrm{P_o}$ where $\mathrm{C}_o^1$ ensures that UEs are  assigned to exactly one cluster and $\mathrm{C}_o^2$  limits the number of UEs within a cluster by $K_b^c$.  $\mathrm{C}_o^3$ constraints the total number of RB allocation to available number of RBs, $\Theta$. The power weight limitation on UE$_i$ is introduced in $\mathrm{C}_o^4$ where the power allocation for UE$_i$ on cluster $c$ is set to zero if UE$_i \notin \mathcal{K}_b^c$. CSCs are given by $\mathrm{C}_o^5$ in order to account for QoS and SIC constraints. Finally, $\mathrm{C}_o^6$ indicates the variable domains.

\begin{equation}
\hspace*{0pt}
 \begin{aligned}
 & \hspace*{0pt} \mathrm{P_o}:   \underset{\vect{\delta}, \vect{\theta}, \vect{\omega}}{\max}
& & \hspace*{3 pt}  \Pi(\vect{\delta}, \vect{\theta}, \vect{\omega})\\
& \hspace*{0pt} \mbox{$\mathrm{C}_o^1$:}\hspace*{20pt} \text{s.t.}
&&  \sum_{c} \delta_{b,c}^i = 1, \textbf{ } \hspace{26 pt} \forall b,i\\ 
 &
 \hspace*{0 pt}\mbox{$\mathrm{C}_o^2$: } & & \sum_{i} \delta_{b,c}^i\leq K_b^c, \textbf{ } \hspace{18 pt} \forall b,c \\
    &
\hspace*{0 pt}\mbox{$\mathrm{C}_o^3$: } & & \sum_{b,c}\theta_b^c\leq \Theta,\\
    &
\hspace*{0 pt}\mbox{$\mathrm{C}_o^4$: } & & \omega_{b,c}^i  \leq \delta_{b,c}^i, \hspace{18 pt} \forall b,c,i \\
& 
  \hspace*{0 pt}\mbox{$\mathrm{C}_o^5$: } & &\bar{\Gamma}_{b,c}^i (\theta_b^c)  \leq \Gamma_{b,c}^i , \hspace{18pt} \forall b,c,i  \\
&
  \hspace*{0 pt}\mbox{$\mathrm{C}_o^6$: } & &  \delta_{b,c}^i \in \{0,1\}, K_b^c \in [0, \rfrac{U}{2} ], \theta_b^c \in [0,1] \hspace{2 pt} 
\end{aligned}
\end{equation}

\subsection{Hierarchically Distributed Solution}
In $\mathrm{P_o}$, obtaining optimal integer valued cluster sizes and binary valued UE-cluster associations yields an NP-Hard mixed integer non-linear programming (MINLP) problem whose time complexity exponentially increases with the number of network entities. Moreover, highly non-convex nature of resource allocation problem puts an additional degree of complexity. Also noting the undesirable communication overhead of centralized solutions, developing fast yet near optimal distributed solutions is of interest to be employed in practice. 

As shown in Fig. \ref{fig:dist}, we develop a distributed solution methodology where we first decouple the CF and power allocation problems by considering the channel gain disparity of cluster members as the main credential of cluster formation policy. This is primarily motivated by the analytical findings of Section \ref{sec:SIC} which shows that NOMA gain is determined by the channel gain disparity of the cluster members. In this way, each BS can independently form its own clusters since they are generally aware of the channel states of associated UEs. Notice that the CF problem is still coupled by the bandwidth allocations since the maximum permissible cluster size is a function of the cluster bandwidth as explained in the next section.

\begin{figure}[t!]
\centering
 \includegraphics[width=0.4\textwidth]{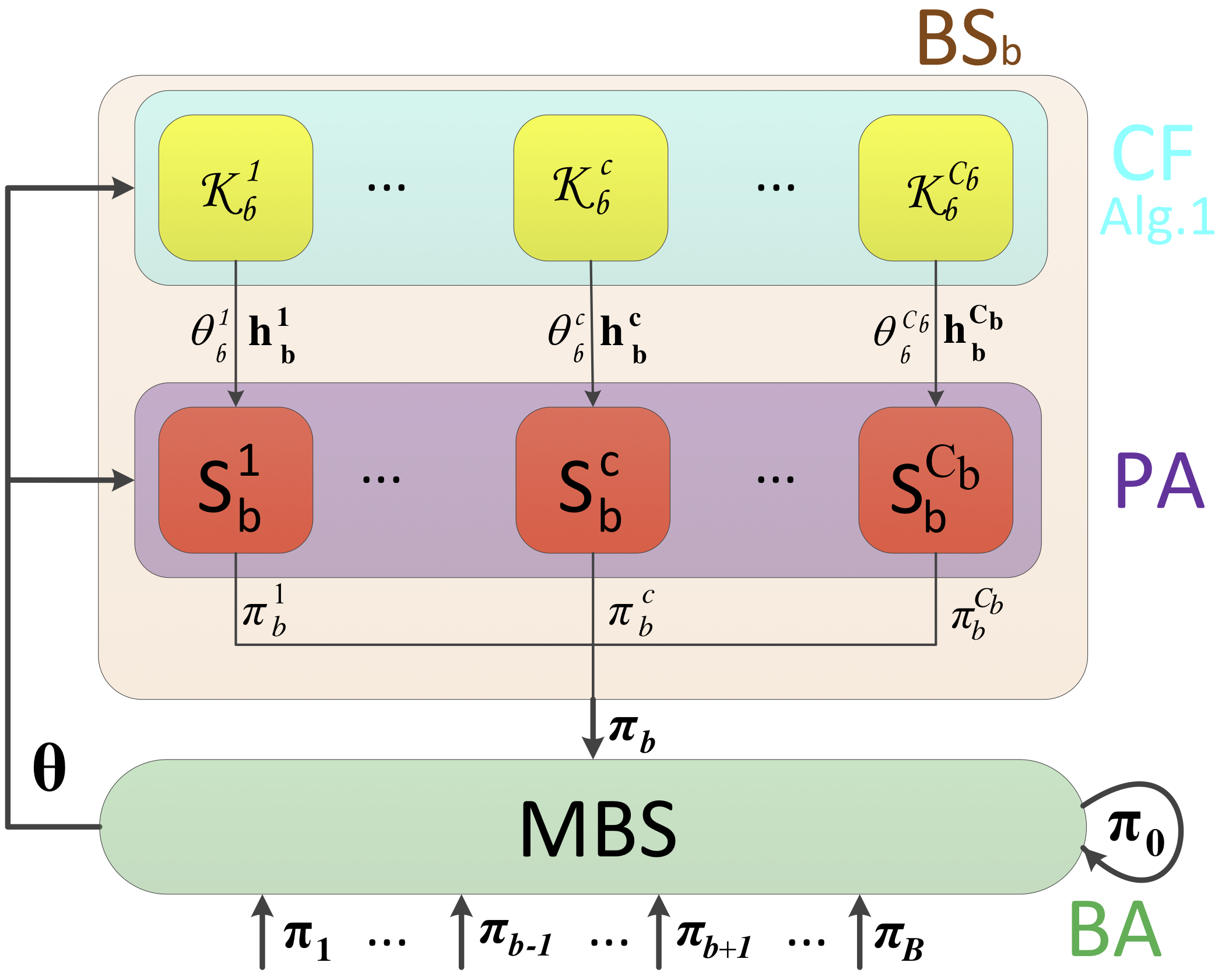}
\caption{Illustration of the proposed distributed clustering and resource allocation scheme [c.f. Algorithm \ref{alg:HD}]}
\label{fig:dist}
\end{figure}

On the other hand, resource allocation problem is further decomposed into slave and master problems which are responsible for power and bandwidth allocation, respectively. Given cluster members and bandwidths, each slave problem is accountable for obtaining an optimal power control policy for imperfect NOMA scheme subject to CSCs. Thereafter, achieved cluster utilities are shared by a central unit (preferably the MBS) that accordingly updates the cluster bandwidths for the next iteration, which is followed by another round of cluster formation and power allocation, so on and so forth. The details of the proposed  distributed solution methodology are addressed in the following sections.

\section{Design and Analysis of NOMA Clusters}
\label{sec:CF}
NOMA clustering involves two main design tasks; 1) determining the number of clusters and their size and 2) assigning UEs to the clusters. Accordingly, this section first analyzes the cluster size based on random matrix theory and derives the largest feasible cluster size as a closed-form function of the FEF levels, CSCs, and cluster bandwidth. Then, we propose a weighted maximum matching based CF method by weighting edges with channel gain disparity of UEs.

\subsection{Fundamental Limits of NOMA Clusters}
\label{sec:CSanalysis}
Without loss of generality, let us consider a cluster of size $K$ and bandwidth allocation $\theta$, whose CSCs can be written in the matrix form as 
\begin{equation}
\label{eq:SINR-matrix}
\left(\vect{\mathrm{I}}-\vect{\mathrm{\Gamma}}(\theta) \vect{\mathrm{H}}\right)\vect{\mathrm{p}} \geq \vect{\bar{\mathrm{\Gamma}}}(\theta) \vect{\mathrm{\sigma}}, \text{ s.t. } \vect{\mathrm{p}> 0} ,
\end{equation}
where vectors are of size $1 \times K$, matrices are of size $K \times K$, $\vect{\mathrm{I}}$ is the identity matrix, $\vect{\bar{\mathrm{\Gamma}}}(\theta)=\diag(\bar{\mathrm{\Gamma}}_1(\theta), \ldots ,\bar{\mathrm{\Gamma}}_k(\theta),\ldots ,\bar{\mathrm{\Gamma}}_K(\theta))$ is the diagonal matrix of the composite SINR demands, $\vect{\mathrm{p}}$ is the column vector of the received powers, $\vect{\mathrm{\sigma}}$ is the column vector of the receiver noise, and $\vect{\mathrm{H}}$ is the interference channel gain matrix with entries
\begin{equation}
H_i^j=\begin{cases}1,& i < j\\
0,& i=j\\
\epsilon,& i>j 
 \end{cases},
\end{equation}
where cases correspond to uncancelled interference, self interference, and residual interference, respectively. Notice that $\vect{\mathrm{H}}$ has non-negative elements and is generally considered to be irreducible \cite{Andrews2005iterative}\footnote{ We assume that $\epsilon$ cannot be zero in practice. To evaluate the numerical results for $\epsilon \rightarrow 0$, we employ the smallest positive normalized floating-point number based on the IEEE Standard for floating-point arithmetic (IEEE 754), i.e., $\epsilon=2.2251\mathrm{e}{-308 }$ \cite{kahan1996ieee}.}. For a non-negative irreducible matrix, \textit{Perron-Frobenius theorem} teaches us that the maximum eigenvalue of $\vect{\mathrm{H}}$ is real-positive and eigenvector corresponding to the maximum eigenvalue is non-negative \cite{horn1990matrix}. Following from the facts known from the standard matrix theory, a necessary and sufficient condition for the existence of a feasible solution to \eqref{eq:SINR-matrix} requires the magnitude of the maximum eigenvalue of $\vect{\mathrm{F}}\bigtriangleq \vect{\mathrm{\bar{\Gamma}}}(\theta)\vect{\mathrm{H}}$ to be less than unity, i.e., $\lambda_F<1$ \cite{Andrews2005iterative}. Assuming the existence of a feasible solution, a Pareto-optimal solution to \eqref{eq:SINR-matrix} is then given by $\vect{\mathrm{p}}^* = \left(\vect{\mathrm{I}}-\vect{\mathrm{\Gamma}}(\theta) \vect{\mathrm{H}}\right)^{-1} \vect{\bar{\mathrm{\Gamma}}}(\theta) \vect{\mathrm{\sigma}}$ where any other feasible $\vect{\mathrm{p}}$ satisfying \eqref{eq:SINR-matrix} would require more power than $\vect{\mathrm{p}}^*$, i.e., $\vect{\mathrm{p}} \geq \vect{\mathrm{p}}^*$. From energy efficiency point of view, we stick with the minimum power consuming solution $\vect{\mathrm{p}}^*$. Based on these discussion, we introduce following lemmas for the largest feasible cluster size as a function of the FEF and CSCs. 

\begin{lemma}[Energy unconstrained cluster size] \label{lem:CS}
 For a cluster of energy unconstrained UEs, the largest feasible cluster size falls within the range of $K_{\min} (\epsilon, \theta)\leq K(\epsilon, \theta) \leq K_{\max}(\epsilon, \theta)$, i.e.,
\begin{equation}
\label{eq:lem1}
 \left \lceil \frac{\ln(\epsilon)}{ \ln \left(\frac{1+\epsilon \max_i (\bar{\Gamma}_{i}(\theta))}{1+\max_i (\bar{\Gamma}_{i}(\theta))} \right)} \right \rceil \leq K (\epsilon, \theta) \leq 
 \left \lfloor \frac{\ln(\epsilon)}{\ln \left(\frac{1+\epsilon \min_i (\bar{\Gamma}_{i}(\theta))}{1+\min_i (\bar{\Gamma}_{i}(\theta))} \right)} \right \rfloor. 
\end{equation}
Accordingly, $K^\star(\epsilon, \theta)=K_{\min}(\epsilon, \theta)$ is the largest feasible cluster size which is mainly determined by the user with the highest composite SINR demand. 
\end{lemma}
\begin{proof}
Please see Appendix \ref{app_lem_1}.
\end{proof}
\begin{lemma}[Cluster size for identical CSCs]
\label{lem:CSspec}
As a special case, $ \bar{\Gamma}_i(\theta)=\bar{\Gamma}(\theta), \forall i $, the range in \eqref{eq:lem1} tightens to an exact size of $K^*(\epsilon, \theta)= \left \lfloor \frac{\ln(\epsilon)}{\ln \left(\frac{1+\epsilon \bar{\Gamma}(\theta)}{1+\bar{\Gamma}(\theta)} \right)} \right \rfloor$ which corresponds to an achievable rate of $\bar{\Gamma}^*(K(\epsilon,\theta))=\frac{e^{\ln(\epsilon)/K^*(\epsilon, \theta)}-1}{\epsilon-e^{\ln(\epsilon)/K^*(\epsilon, \theta)}}$.
\end{lemma}
\begin{proof}
Please see Appendix \ref{app_lem_1}.
\end{proof}

\begin{lemma}[Energy constrained cluster size] \label{lem:CSconstrained}
 For a cluster of energy constrained UEs with $\bar{\Gamma}(\theta)=\max_i( \bar{\Gamma}_i(\theta))$, the largest feasible cluster size falls within the range of $K_{\min}(\epsilon, \theta) \leq K(\epsilon, \theta) \leq K_{\max}(\epsilon, \theta)$
where $K_{\min}(\epsilon, \theta)=\left \lfloor 1+ \frac{\ln \left( \frac{\epsilon(1+\bar{\Gamma}(\theta))}{\epsilon-1}\right)-\ln \left( \frac{\bar{\Gamma}(\theta) \sigma^2}{\bar{P}_u g_K}-\frac{1+\epsilon\bar{\Gamma}(\theta)}{1-\epsilon}\right) }{\ln \left( \frac{1+\epsilon\bar{\Gamma}(\theta)}{1+ \bar{\Gamma}(\theta)}\right) } \right \rfloor$ and $K_{\max}(\epsilon, \theta)=\left \lceil 1+ \frac{\ln \left( \frac{\bar{\Gamma}(\theta) \sigma^2}{\bar{P}_u g_1}+  \frac{\epsilon(1+\bar{\Gamma}(\theta))}{1-\epsilon}\right) -  \ln \left( \frac{1+\epsilon\bar{\Gamma}(\theta)}{1-\epsilon} \right)}{\ln \left( \frac{1+\epsilon\bar{\Gamma}(\theta)}{1+ \bar{\Gamma}(\theta)}\right) } \right \rceil.$ Accordingly, $K^\star(\epsilon, \theta)=K_{\min}(\epsilon, \theta)$ is the largest feasible cluster size which is mainly determined by the user with the lowest channel gain. 
 \end{lemma}
\begin{proof}
Please see Appendix \ref{app_lem_3}.
\end{proof}

Once can draw the inference from these lemmas that the largest feasible cluster size increases as $\bar{\Gamma}(\theta)$ and $\epsilon$ decreases, that is, NOMA can serve more users with low rates as the SIC efficiency improves. Notice that cluster size analyses in Lemma \ref{lem:CS} and Lemma \ref{lem:CSspec} are only valid for UEs with unlimited transmission power as $\vect{p^*}$ is a solution over the feasible set of  $\vect{\mathrm{p}>0}$. However,  Lemma \ref{lem:CSconstrained} accounts for  power constrained users, where channel gain of the lowest cluster member plays an important role.

\subsection{Cluster Formation Design}
\label{sec:CF_design}
Unlike the basic NOMA clusters of size two, one can reap the full benefits of high-spectral efficiency offered by NOMA if the large cluster size is considered. In addition to the enhanced spectral efficiency, increasing the cluster size also reduces the total power consumption of UEs within a BS, that magnifies the efficiency of energy spent per bit. Therefore, our clustering strategy is to exploit the largest feasible cluster size obtained in the previous section. This strategy is especially important to provide the massive connectivity required by the ever-increasing number of devices. When each cluster is allocated to a single RB, for example, basic NOMA clustering can accommodate at most $2 \Theta$ UEs at a time. Notice that employing large clusters is eminently suitable for UL-NOMA scheme since UEs do not compete for the BS transmit power as in the DL case. However, a larger cluster size requires more computational power to compute optimum power levels and yields a longer decoding delay as the SIC latency linearly increase with the cluster size \cite{Andrews2005SIC}. Fortunately, BSs can be equipped with high computational power with more sophisticated receivers with desirable FEF and latency specifications. 

Based on the analytical findings in Section \ref{sec:SIC}, our strategy on assigning UEs to clusters focuses on maximizing the channel gain disparity among the cluster members to enhance the achievable NOMA gain. Accordingly, algorithmic implementation of these strategies  is given in Algorithm \ref{alg:CF} where the first line uses Lemma \ref{lem:CSconstrained} to determine the largest cluster size that is allowable by each UEs within BS$_b$, i.e., $\kappa_i=\min \left(\bar{K}_b, K_{\min} \right), i \in \mathcal{U}_b$, where $\bar{K}_b$ is a design parameter in order to prevent unnecessarily high delay and computational power due to the large cluster sizes. Accordingly, $\left(\kappa_i,\forall i \right)$ values are sorted in the ascending order to generate the vector $\vect{\kappa}=\left[ \kappa_i \vert \: i \in \mathcal{U}_b, \kappa_j > \kappa_{j+1}, 1 \leq j \leq U_b-1 \right]$. Starting from the larges
t cluster size, line 2 increases the number of clusters in BS$_b$ until total number of cluster sizes are no less than $\mathcal{U}_b$, i.e., 
\begin{equation}
\label{eq:Cb}
C_b = \argmin_\mathrm{I} \left \{\mathrm{I}\left \vert  \sum_{i=1}^{\mathrm{I}} \kappa_{i} \geq U_b \right .  \right \}
\end{equation} 
which provides the least number of clusters and thus the largest size of clusters. 

\begin{algorithm}
\footnotesize
 \caption{\small Cluster Formation for BS$_b, \: \forall b$.}
  \label{alg:CF}
\begin{algorithmic}[1]
 \renewcommand{\algorithmicrequire}{\textbf{Input:}}
 \renewcommand{\algorithmicensure}{\textbf{Output:}}
\REQUIRE$\vect{\bar{\Gamma}}(\theta)$, $\vect{h}$
    \STATE      $\vect{\kappa} \gets $ Sort UEs in ascending order as per Lemma \ref{lem:CSconstrained}.
	\STATE     $C_b \gets$ Determine the least number of clusters as per \eqref{eq:Cb}.
    \STATE     $\mathcal{K}_b^i \gets \vect{\kappa}[i]$, $1\leq i \leq C_b$ Predetermination of first cluster members. 
    \STATE     $\mathcal{U}_b' \gets$ Update the remaining set of UEs.
    	\FOR{$s=1: \left( \left \lceil \frac{U_b}{C_b} \right \rceil-1 \right)$}
		\STATE $\mathcal{E}_i^j(s) \gets $ \eqref{eq:Eij} Calculate edge weights 
		\STATE $\mathcal{K}_b^c \gets \underset{\vect{x}}{\min}  \sum_{i,j}  \mathcal{E}_i^j(s) x_i^j \text{ (s.t.) }  \sum_{i} x_i^j \leq 1 ,\sum_{j} x_i^j =1, x_i^j \in \{ 0,1\}, i \in [1,C_b], j \in [1, |\mathcal{U}_b'|]$
		\STATE     $\mathcal{U}_b' \gets$ Update the remaining set of UEs.
	\ENDFOR
\hspace{-16pt}\RETURN  $\mathcal{K}_b^c$, $\forall c$.
 \end{algorithmic}
 \end{algorithm}
 
\begin{figure}[!t]
\centering
  \includegraphics[width=0.45\textwidth]{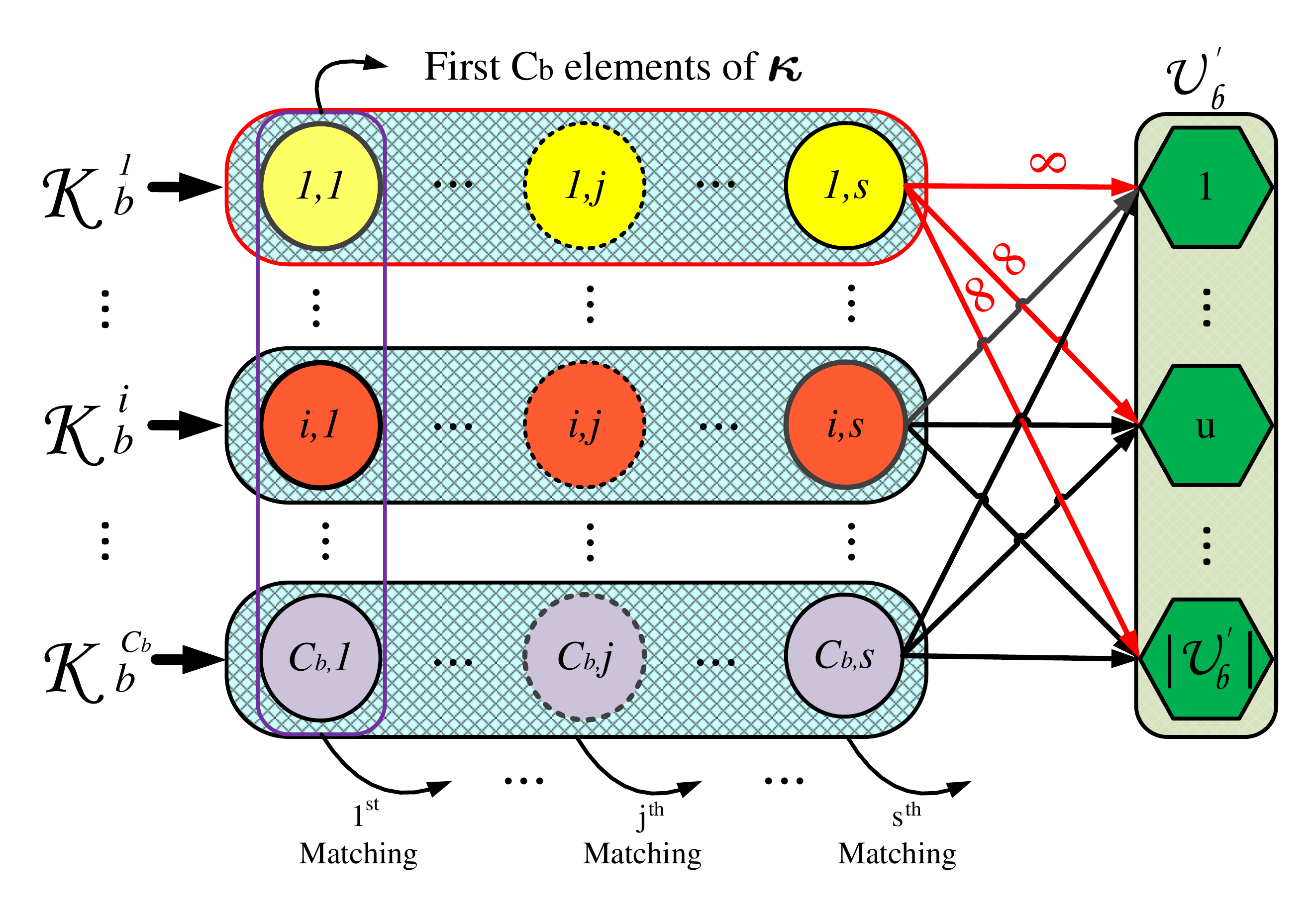}
      \caption{Illustration of the proposed CF method for sum-rate Maximization.}
        \label{fig:CF}
\end{figure}

As illustrated in Fig. \ref{fig:CF}, line 3 of Algorithm \ref{alg:CF} predetermines $i^{th}$, $1 \leq i \leq C_b$, element of $\vect{\kappa}$ as the first member of $i^{th}$ cluster $\mathcal{K}_b^i$ which has the size of $ \vect{\kappa}_{i}$. Thereafter, the while loop between lines 6 and 10 iteratively matches clusters with the remaining set of UEs, i.e., $\mathcal{U}_b'=\mathcal{U}-\bigcup_{i=1}^{C_b}\mathcal{K}_b^i$. In line 7 of iteration $s$, matching weight from $i^{th}$ cluster to $j^{th}$ element of $\mathcal{U}_b'$ is calculated as
\begin{equation}
\label{eq:Eij}
\mathcal{E}_i^j (s)=\begin{cases} \frac{ \inf \left(\widehat{\mathcal{H}}_i^j\right)}{h_j^b} + \frac{h_j^b}{\sup \left(\widecheck{\mathcal{H}}_i^j\right)}, i \in \mathcal{K}_b, j \in \mathcal{U}_b'&, \text{if } \vect{\kappa}_i > s \\
\infty &, \text{otherwise} \end{cases}
\end{equation}
 where $\widehat{\mathcal{H}}_i^j=\{h_k^b \vert h_k^b \geq h_j^b, k \in \mathcal{K}_b^i, j \in \mathcal{U}'  \}$ and $\widecheck{\mathcal{H}}_i^j=\{h_k^b \vert h_k^b \leq h_j^b, k \in \mathcal{K}_b^i, j \in \mathcal{U}'  \}$ are set of cluster members with higher and lower channel gain than the UE$_j \in \mathcal{U}_b'$, respectively\footnote{Notice that either $\widehat{\mathcal{H}}_i^j = \emptyset$ or $\widecheck{\mathcal{H}}_i^j=\emptyset$ happens for $s=1$. Since the order theory of the real analysis tells us that $\inf \left(\emptyset \right)= \infty$ and $\inf \left(\emptyset \right)=-\infty$, we ignore the first (second) term of \eqref{eq:Eij} if $\widehat{\mathcal{H}}_i^j = \emptyset$ ($\widecheck{\mathcal{H}}_i^j=\emptyset$) occurs.}. The first and second term of \eqref{eq:Eij} favors for new members who give a desirable channel gain disparity between UE$_j \in \mathcal{U}_b'$ and current cluster members with high and low channel gains. Notice that clusters that reached to its maximum affordable size are taken out of consideration by setting their edge weights to infinity. Line 8 executes maximum weighted bipartite matching,  $\mathcal{K}_b^c \gets \underset{\vect{x}}{\min}  \sum_{i,j}  \mathcal{E}_i^j(s) x_i^j \text{ (s.t.) }  \sum_{i} x_i^j \leq 1 ,\sum_{j} x_i^j =1, x_i^j \in \{ 0,1\}, i \in [1,C_b], j \in [1, |\mathcal{U}_b'|]$, which is in the form of a rectangular assignment problem and can be solved in cubic order.
 Algorithm \ref{alg:CF} is run by each BS independent from others and its overall complexity can be given as $\mathrm{O} \left( U_b \log U_b + \sum_{s=1}^{ \left \lceil \frac{U_b}{C_b} \right \rceil-1} \left(U_b-sC_b\right)^3 \right)$ where the first and second terms are due to sorting and matching operations in lines 1 and 8, respectively. Since the second term is more dominant, the proposed clustering solutions has cubic time complexity. On the other hand, exhaustively checking all clustering sizes and corresponding user combination $\sum_{k=2}^{Ub} \binom{U_b}{k} \approx 2^{U_b}$ which yields an exponential time complexity.

\section{$\alpha$-Fair Resource Allocation}
\label{sec:PBA}
In this section, we handle the RA problem by decoupling it into two stages: In the former, a slave problem is defined for each cluster such that optimal power allocations are obtained in closed-form for given cluster formations and bandwidths. In the latter, each slave problem reports its obtained utility which is exploited by a master problem to update cluster bandwidths.
\subsection{Slave Problems: Power Allocation}
\label{sec:slave}
Power allocation problem of clusters can be formulated as in \eqref{eq:Slave} where we omit BS and cluster indices for the sake of simplicity without loss of generality. 
\begin{equation*}
\label{eq:Slave}
\hspace*{0pt}
 \begin{aligned}
 & \hspace*{0pt} \vect{\mathrm{S}}: \underset{\vect{\omega}}{\max}
& & \hspace*{3 pt} \sum_{i=1}^{K} \pi_i\left( \vect{\omega} \right)  \\
& \hspace*{0pt} \mbox{$\mathrm{S_i^1}$:}\hspace*{8pt} \text{s.t.}
&& \hspace*{-4 pt} \omega_i \leq 1, \textbf{ } \hspace{0 pt} \forall i\\ 
&
  \hspace*{0 pt}\mbox{$\mathrm{S_i^2}$: } & & \hspace*{-4 pt}0 \leq  \omega_{i} h_i -   
   \bar{\Gamma}_i \left( \epsilon  \sum_{j =1}^{i-1}  \omega_{j} h_j+ \sum_{k =i+1}^{K} \omega_{k} h_k +\varrho \right),\forall i
\end{aligned}
\end{equation*}
which can be locally solved by each cluster member for given cluster bandwidth and channel gains of other cluster members. In order to derive the closed-form expressions for optimal power allocations, we first apply dual decomposition method to the slave problems. Accordingly, Lagrangian function of $\vect{\mathrm{S}}$ is given in (\ref{eq:Lag}) where $\lambda_i$ and $\mu_i, \forall i,$ are Lagrange multipliers. Taking derivatives of Lagrangian function with respect to  $\omega_i$, $\lambda_i$ and $\mu_i$, Karush-Kuhn-Tucker (KKT) conditions can be obtained as in (\ref{eq:Lag_derv_omg})-(\ref{eq:Lag_derv_lam}). 

KKT conditions are first-order necessary conditions for a nonlinear programming solution to be optimal, which is still subject to satisfaction of some regularity conditions. In particular, if all equality and inequality constraints are affine functions, i.e., linearity constraint qualification is held, no other regularity condition is needed. This is indeed the case for $\vect{\mathrm{S}}$ as all constraints are affine functions of $\vect{\omega}$. In the slave problem, there exists a total of $2K$ Lagrange multipliers that can be categorized into two subsets $\mathcal{S}_1=\{\lambda_i\vert 1 \leq i \leq K\}$ and $\mathcal{S}_2=\{\mu_i \vert 1 \leq i \leq K\}$. Therefore, each slave problem requires the KKT condition verification of $2^{2K}$ Lagrange multiplier combinations. Even though this  is computationally impractical, we fortunately need to check only $2^{K}$ combinations \cite{Ali2016,chong2013introduction} for the following reasons: Notice that each UE would transmit at the maximum transmission power in case of no interference, i.e, OMA. However, optimal power levels of NOMA can either be determined by CSCs or maximum transmission power according to SINR requirements and achievable capacity of UEs. That is, UE$_i$ can be active either on maximum transmission power or CSCs at the optimal point. Hence, we  need to consider the following solution set $\mathcal{S}=\{ \lambda_i \text{ or } \mu_i   \vert  i \in[1, K]\}$ in order to obtain a closed-form solution. For a basic NOMA cluster, combinations of solution set can be given as $\{ \lambda_1, \lambda_2 \}$, $\{ \lambda_1, \mu_2 \}$, $\{ \mu_1, \lambda_2 \}$, and $\{ \mu_1, \mu_2 \}$.  Furthermore, $\mathcal{S}_{\lambda}=\mathcal{S}-\mathcal{S}_2$ and $\mathcal{S}_{\mu}=\mathcal{S}-\mathcal{S}_1$ represents the subset of the solution set $\mathcal{S}$ which define cluster members active at $\lambda$ and $\mu$, respectively. Finally, $\mathcal{I}_{\lambda}$ and $\mathcal{I}_{\mu}$ denotes the index set of $\mathcal{S}_{\lambda}$ and $\mathcal{S}_{\mu}$, respectively. For example, for the solution set of $\mathcal{S}=\{\mu_{1},\lambda_{2},\lambda_{3},\mu_{4},\lambda_{5},\mu_{6}\}$, we have $\mathcal{S}_{\lambda}=\{\lambda_{2},\lambda_{3},\lambda_{5}\}$, $\mathcal{S}_{\mu}=\{\mu_{1},\mu_{4},\mu_{6}\}$, $\mathcal{I}_{\lambda}=\{2,3,5\}$ and $\mathcal{I}_{\mu}=\{1,4,6\}$.

   \begin{figure}[!t]
    \centering
        \begin{subfigure}[b]{0.24 \textwidth}
\includegraphics[width=\columnwidth]{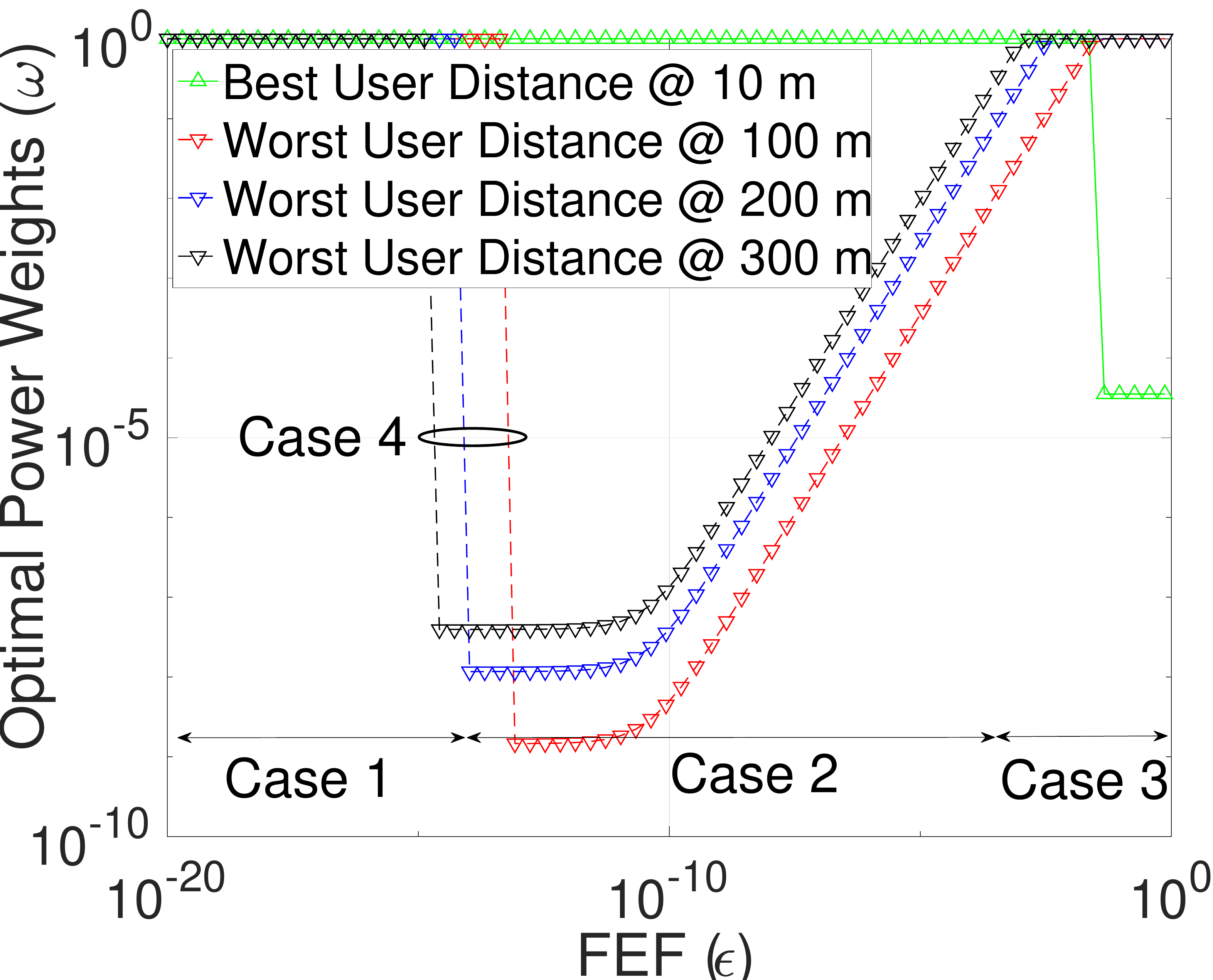}  
        \caption{}
 \label{fig:gain}
    \end{subfigure}
        \begin{subfigure}[b]{0.24\textwidth}
\includegraphics[width=\columnwidth]{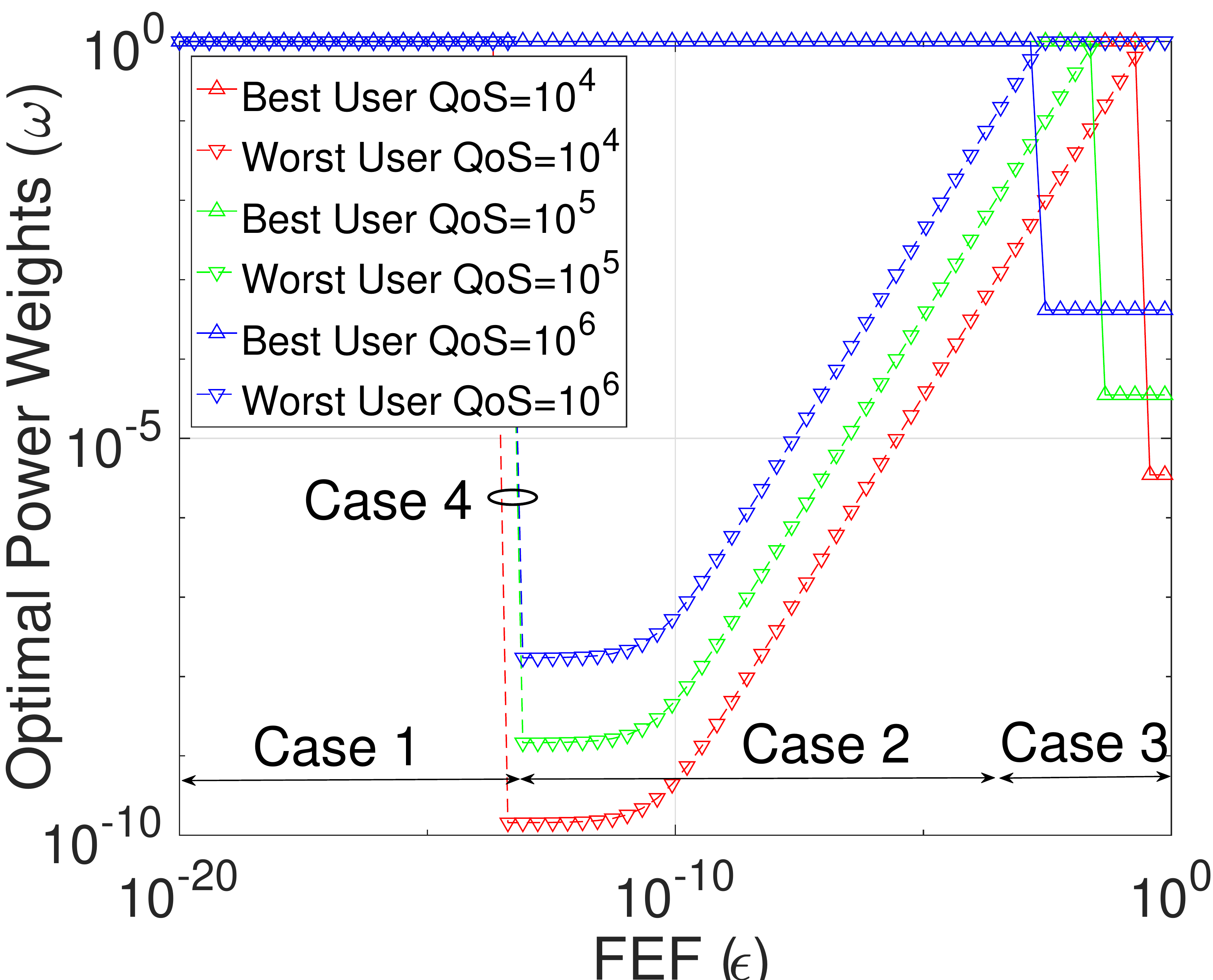}  
        \caption{ }
 \label{fig:qos}
    \end{subfigure}
   \caption{Optimal power allocations of a basic cluster vs. FEF levels; a) different channel gain disparity and b) QoS constraint scenarios.}
        \label{fig:cases}
\end{figure}
\begin{figure*}
\small
\begin{align}
 & L(\vect{\omega}, \vect{\lambda},\vect{\mu})=\frac{1}{1-\alpha}\sum_{i=1}^{K} \left( R_i^{1-\alpha}\left(\vect{\omega}\right)-1 \right)+\sum\limits_{i=1}^{K}\lambda_{i}(1-\omega_i) +\sum\limits_{i=1}^{K}\mu_{i} \left(\omega_{i} g_i- \bar{\Gamma} _i\left( \epsilon _i \sum\limits_{j = 1}^{i - 1} {\omega_j g_j}- \sum\limits_{k=i+1}^{K}\omega_{k} g_k- \varrho\right) \right) \label{eq:Lag} \\
\nonumber   & \frac{\partial \mathcal{L}}{\partial \omega_i^{\star}}  = W\theta \left\{ \frac{ R_i^{-\alpha}\left(\vect{\omega}\right)}{\sum _{k=1}^{i-1} \epsilon  h_k \omega_k+\sum _{l=i}^K \omega_l h_l+\varrho } \right. -\sum _{j=1}^{i-1} \frac{   \omega_j h_j R_j^{-\alpha}\left(\vect{\omega}\right) }{\left(\sum _{k=1}^{j-1} \epsilon  h_k \omega_k+\sum _{l=j}^K \omega_l h_l+\varrho\right)\left(\sum _{k=1}^{j-1} \epsilon  h_k \omega_k+\sum _{l=j}^K \omega_l h_l+\varrho\right) }\\
\nonumber &   \hspace{25pt}  \left. -\sum _{j=i+1}^{K} \frac{ \epsilon \omega_j h_j R_j^{-\alpha}\left(\vect{\omega}\right) }{\left(\sum _{k=1}^{j-1} \epsilon  h_k \omega_k+\sum _{l=j}^K \omega_l h_l+\varrho\right)\left(\sum _{k=1}^{j-1} \epsilon  h_k \omega_k+\sum _{l=j}^K \omega_l h_l+\varrho\right) } \right\} \\
&  \hspace{25pt}  -\lambda_{i}+\left(1-\Gamma _i\right) \mu _i-\sum _{j=1}^{i-1} \Gamma _j \mu _j-\epsilon  \sum _{j=i+1}^{\text{K}} \Gamma _j \mu _j \geq 0, \forall i,   \label{eq:Lag_derv_omg}
\\
 \label{eq:Lag_derv_lam}
  & \frac{\partial \mathcal{L}}{\partial \lambda_i^\star} = 1-\omega_i \geq 0, \text{ if } \lambda_i^\star  \geq 0,  \frac{\partial \mathcal{L}}{\partial \mu_{i}^{\star}} =  \omega_i- \left(\sum_{j=1}^{i-1} \omega_j +\epsilon_i \sum_{k=i+1}^{K_c^r} \omega_k + \rho_i \right)(q_i-1) \geq 0,\text{ if } \mu_{i}^{\star} \geq 0.
\end{align}
\hrule
\vspace{1pt}
\hrule
\end{figure*}
\renewcommand{\arraystretch}{1.5}
\begin{table*}[htbp!]
\tiny
\centering
\caption{Necessary conditions and closed-form power allocations (Please see Appendix \ref{app_lem_4} for the proof).}
\label{tab:pwr_cond}
\begin{tabular}{|l|l|l|}
\hline
\multicolumn{1}{|l|}{\small \textbf{K}} & \centering \textbf{\small Necessary Conditions} & \multicolumn{1}{c|}{\textbf{\small Power Allocations}} \\ \hline
\multirow{4}{*}{\vspace*{-65pt} \small \textbf{2}}      &  \hspace{-27pt} \parbox{6cm}{\begin{flalign*} & \mathcal{S}=\{\lambda_1, \lambda_2\}: \mathrm{S}_l^2 , \mu_l>0 ,l=1,2.
\end{flalign*}}                     &    \parbox{7cm}{ $\omega_1=\omega_2=1$                                        } \\ \cline{2-3} 
                        &     \parbox{6cm}{\begin{flalign*}& \mathcal{S}=\{\lambda_1, \mu_2\}: \mathrm{S_k^1} , \lambda_k >0, k=2. \: \vert \: \mathrm{S}_l^2, \mu_l>0, l=1.
\end{flalign*}}                  &                \parbox{6cm}{ $\omega_1=1  \: \left \vert \: \omega_2= \frac{\bar{\Gamma}_2 \left(h_1 \epsilon +\varrho \right)}{h_2} \right.$}                                    \\ \cline{2-3} 
                        &    \parbox{6cm}{\begin{align*}& \mathcal{S}=\{\mu_1, \lambda_2\}: \mathrm{S_k^1} , \lambda_k >0, k=1. \: \vert \:  \mathrm{S}_l^2, \mu_l>0, l=2
\end{align*}}                             &           \parbox{6cm}{ $\left. \omega_1= \frac{\bar{\Gamma}_1 \left(h_2+\varrho \right)}{h_1} \: \right \vert \: \omega_2=1$}                                                 \\ \cline{2-3} 
                        &        \hspace{-27pt} \parbox{6cm}{\begin{align*}& \mathcal{S}=\{\mu_1, \mu_2\} :  \mathrm{S}_k^1 , \lambda_k>0, k=1,2.
\end{align*}}                         &        \parbox{6cm}{ $\omega_1= \frac{\varrho }{\frac{h_1}{\bar{\Gamma}_1}-a_1 h_2} \: \left \vert \: \omega_2= \frac{a_1 \varrho}{\frac{h_1}{\bar{\Gamma}_1}-a_1 h_2} \right. $}                                                                                \\ \hline
\multirow{8}{*}{\vspace*{-150pt} \small \textbf{3}}      &    \hspace{-17pt}   \parbox{6cm}{\begin{align*}& \mathcal{S}=\{\lambda_1, \lambda_2, \lambda_3\} : \mathrm{S}_l^2 , \mu_l>0 ,l=1,2,3.
\end{align*}}                       &        \parbox{6cm}{   $\omega_1=\omega_2=\omega_3=1$     }       \\ \cline{2-3} 
                        &    \parbox{6cm}{\begin{align*}& \mathcal{S}=\{\lambda_1, \lambda_2, \mu_3\} : \mathrm{S_k^1} , \lambda_k >0, k=3. \: \vert \:  \mathrm{S}_l^2, \mu_l>0, l=1,2.
\end{align*}}                         &          $\omega_1=\omega_2=1 \: \left \vert \: \omega_3= \frac{\bar{\Gamma}_3 \left(h_1 \epsilon +h_2 \epsilon +\varrho \right)}{h_3}\right. $                                         \\ \cline{2-3} 
                        &       \parbox{6cm}{\begin{align*}& \mathcal{S}=\{\lambda_1, \mu_2, \lambda_3\} : \mathrm{S_k^1} , \lambda_k >0, k=2. \: \vert \:  \mathrm{S}_l^2, \mu_l>0, l=1,3.
\end{align*}}                   &              $\omega_1=\omega_3=1 \: \left \vert \: \omega_2= \frac{\bar{\Gamma}_2 \left(h_1 \epsilon +h_3+\varrho \right)}{h_2} \right.$                                                 \\ \cline{2-3} 
                        &        \parbox{6cm}{\begin{align*}& \mathcal{S}=\{\lambda_1, \mu_2, \mu_3\} : \mathrm{S_k^1} , \lambda_k >0, k=2,3. \: \vert \:  \mathrm{S}_l^2, \mu_l>0, l=1.
\end{align*}}                          &         $\omega_1=1 \: \left \vert \: \omega_2=  \frac{h_1 \epsilon +\varrho }{\frac{h_2}{\bar{\Gamma}_2}-a_1 h_3} \right. \: \left \vert \: \omega_3= \frac{a_1 \left(h_1 \epsilon +\varrho \right)}{\frac{h_2}{\bar{\Gamma}_2}-a_1 h_3} \right.$                                                                                      \\ \cline{2-3} 
                        &            \parbox{6cm}{\begin{align*}& \mathcal{S}=\{\mu_1, \lambda_2, \lambda_3\} : \mathrm{S_k^1} , \lambda_k >0, k=1. \: \vert \:  \mathrm{S}_l^2, \mu_l>0, l=2,3.
\end{align*}}                        &             $\left. \omega_1=  \frac{\bar{\Gamma}_1 \left(h_2+h_3+\varrho \right)}{h_1} \: \right \vert \: \omega_2=\omega_3=1 $                                                           \\ \cline{2-3} 
                        &          \parbox{6cm}{\begin{align*}& \mathcal{S}=\{\mu_1, \lambda_2, \mu_3\} : \mathrm{S_k^1} , \lambda_k >0, k=1,3. \: \vert \:  \mathrm{S}_l^2, \mu_l>0, l=2.
\end{align*}}                             &                          $\left. \omega_1=  \frac{c_2 h_3+h_2+\varrho }{\frac{h_1}{\bar{\Gamma}_1}-a_1 h_3} \: \right \vert \: \omega_2=1 \left \vert \: \omega_3= \frac{a_1\left(c_2 h_3+h_2+\varrho \right)}{\frac{h_1}{\Gamma _1}-a_1 h_3}+c_2\right. $                                                                               \\ \cline{2-3} 
                        &        \parbox{6cm}{\begin{align*}& \mathcal{S}=\{\mu_1, \mu_2,\lambda_3\} : \mathrm{S_k^1} , \lambda_k >0, k=1,2. \: \vert \:  \mathrm{S}_l^2, \mu_l>0, l=3.
\end{align*}}                              &                       $\left. \omega_1=  \frac{h_3+\varrho }{\frac{h_1}{\bar{\Gamma}_1}-a_1 h_2} \: \right \vert \: \left.  \omega_2= \frac{a_1 \left(h_3+\varrho \right)}{\frac{h_1}{\bar{\Gamma}_1}-a_1 h_2} \: \right\vert \: \omega_3=1 $                                                                     \\ \cline{2-3} 
                        &      \hspace{-17pt}   \parbox{6cm}{\begin{align*}& \mathcal{S}=\{\mu_1, \mu_2, \mu_3\} : \mathrm{S_k^1} , \lambda_k>0 ,k=1,2,3.
\end{align*}}                             &                 \parbox{6cm}{\begin{align*}              & \left. \omega_1= \frac{\varrho }{\frac{h_1}{\bar{\Gamma}_1}-a_1 h_2-a_1 a_2 h_3} \: \right \vert \: \omega_2=\frac{a_1 \varrho }{\frac{h_1}{\bar{\Gamma}_1}-a_1 h_2-a_1 a_2 h_3} \\ 
& \: \omega_3=\frac{a_1 a_2 \varrho }{\frac{h_1}{\bar{\Gamma}_1}-a_1 h_2-a_1 a_2 h_3}                                            \end{align*}}                          \\ 
\hline
\end{tabular}
\vspace{8pt}
\hrule
\vspace{1pt}
\hrule
\end{table*}

\begin{figure*}
\begin{lemma}
\label{lem:cfp}
Given that necessary conditions are satisfied, closed-form power allocations of cluster members is given as 
\begin{equation}
\label{eq:cfp}
\omega_{i}=\begin{cases}
1 &, \text{for } \forall i \in \mathcal{I}_\lambda,\\
\left(\prod_{\substack{1\leq j<i \\ j\in \mathcal{I}_{\mu}}}a_j\right)\omega_{\mathrm{mind}}+\sum_{\substack{1\leq j<i  \\ j \in \mathcal{I}_\mu}}\left(\prod_{\substack{j<k<i \\ k\in \mathcal{I}_{\mu}} } a_k \right)b_j&, \text{for } \forall i \in \mathcal{I}_\mu, i > \mathrm{mind}. \\
\end{cases}, 
\end{equation}
where $\omega_{\mathrm{mind}}=\frac{c_{\mathrm{mind}}+\sum_{\substack{\mathrm{mind}<k \leq K \\k\in \mathcal{I}_{\mu}}}h_{k}\left(\sum_{\substack{1 \leq j <k \\j\in \mathcal{I}_{\mu}}}\left(\prod_{\substack{ l>j \wedge j<k \\l\in \mathcal{I}_{\mu} }} a_l \right)c_j \right)}{\frac{h_{\mathrm{mind}}}{\bar{\Gamma}_{\mathrm{mind}}}-\sum_{\substack{\mathrm{mind}< k \leq K \\k\in \mathcal{I}_{\mu}} }h_k \left(\prod_{\substack{1 \leq j < k \\j\in I_{\mu}}}a_j \right)   }$,  $\mathrm{mind}\triangleq \argmin(\mathcal{I}_\mu)$ is the minimum index of UEs within $\mathcal{S}_\mu$, 
$c_{\mathrm{mind}}=\epsilon \sum_{\substack{1 \leq j < \mathrm{mind} \\j\in \mathcal{I}_{\lambda}}}h_{j}+ \sum_{\substack{\mathrm{mind} < j \leq K  \\j\in \mathcal{I}_{\lambda}}}h_{j} + \sum_{\substack{\mathrm{mind} < j \leq K  \\j\in \mathcal{I}_{\mu}}}\omega_j h_{j} $, $a_i=\frac{h_{\mathrm{max_i}}\left(\epsilon+\frac{1}{\bar{\Gamma}_{\mathrm{max_i}}} \right)}{h_{i}\left(1+\frac{1}{\Gamma_{i}} \right)}$, $b_i=\frac{\left(\epsilon-1\right)}{h_{i}\left(1+\frac{1}{\bar{\Gamma}_{i}}\right)}\sum_{\substack{\mathrm{max_i} <j<i \\ j\in \mathcal{I}_{\lambda} }}h_{j}$, and $\mathrm{max_i}=\argmax\{m \vert m \in \mathcal{I}_\mu , m<i\}$ is the maximum index of $\mathcal{S}_\mu$ among the indices less than $i$.
\end{lemma}
\begin{proof}
Please see Appendix \ref{app_lem_4}.
\end{proof}
\hrule
\vspace{1pt}
\hrule
\end{figure*}

For example, let us consider $\mathcal{S}=\{ \lambda_1,  \mu_2, \lambda_3, \mu_4 \}$, then the power allocations can be derived from active primal constraints, i.e., $\{ \mathrm{S}_1^1, \mathrm{S}_2^2, \mathrm{S}_3^1, \mathrm{S}_4^2 \}$, which also requires the satisfaction of corresponding primal KKT conditions, i.e., $\{ \mathrm{S}_1^2, \mathrm{S}_2^1, \mathrm{S}_3^2, \mathrm{S}_4^1 \}$.  That is, active primal constraints form the KKT conditions while inactive constraints are used for calculating the corresponding power allocations. Accordingly, we tabulate power allocations and corresponding KKT conditions for cluster sizes 2 and 3 in Table \ref{tab:pwr_cond} where the first column indicates the cluster size $K$, the second column presents $2^K$ solution set cases and corresponding necessary conditions, and finally the last row provides the closed-form optimal power allocations. Excluding the first case, both power allocations and necessary conditions are functions of three parameters; $\epsilon$, $\bar{\Gamma}$, and channel gain disparity. Based on these parameters, while some cases can be infeasible due to the violation of the necessary conditions, there might me multiple cases which satisfy the constraint with different performance.

Before we explain how the optimal case is determined, we consider an exemplary basic NOMA cluster size of two in order to have a deeper insight into how the power allocation strategy changes with different parameters which have a direct impact on the constraints, i.e., necessary conditions. The optimal power weights versus different FEF levels under various channel gain disparity and QoS constraints are shown in Fig. \ref{fig:gain} and Fig. \ref{fig:qos}, respectively. As it is already tabulated in Table 1, there exist four cases: In case 1, both UEs transmit at maximum power. In case 2, the worst UE is active at the QoS constraint whereas the best user keeps transmitting at the maximum power, which is in the opposite direction for case 3. In the last case, both UEs have power levels that exactly and barely satisfy their QoS demands. As it is obvious in Fig. \ref{fig:cases}, power levels and case regions vary with FEF levels, channel gain disparity, and QoS demands. While we observe case 1 and case 3 in very low and very high FEF levels, respectively, intermediate FEF levels operate on case 2. On the other hand, case 4 is observed during the interval where optimal case is in transition from case 1 to case 2. Notice that channel gain disparity and QoS constraints have significant impacts on both optimal power levels and the points where cases start and end.

At this point, let us explain how Table \ref{tab:pwr_cond} can be used to decide on the optimal case: First, power allocations of each case are computed by expressions on the rightmost column, and then substituted into the corresponding constraints in the central column to verify if the corresponding KKT conditions are satisfied. Thereafter, optimal power allocations are determined by the case which gives the highest objective value among the cases who satisfy the KKT conditions. Therefore, the worst case complexity  is given as $\mathrm{O}(2^{K}+K\log K)$ where the first term is the cost of calculating and checking $2^K$ cases and the second term is the cost of sorting and selecting the best case. Since the complexity of the first term dominates that of the second, overall complexity can be approximated by $\mathrm{O}(2^{K})$. Generalizing Table \ref{tab:pwr_cond}, Lemma \ref{lem:cfp} provides the closed-form expression for optimal power allocations for an imperfect NOMA cluster of size $K$.

\subsection{Master Problem: Bandwidth Allocation}
\label{sec:master}
Following the optimal power allocation of the slave problems, BSs report the achieved SINR levels of cluster members to the MBS which then updates the bandwidth allocations as follows
\begin{equation*}
\label{eq:Master}
\hspace*{0pt}
 \begin{aligned}
 & \hspace*{0pt} \vect{\mathrm{M}}: \underset{\vect{\theta}}{\max}
& & \hspace*{3 pt} \frac{1}{1-\alpha}\sum_{\forall (b,c,i)}\left[ \left( \theta_b^c \vartheta_{b,c}^i\right)^{1-\alpha}-1  \right]\\
& \hspace*{0pt} \mbox{$\mathrm{M_i^1}$:}\hspace*{12pt} \text{s.t.}
&& \sum_{b,c} \theta_{b,c} \leq \Theta, \textbf{ } \hspace{0 pt} \forall i\\ 
&
  \hspace*{0 pt}\mbox{$\mathrm{M_i^2}$: } & & \bar{R}_i \leq \theta_b^c \vartheta_{b,c}^i, \textbf{ }  \hspace{0pt}\forall i \in \mathcal{K}_b^c, \forall b, \forall c.
\end{aligned}
\end{equation*}
where $\vartheta_{b,c}^i  \triangleq  W \log_2(1+\Gamma_{b,c}^i)$ is the achieved utility of clusters and given by the slave problems. 
An effective method of solving this problem is unintegerizing the integer valued optimization variable $\theta_b^c$. In this manner, $\vect{\mathrm{M}}$ reduces to a convex optimization problem and fractional part of the optimal bandwidth allocations can be handled by RB scheduling mechanisms. 

\begin{algorithm}
 \caption{\small Distributed $\alpha$-Fair Resource Allocation}
  \label{alg:HD}
\begin{algorithmic}[1]
\small
 \renewcommand{\algorithmicrequire}{\textbf{Input:}}
 \renewcommand{\algorithmicensure}{\textbf{Output:}}
\REQUIRE Channel gains
 \STATE $t \leftarrow 0$
     \STATE $\vect{\theta(k)} \leftarrow$   Initialize the bandwidth allocations, $\forall b, c$. 
   \WHILE {$t \in \mathcal{T}$}
     \STATE $\vect{\delta_b} \leftarrow$ BS$_b$ forms its clusters based on Algorithm \ref{alg:CF}.
  \STATE $\vect{\omega_b^c(t)} \leftarrow$ Check power levels \& conditions. \STATE $\vect{\omega_{b,c}^{\star} (t)} \leftarrow$ Select the best cases for optimal power allocation. 
       \STATE UE$_i \gets \vect{\omega_{b,c}^i}$; UE$_i$ receives its power level from BS$_b, \forall i \in \mathcal{U}_b$.
\STATE BS$_0 \gets \vect{\vartheta_{b}^c}$; The MBS receives the utilities from BS$_b, \forall b$ 
\STATE BS$_b \gets \vect{\theta(t+1)}$; The MBS updates and disseminates bandwidths to BS$_b, \forall b$.
\STATE  $t \leftarrow t+1$
\ENDWHILE
\hspace{-16pt}\RETURN  Power and bandwidth allocations
 \end{algorithmic}
 \end{algorithm}
 
Proposed distributed  $\alpha$-fair resource allocation framework is summarized in Algorithm \ref{alg:HD} which is indeed a detailed algorithmic version of Fig. \ref{fig:dist}.  In Algorithm \ref{alg:HD}, BSs are only required to know channel gains of their own UEs. Following the initialization of the cluster bandwidths in line 2, the while loop between lines 3 and 11 iteratively forms clusters, obtains power allocations and update bandwidths until a termination term is not reached. In line 4, each BS first forms its clusters based on the steps given in Algorithm \ref{alg:CF}. According to the CF outcome, BSs solve slave problems to calculate the optimal power levels as explained in the previous section in lines 5 and 6, then transmits optimal power allocations to UEs in line 7. Thereafter, BSs share observed utilities with the MBS in line 8, which is followed by a bandwidth allocation update and dissemination  in line 9. 

It is necessary to point out that the first step of next iteration starts with reclustering if there is a change in cluster size or a significant variation in channel gains\footnote{While channel gain variations can be caused by user mobility, cluster size varies  either with bandwidth or QoS updates.}. Since all steps between lines 4 and 8 are executed by BSs in a parallel fashion\footnote{ If a BS fails to implement the proposed scheme, it can switch to OMA scheme until it recovers from the failure.}, the computational complexity for each BS is mainly driven by clustering and power control steps whose time complexity are given in Section V-B and Section VI-A, respectively.  Although the MBS has an extra duty for solving the master problem $\mathbf{\mathrm{M}}$, complexity of solving a convex problem is negligible in comparison with clustering and power allocation. 

Notice that there are two types of message passing in Algorithm \ref{alg:HD}: The former occurs between BS$_b$ and its users to share optimal power allocations, which is in the order of $U_b$. The latter takes place between smallcells and the MBS to receive bandwidth updates and report obtained utilities (line 9), which is in order of the total number of clusters, $C_b$. Therefore, proposed distributed method has a low communication overhead. 

Algorithm \ref{alg:HD} starts with clustering and then proceed with the power allocation.  Even though reversing this order can be thought as an alternative method, it is challenging due to several practical reasons: First, initial cluster bandwidths are necessary to calculate the feasible cluster sizes, that is the first step of clustering. Second, initial cluster bandwidths are also necessary for the power allocation problem because QoS constraints depends on the available cluster bandwidth. Since it is challenging to solve these two main subproblems without an initial bandwidth allocation at the first iteration, we follow the former approach.

\renewcommand{\arraystretch}{1}
\begin{table}[t!]
\centering
\caption{Table of Parameters}
\label{tab:parameters}
\scriptsize 
\resizebox{0.45\textwidth}{!}{%
\begin{tabular}{|l|l|l|l|l|l|}
\hline
Par.         & Value       & Par.          & Value           & Par.          & Value         \\ \hline
$\eta_b^u$   & $3.76$   &     $\epsilon$      & $10^{-7}$  & $\beta$   & $0.025$       \\ \hline
   $N_0$         & $-174$ $dBm$ &    $\bar{K}_b$ & $10$ & $P_u$  & $23$ $dBm$        \\ \hline
$W$        & $180$ kHz  & $U$ & $100$   &  $P_s$  & $30$ $dBm$    \\ \hline
$\Theta$      & $100$  & $B$           & $10$             & $P_m$ & $46$ $dBm$      \\ \hline
\end{tabular}%
}
\end{table}

\section{Numerical Results and Analysis}
\label{sec:res}

For the simulations, we consider $U$ UEs and $B$ SBSs uniformly distributed over a cell area of $500\:m \times 500\:m$ MBS. QoS requirements of UEs are randomly determined with a mean of 1 Mbps. All results are obtained by averaging over 200 network scenarios. The composite channel gain, $h_b^i$, between BS$_b$ and UE$_i$ is given as
\begin{equation}
\label{eq:gij}
h_b^i=A_b^i \delta_{b,i}^{-\eta_b^i} 10^{\xi_b^i/10} \E\{ | g_b^i|^2 \}
\end{equation}
where  $A_b^i $ is a constant related to antenna parameters, $\delta_{b,i}$ is the distance between the nodes, $\eta_b^i$ is the path loss exponent, $10^{\xi_b^i/10}$ represents the log-normally distributed shadowing, $\xi_b^i $ is a normal random variable representing the variation in received power with a variance of $\varsigma_b^i $, i.e., $\xi_b^i \sim \mathcal{N}(0,\varsigma_b^i )$, $\tilde{h}_b^i$ is the complex channel fading coefficient, $\E\{ \cdot \}$ is the expectation to average small scale fading out, and $ \E\{ | g_b^i|^2 \}$ is assumed to be unity. Unless it is stated explicitly otherwise, we use the default simulation parameters given in Table \ref{tab:parameters}. 
\begin{figure}[!htb]
    \centering
        \includegraphics[width=0.45 \textwidth]{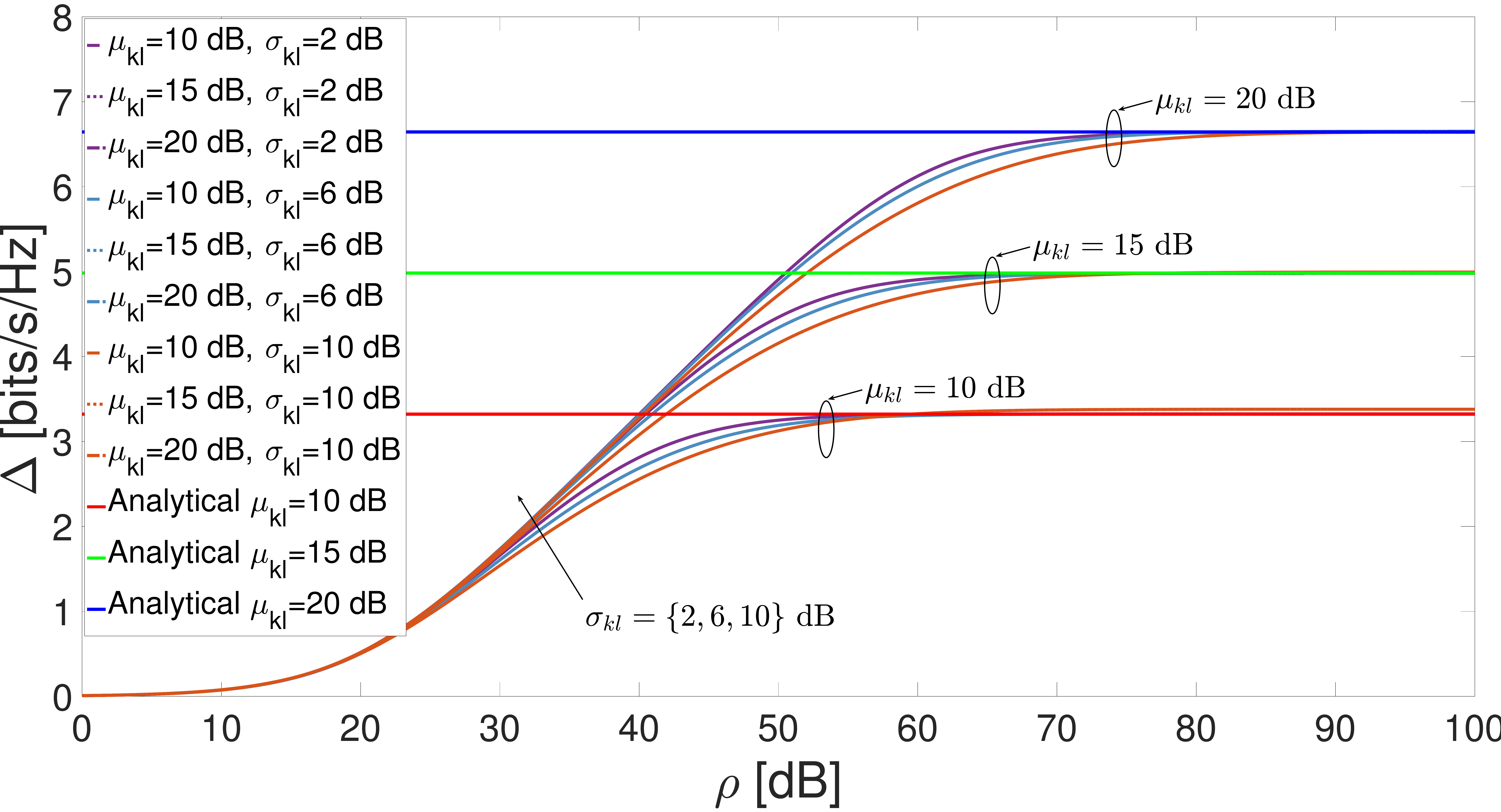}
        \caption{Impact of channel gain disparity on NOMA gain.}
        \label{fig:analy}
\end{figure}

\subsection{Impacts of Channel Gain Disparity and Decoding Order}
 Fig. \ref{fig:analy} compares the analytical findings obtained in Section \ref{sec:SIC} with the simulations where the reference user, UE$_k$,  placed 10 m away from the MBS with a deviation of 2 dB shadow fading. UE$_\ell$ is placed in (100, 300, 1000) m away with (0, 4, 8) dB deviation which results in $\mu_{kl}=\{10,15,20 \}$ dB and $\sigma_{kl}=\{2,6,10 \}$  dB, respectively.  As $\rho=\rfrac{\bar{P}_u}{N_0B}$ reaches up to 100 dB, simulation converges to the upper bound in (\ref{eq:delta_up_asymp2}). Please note that for  $P_u$ and $N_0$ given in Table \ref{tab:parameters}, practical values of $\rho$ ranges from 245 dB to 320 dB for bandwidths ranging from $1$ Hz to $20$ MHz. That is, the analytical upper bound is tight enough for practical values of $\rho$.

\subsection{The Largest Feasible Cluster Size Analysis}
Fig. \ref{fig:CS} shows the maximum feasible cluster size that can be handled by a single RB with respect to different $\bar{\Gamma}$ and $\epsilon$ values, where we do not use ceiling and floor functions in the lemmas for a better comparison. It is obvious that the cluster size increases as $\bar{\Gamma}$ and $\epsilon$ decrease, that is, NOMA can serve more users with low rates as the SIC efficiency improves. As a numerical example, a single RB with $\epsilon=10^{-5}$ can serve $3$ and $4$ UEs each with $0.5$ Mbps and $1$ Mbps, respectively. Fig. \ref{fig:CS} also compares the energy constrained cluster size with the unconstrained cluster size, where the weakest UE is located at the cell-edge. From numerical results, we observe that the cluster size difference between the two cases is negligible for practical channel gain values. Hence, changes in the largest cluster size are primarily affected by changes in cluster bandwidth and/or QoS demands. 

\begin{figure}[t!]
\centering
        \includegraphics[width=0.5\textwidth]{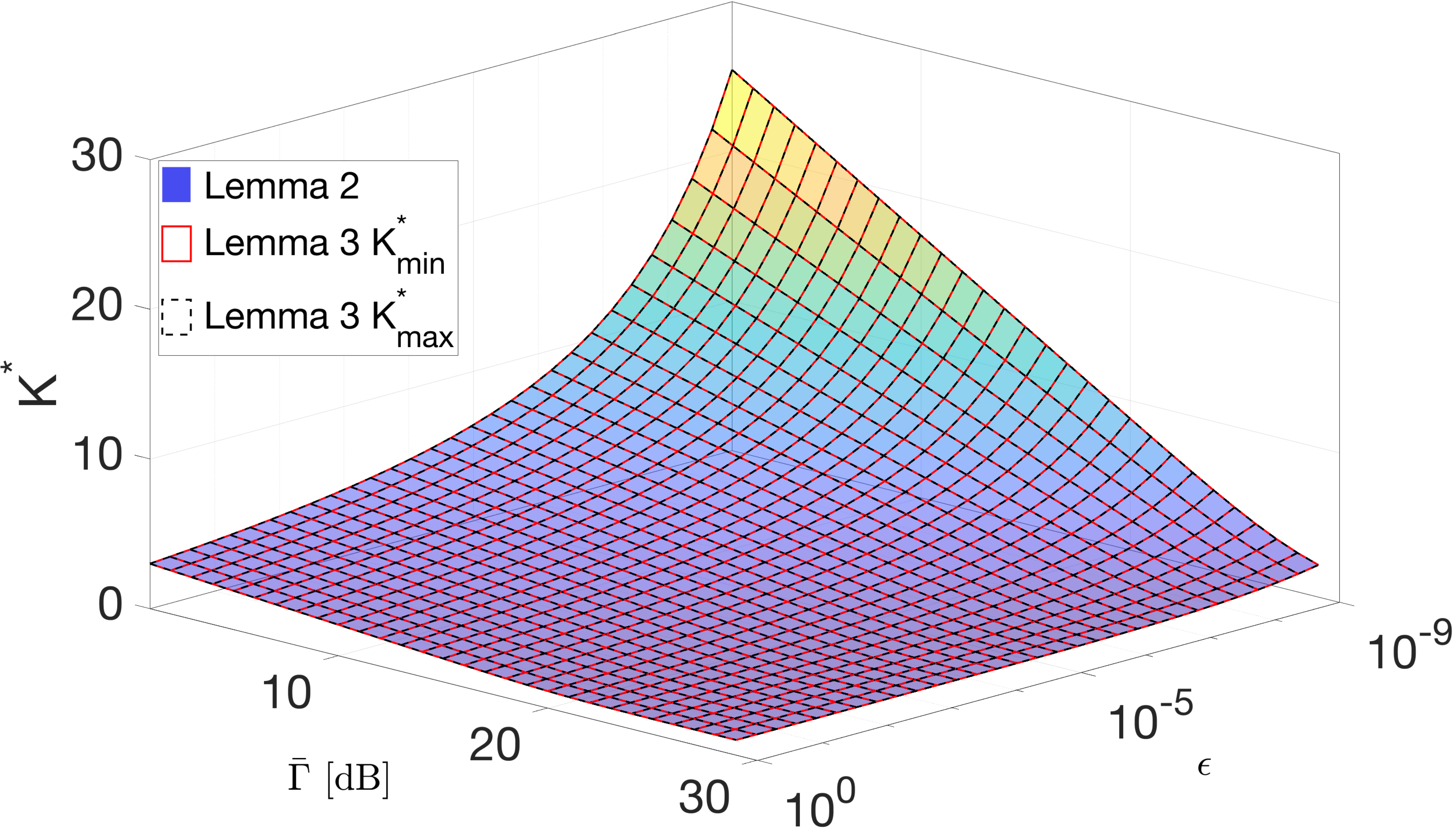}
        \caption{The largest feasible cluster size vs. different FEF and CSCs values.}
        \label{fig:CS}
\end{figure}

\subsection{Spectral and Energy Efficiency}

To investigate the impacts of cluster size on spectral and power efficiency, let us consider a single BS with 12 UEs which can be grouped into $\{6,4,3,2,1\}$ clusters with corresponding sizes of $\{2,3,4,6,12\}$. The normalized values for spectral efficiency, total power consumption, and energy efficiency are shown in Fig. \ref{fig:efficiency} where normalization is done for each curve individually (each has different units) using feature scaling, i.e.,  $x'=\frac{x-x_{\min}}{x_{\max}-x_{\min}}$, where $x'$ is the normalized value, $x$ is the actual value before the normalization, and $x_{\min}$ ($x_{\max}$) is the minimum (maximum) of all values of the curve before the normalization. The available bandwidth is uniformly divided between the cluster, thus, sumrate and spectral efficiency are both illustrated with red colored curve. Notice in Fig. \ref{fig:efficiency} that curves are not comparable to each other as a point in a curve is relative to other points in the same curve. Throughput the paper we plot figures with normalized values for two reasons; to reduce the number of figures by displaying different curves in different units and to provide a clear comparison in a 0-1 scale is intuitive to infer changes in percentage.  
\begin{figure}[!t]
    \centering
        \begin{subfigure}[b]{0.24 \textwidth}
\includegraphics[width=\columnwidth]{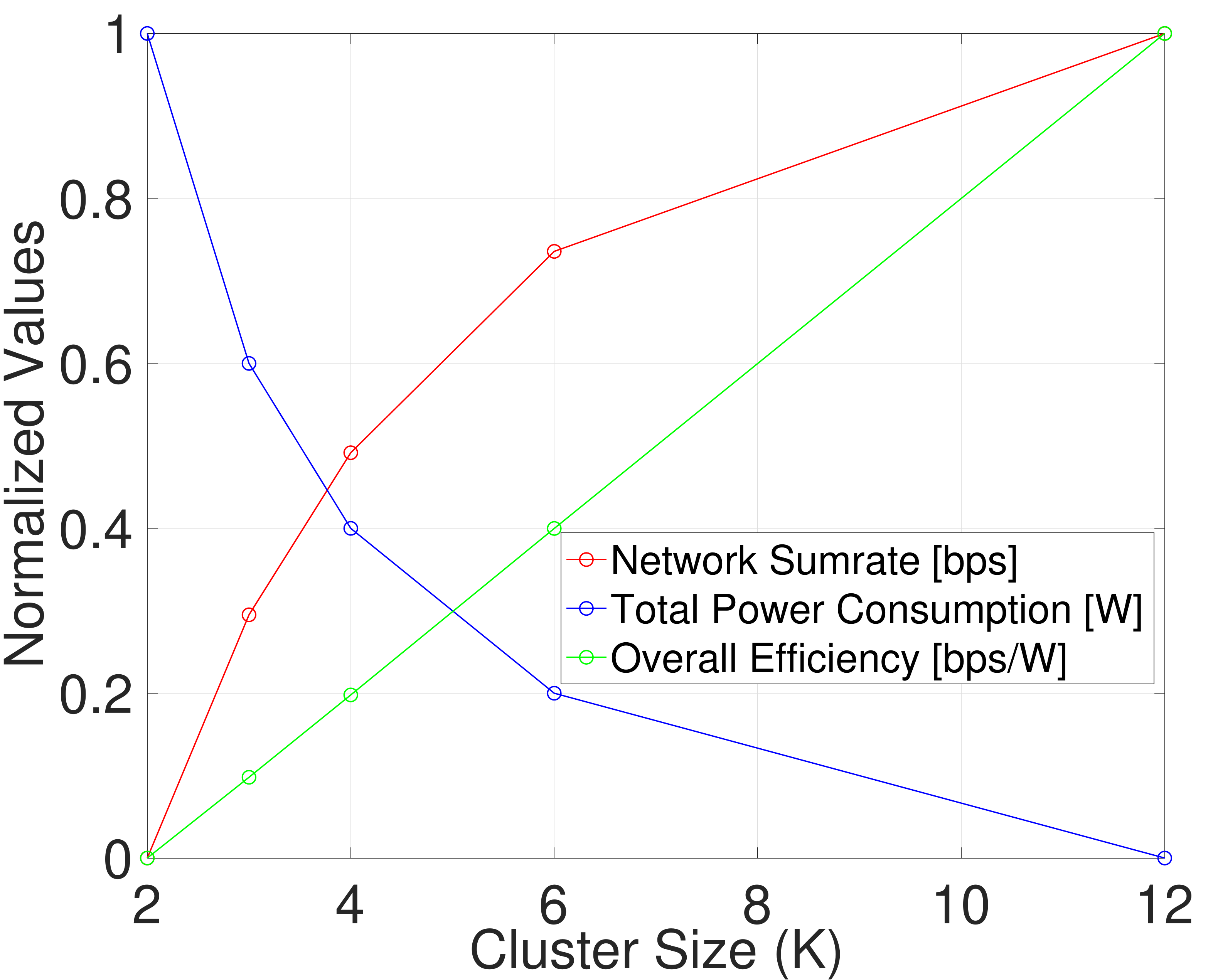}  
        \caption{}
 \label{fig:sum}
    \end{subfigure}
        \begin{subfigure}[b]{0.24\textwidth}
\includegraphics[width=\columnwidth]{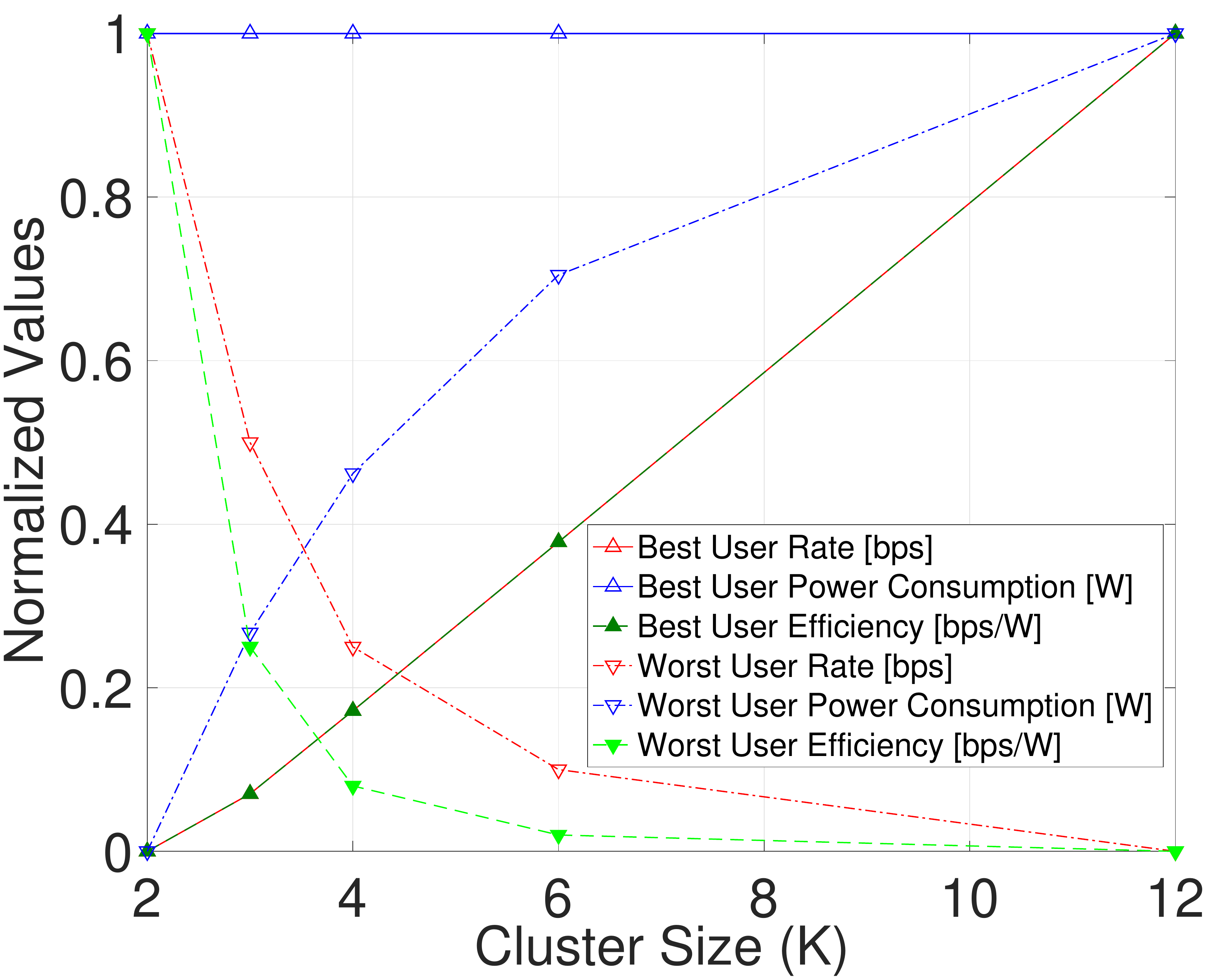}  
        \caption{ }
 \label{fig:worstbest}
    \end{subfigure}
   \caption{Impacts of cluster size on spectrum and energy efficiency.}
        \label{fig:efficiency}
\end{figure}

In Fig. \ref{fig:sum}, increasing the cluster size obviously enhances the overall spectral efficiency while it has a diminishing impact on the total power consumption of the clusters. For example, the number of UEs transmitting at the maximum power (i.e., the best UE) for cluster sizes of 2 and 3 is 6 and 4 (i.e., the number of clusters), respectively. As a result, cluster size 3 requires \%40 less power consumption while providing \%30 more sumrate than the basic NOMA, that yields a higher energy efficiency in units of $bps/W$. Let us now focus on Fig. \ref{fig:worstbest} where we depict the individual performance metrics of the best and worst UEs. While the best UEs keep transmitting at the same power level, their rates and efficiency increase with the cluster size since the bandwidth increases with decreasing number of clusters. However, this behavior follows an opposite direction for the worst UE case. Although the proposed solution allocates powers and bandwidths by taking the QoS constraints of all UEs into consideration, decreasing trend of the worst UE's data rate may cause coverage issues for large cluster sizes in large cells. In particular, providing the demanded QoS for cell-edge macro cell users may hinders the service coverage. At this point, decoupling the DL and UL user association can help in a great extend, which is already investigated and explained in Fig. \ref{fig:beta}. It is important to shed lights on the tradeoff between the worst and best case user performances. The worst case performance can be enhanced by setting a higher QoS requirement which naturally decreases the achievable rate of the best UE and thus the cluster sumrate. Nonetheless, this could yield a lower best case UE spectral efficiency than that is achievable by OMA, i.e., individual operation of the best case UE. Therefore, a good compromise must be forged to incentivize both UEs to enhance overall spectral efficiency of the network by using NOMA. To this end, cellular network operators can settle certain marketing policies to outline the rules for QoS setting which is satisfactory for both UEs.

\subsection{Impacts of Network Parameters on the NOMA Performance}

\begin{figure}[!t]
    \centering
        \begin{subfigure}[b]{0.24 \textwidth}
\includegraphics[width=\columnwidth]{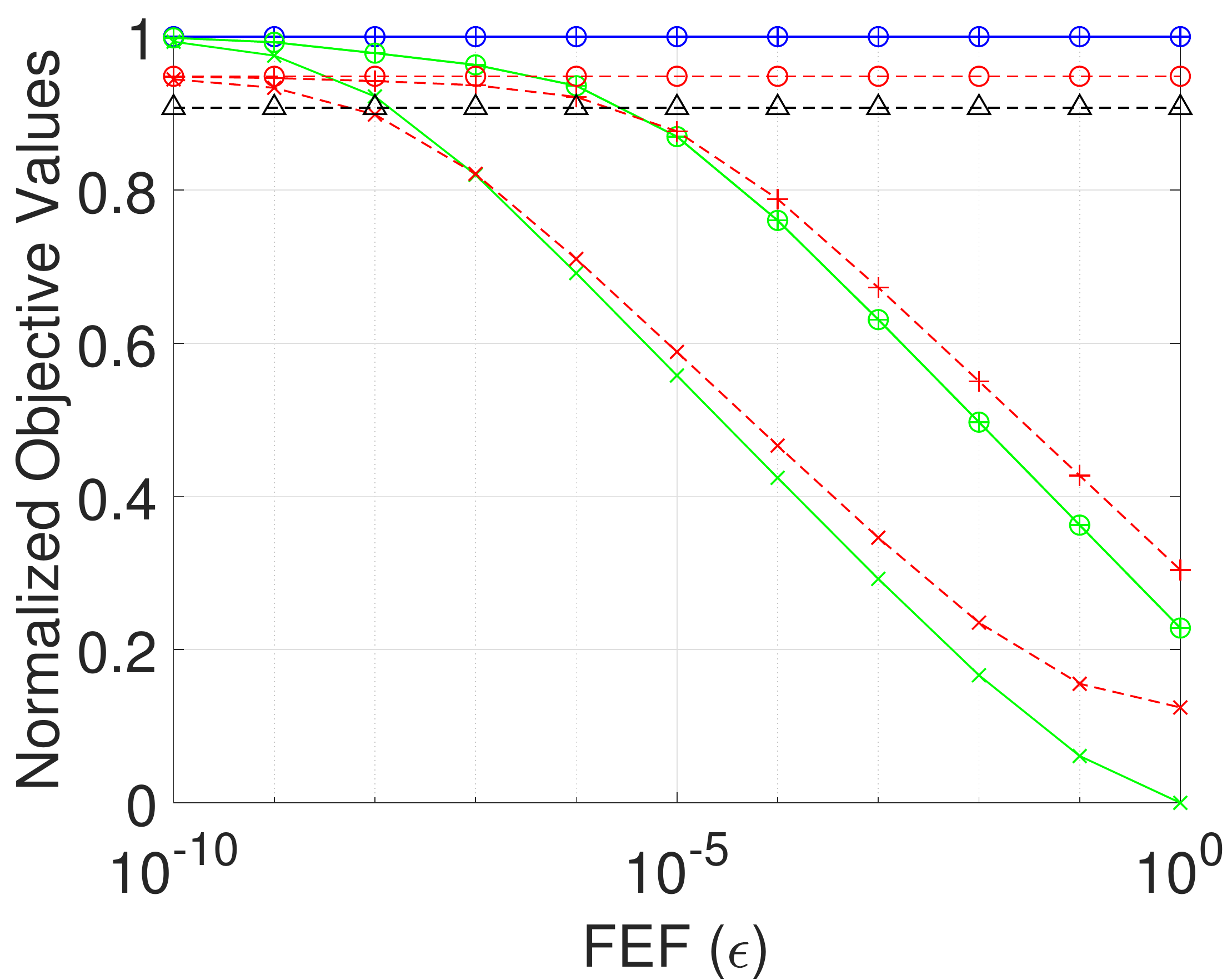}  
        \caption{$\alpha=0$ (Max. Throughput)}
 \label{fig:fef_0}
    \end{subfigure}
        \begin{subfigure}[b]{0.24\textwidth}
\includegraphics[width=\columnwidth]{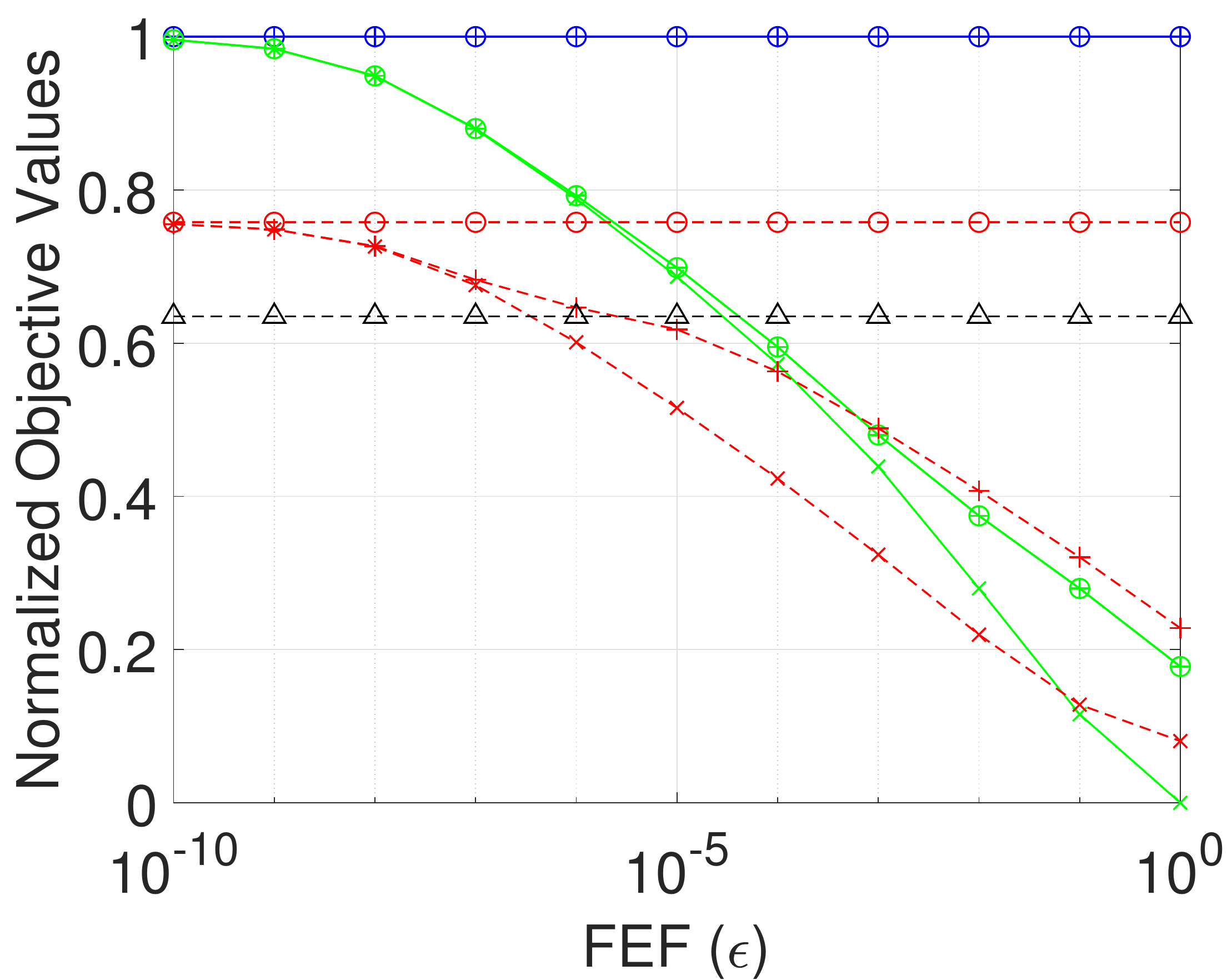}  
        \caption{$\alpha=0.25$}
 \label{fig:fef_25}
    \end{subfigure}
            \begin{subfigure}[b]{0.24 \textwidth}
\includegraphics[width=\columnwidth]{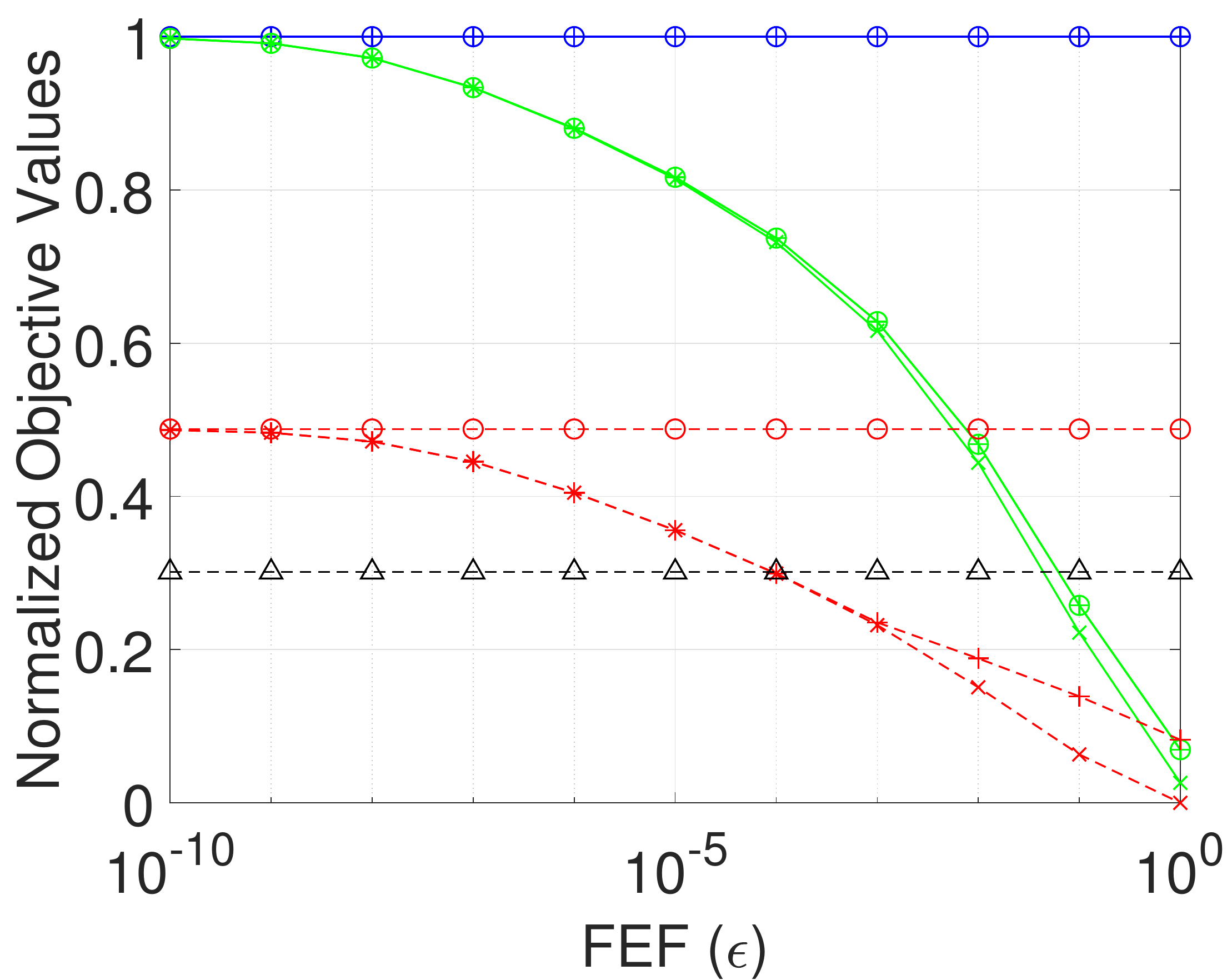}  
        \caption{$\alpha=0.5$}
 \label{fig:fef_50}
     \end{subfigure}
         \begin{subfigure}[b]{0.24 \textwidth}
\includegraphics[width=\columnwidth]{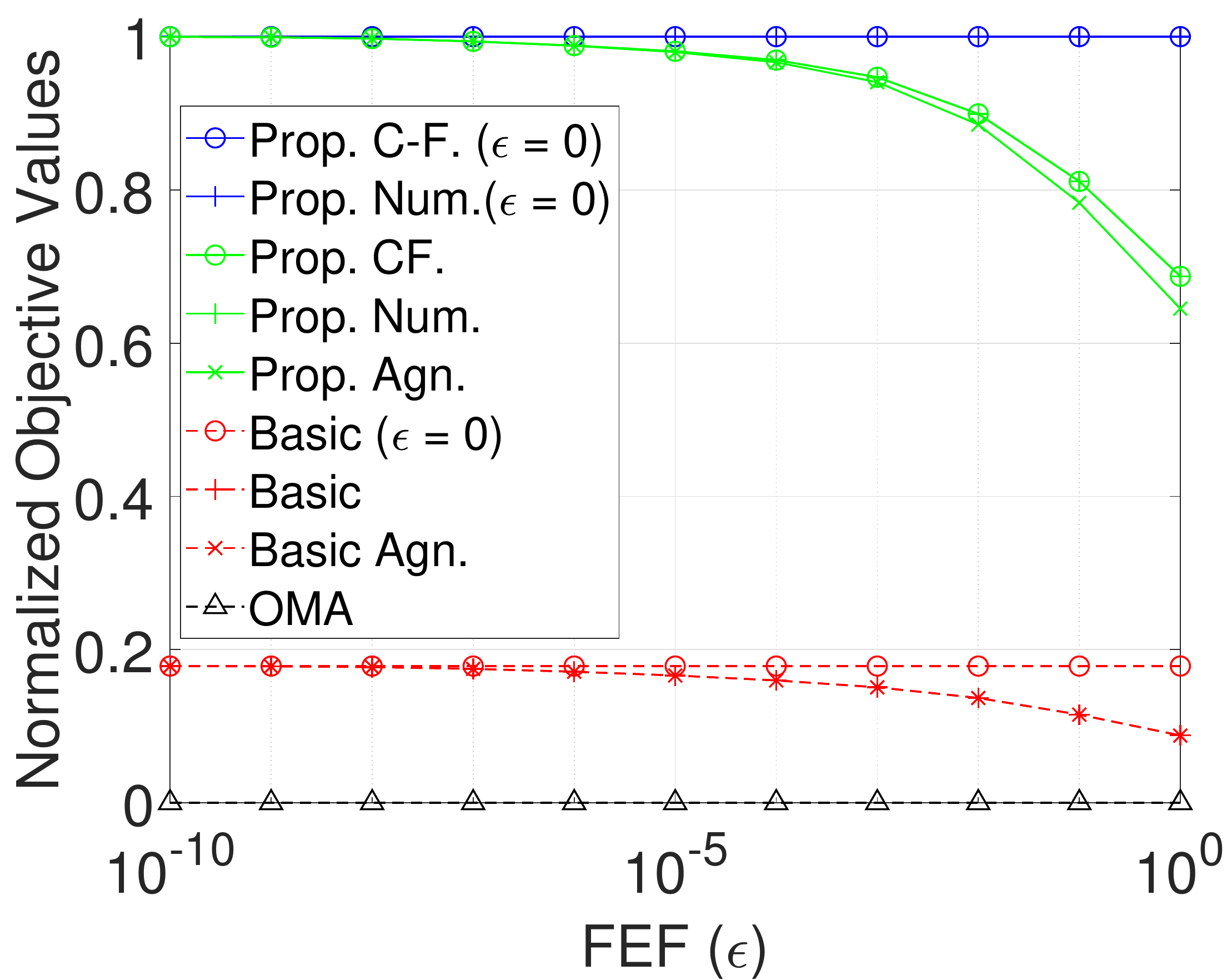}  
        \caption{$\alpha \to 1$ (Prop. Fairness)}
 \label{fig:fef_1}
            \end{subfigure}
    \caption{Normalized network sum-rate vs. FEF levels $\epsilon$.}
    \vspace{-0.5cm}
          \label{fig:fef}    
\end{figure}

For the sake of a better comparison, let us consider the following cases: 1) \textit{Prop. CF ($\epsilon=0$)}: Proposed CF algorithm under perfect NOMA scheme which is obtained by the closed-form expression given in Lemma \ref{lem:cfp} and drawn by blue colored {\color{blue}$-\hspace{-4pt}\ominus \hspace{-4pt}-$}, 2) \textit{Prop. Num. ($\epsilon=0$)}: This case is to check the validity of the previous case and drawn by blue colored {\color{blue}$- \hspace{-7pt} + \hspace{-7pt}-$}, 3) \textit{Prop. CF}: Proposed CF algorithm under imperfect NOMA scheme which is obtained by the closed-form expression given in Lemma \ref{lem:cfp} and drawn by green colored {\color{green}$-\hspace{-4pt}\ominus \hspace{-4pt}-$}, 4) \textit{Prop. Num.}: This case is to check the validity of the previous case and drawn by green colored {\color{green}$- \hspace{-7pt} + \hspace{-7pt}-$}, 5) \textit{Prop. Agn. }: The agnostic case is used to show the consequences of treating an imperfect NOMA as perfect by falsely assuming $\epsilon=0$ and  drawn by green colored {\color{green}$- \hspace{-7pt} \times \hspace{-7pt}-$}, 6)  \textit{Basic ($\epsilon=0$)}/\textit{Basic}/\textit{Basic Agn.}: This case compares the basic NOMA cluster of size two with the proposed case in 3/4/5 and drawn by red colored {\color{red}$-\hspace{-2pt}\circleddash \hspace{-12.5pt}- \hspace{2pt} -$}/{\color{red}$- \hspace{-2pt} + \hspace{-2pt}-$}/{\color{red}$- \hspace{-1pt} \times \hspace{-12pt}- -$}, and 7) \textit{OMA} corresponds to traditional OMA scheme where entire bandwidth is equally shared among the users and drawn by black colored $- \cdot \hspace{-9.5pt} \bigtriangleup \hspace{-2pt} -$. 

We demonstrate normalized network performance with respect to different network parameters in Fig. \ref{fig:fef} - Fig. \ref{fig:u} where normalized objective value is obtained via feature scaling, i.e., $x'=\frac{x-x_{\min}}{x_{\max}-x_{\min}}$, where $x'$ is the normalized value, $x$ is the value before the normalization, and $x_{\min}$ ($x_{\max}$) is the minimum (maximum) of all values within the figure before the normalization \footnote{
{Although Fig. \ref{fig:fef} - Fig. \ref{fig:u} show the network sumrate as a product of the optimized bandwidth and spectral efficiency, readers can also have an insight into the spectral efficiency trends under different network settings. }
}
It is common for all subfigures that the proposed solution provides a superior performance in comparison with the traditional basic NOMA and OMA schemes in all cases. This is mainly because of allowing a large number of cluster size, which enhances the spectral efficiency of the network.  On the other hand, the basic NOMA scheme delivers a performance in between the proposed solution and OMA scheme. Noting that obtained closed-form expressions perfectly match with numerical solutions, the agnostic approach deteriorates the network performance, which goes even below the OMA scheme in certain cases, especially in the maximum throughput case.

\begin{figure}[!t]
    \centering
        \begin{subfigure}[b]{0.24 \textwidth}
\includegraphics[width=\columnwidth]{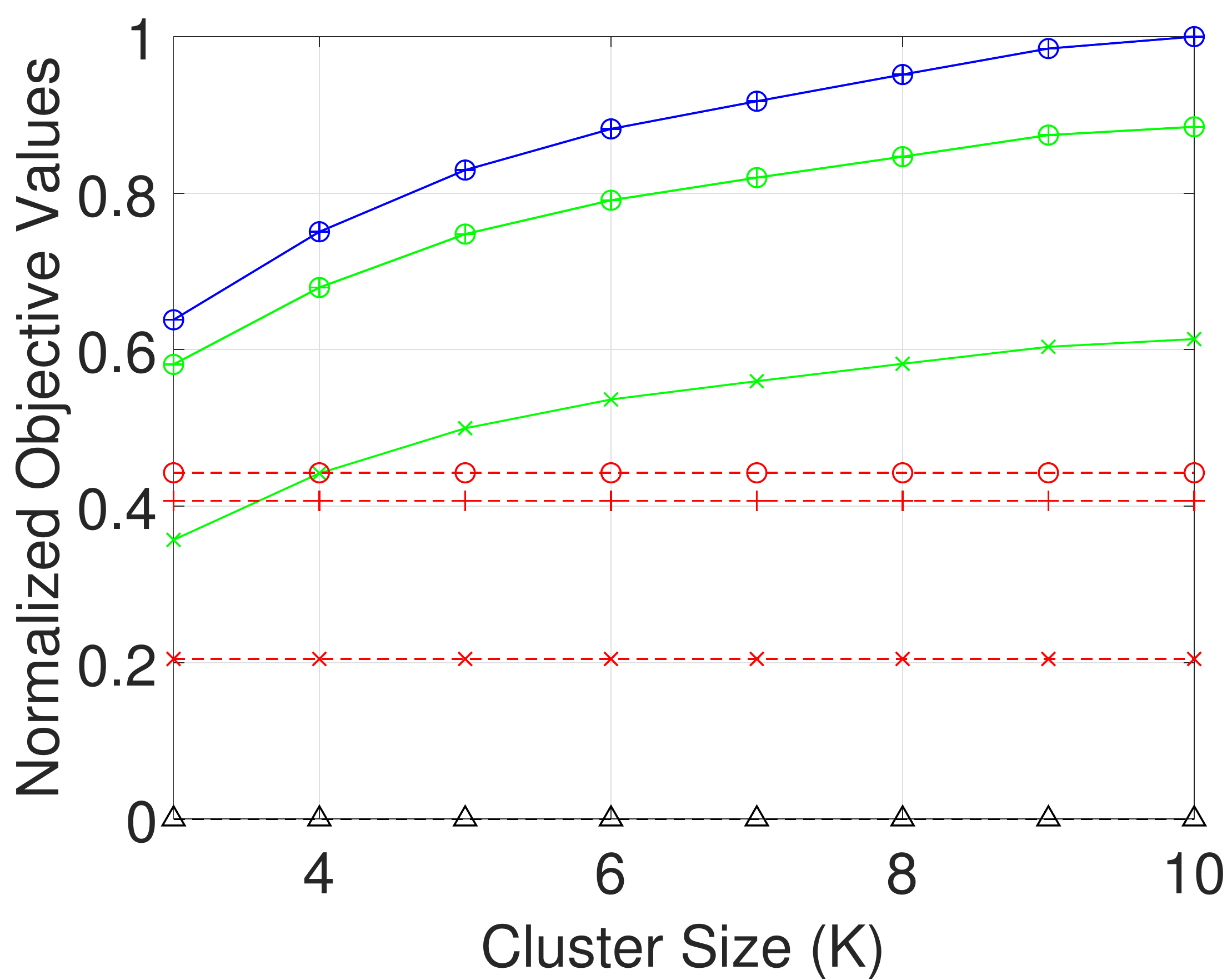}  
        \caption{$\alpha=0$ (Max. Throughput)}
 \label{fig:k_0}
    \end{subfigure}
        \begin{subfigure}[b]{0.24 \textwidth}
\includegraphics[width=\columnwidth]{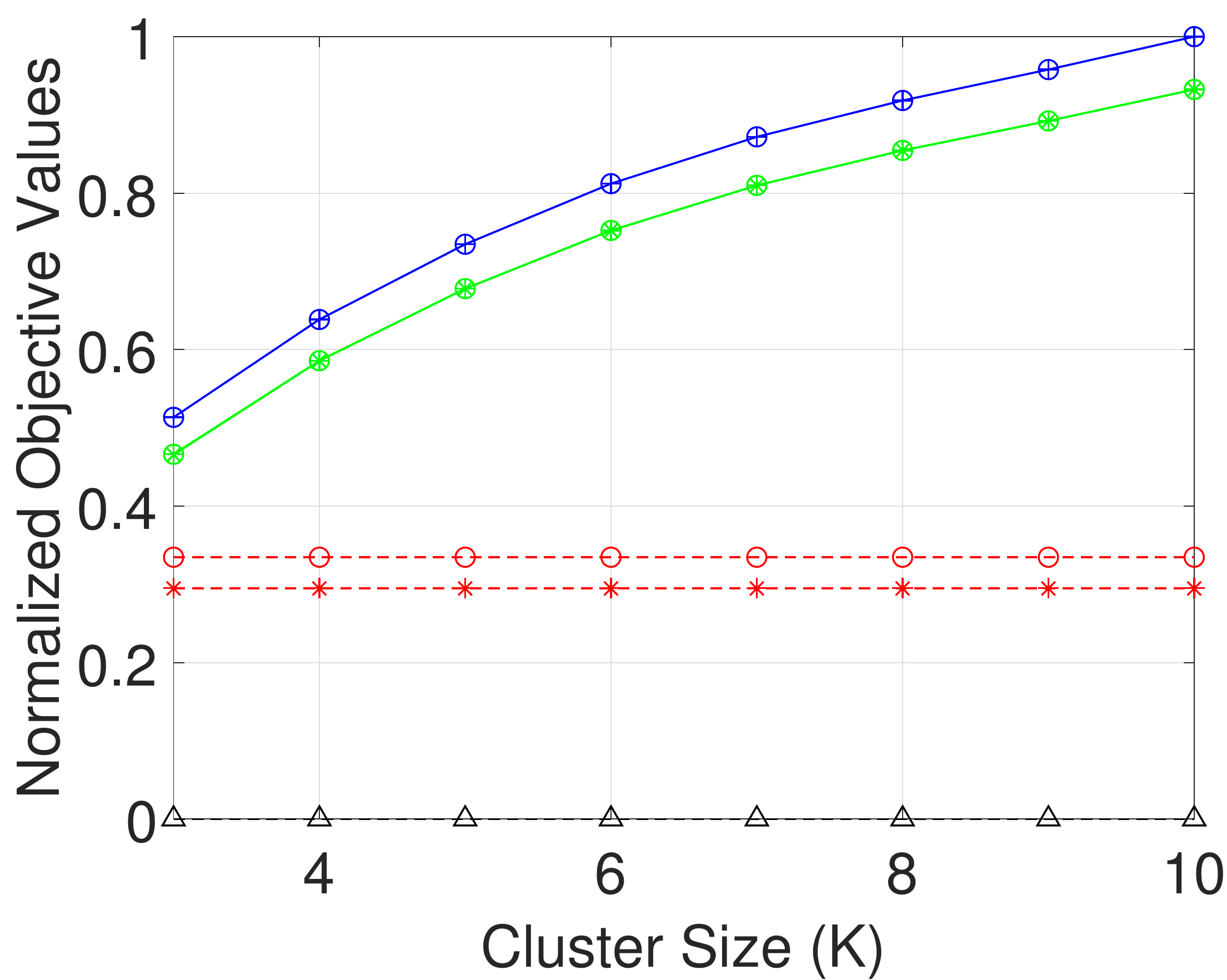}  
        \caption{$\alpha=0.25$}
 \label{fig:k_25}
    \end{subfigure}
            \begin{subfigure}[b]{0.24 \textwidth}
\includegraphics[width=\columnwidth]{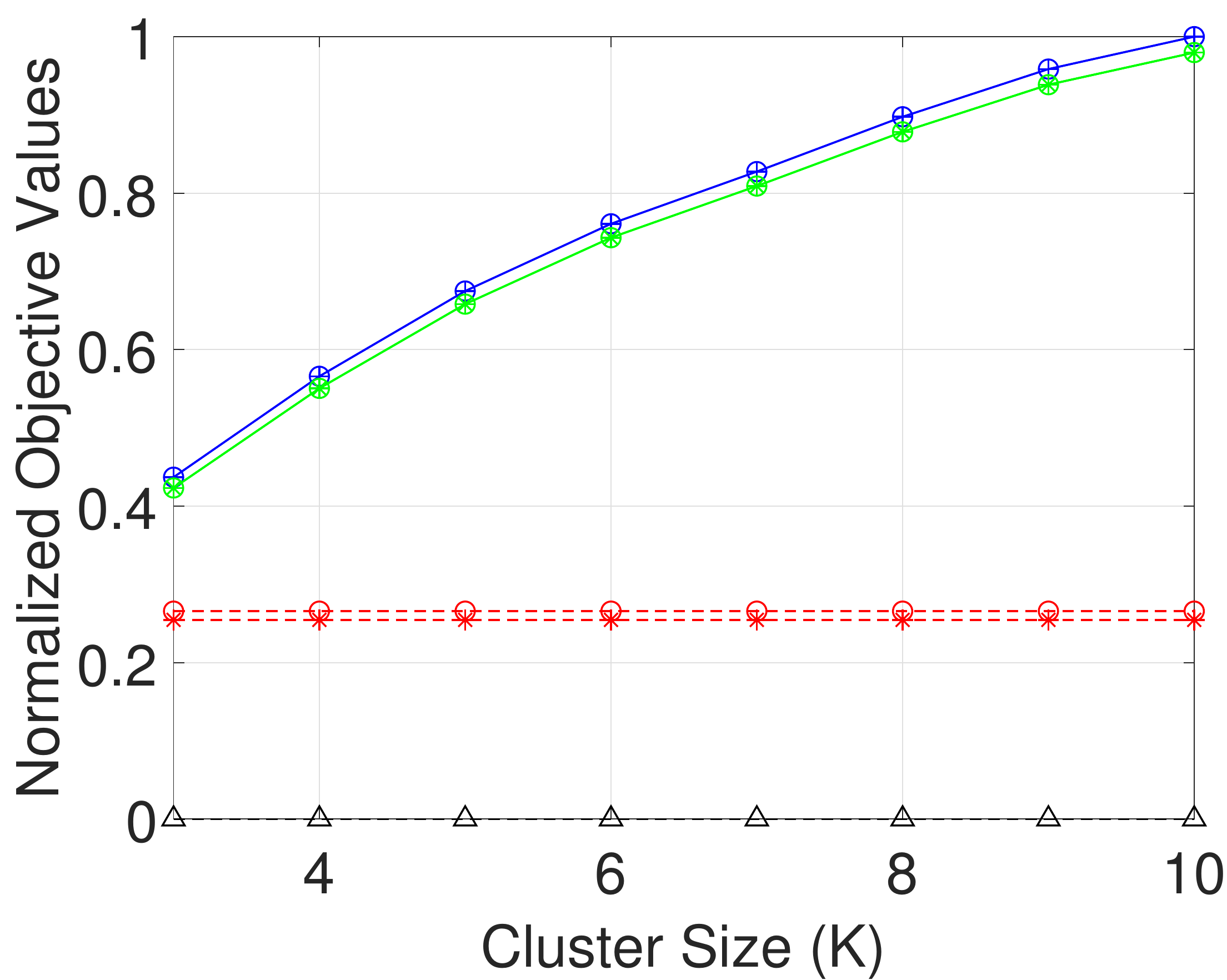}  
        \caption{$\alpha=0.5$}
 \label{fig:k_50}
     \end{subfigure}
         \begin{subfigure}[b]{0.24 \textwidth}
\includegraphics[width=\columnwidth]{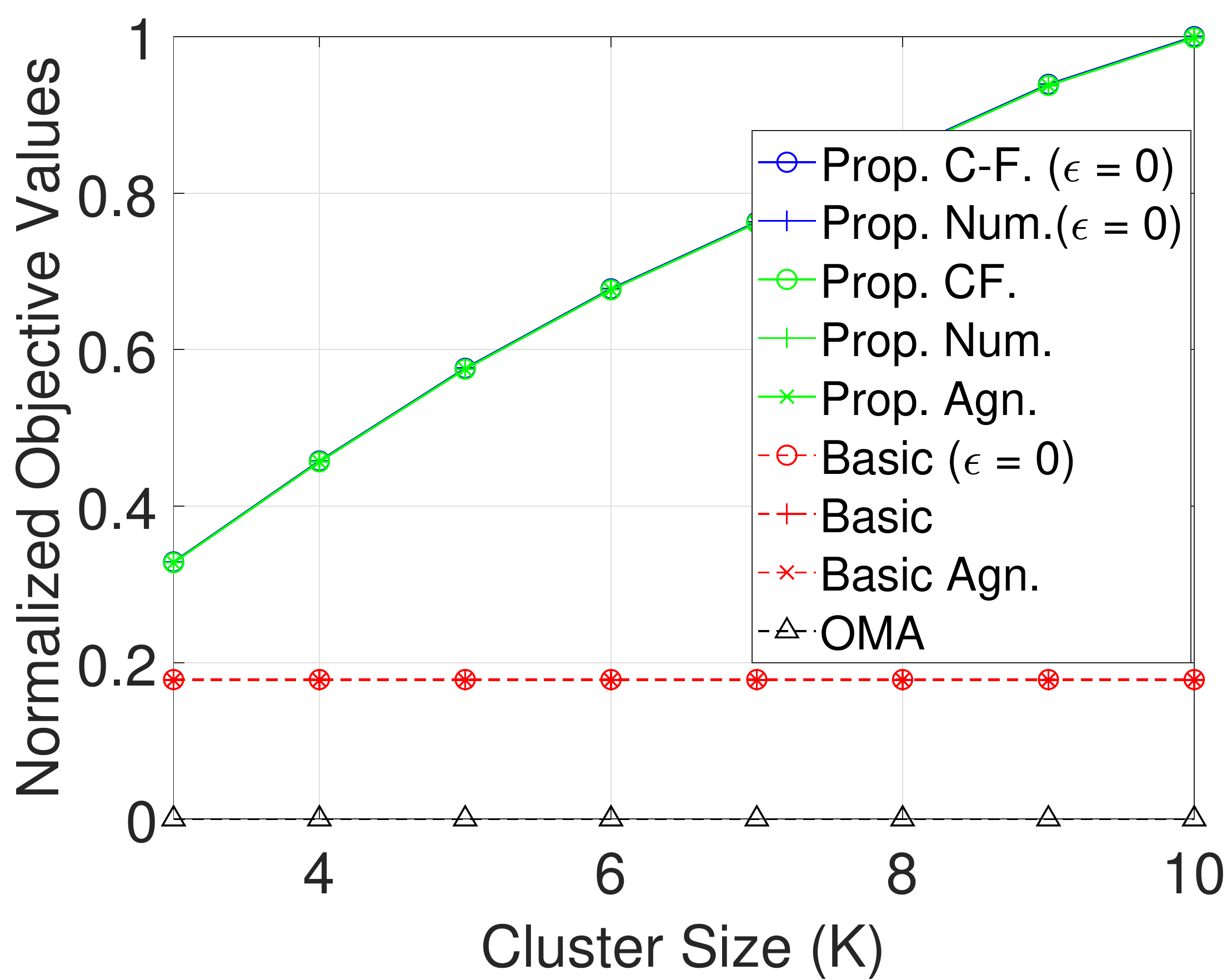}  
        \caption{$\alpha \to 1$ (Prop. Fairness)}
 \label{fig:k_1}
            \end{subfigure}
    \caption{Normalized network sum-rate vs. affordable cluster size $\bar{K}$.}
          \label{fig:k}    
\end{figure}
\begin{figure}[!t]
    \centering
        \begin{subfigure}[b]{0.24 \textwidth}
\includegraphics[width=\columnwidth]{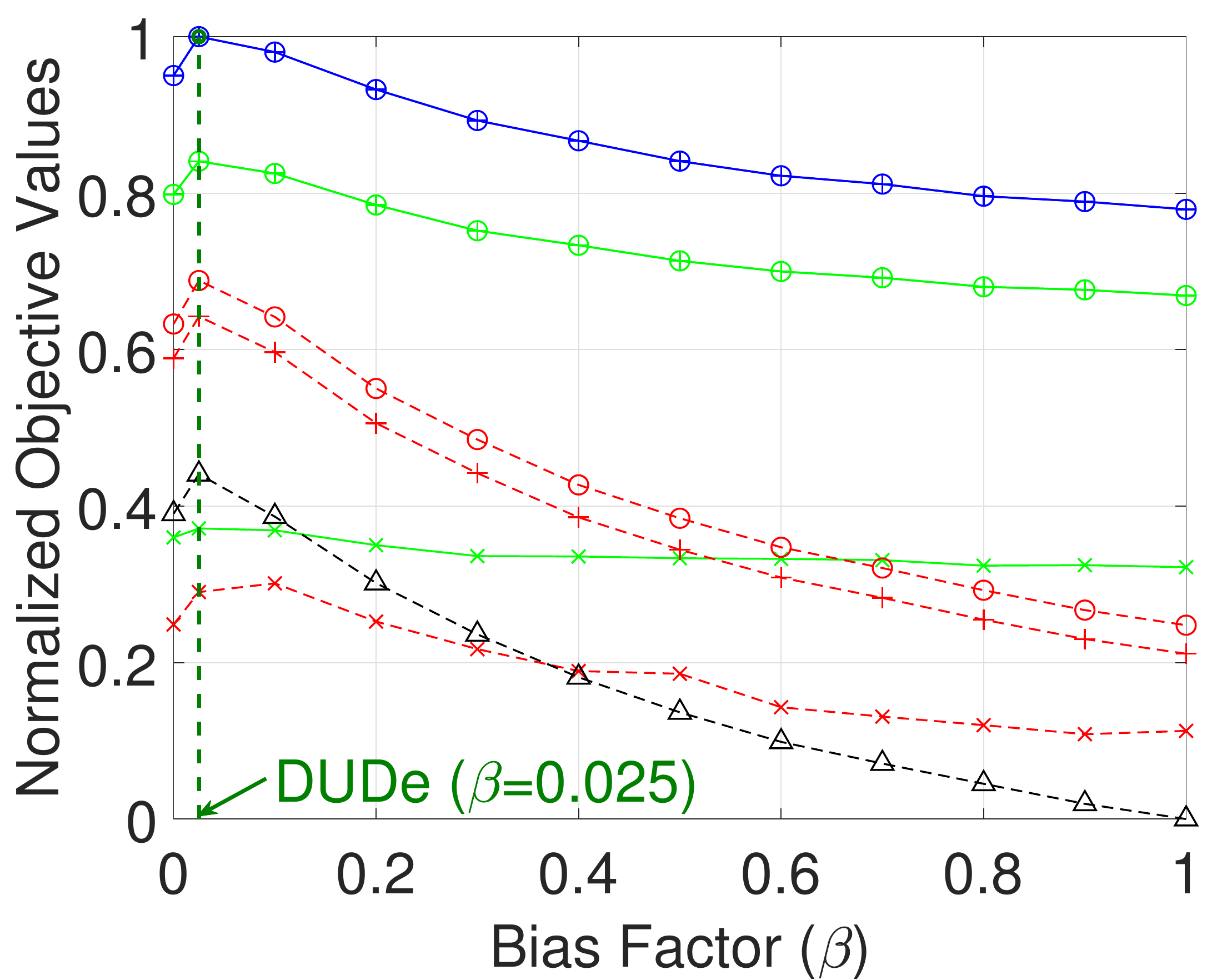}  
        \caption{$\alpha=0$ (Max. Throughput)}
 \label{fig:beta_0}
    \end{subfigure}
        \begin{subfigure}[b]{0.24 \textwidth}
\includegraphics[width=\columnwidth]{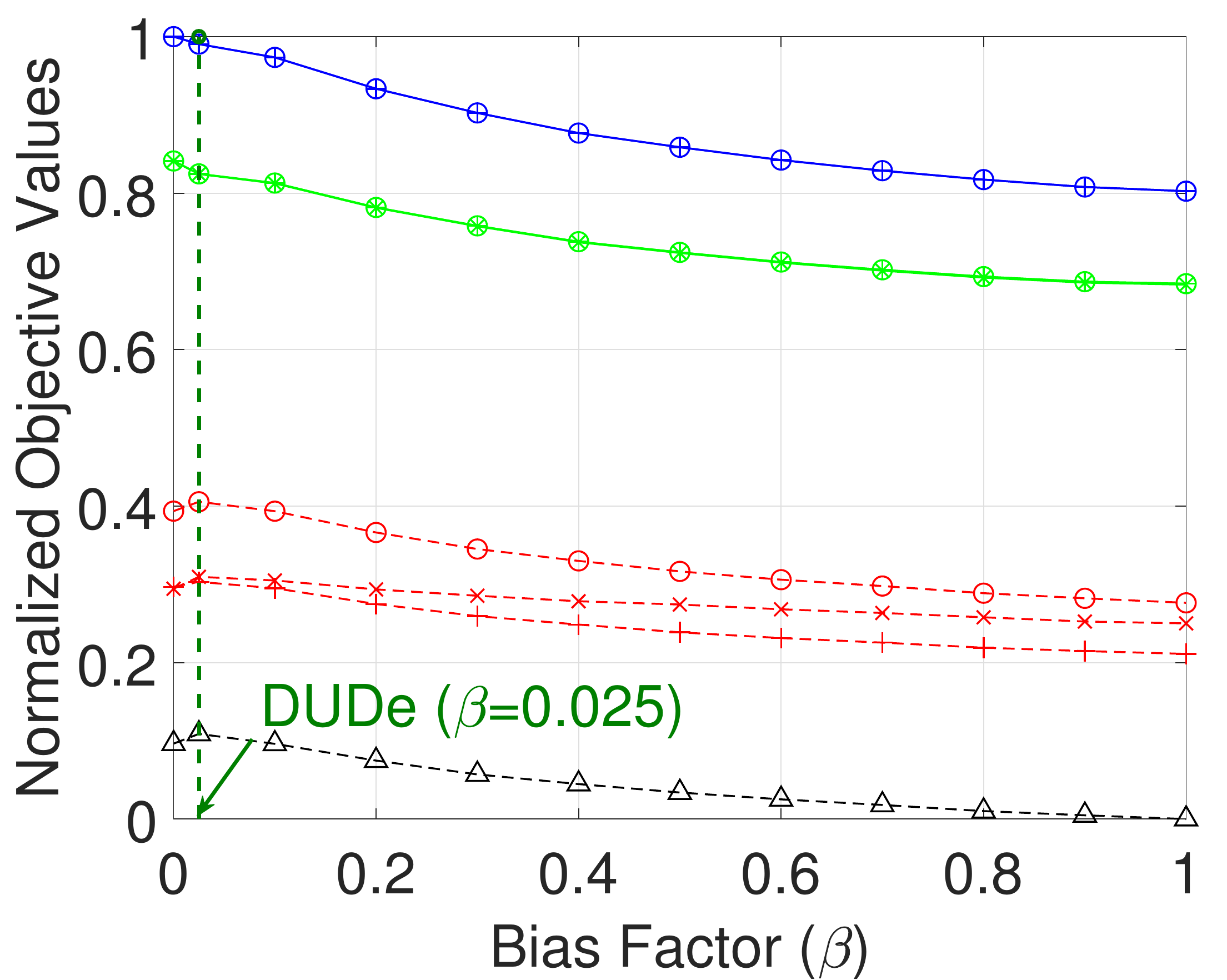}  
        \caption{$\alpha=0.25$}
 \label{fig:beta_25}
    \end{subfigure}
            \begin{subfigure}[b]{0.24 \textwidth}
\includegraphics[width=\columnwidth]{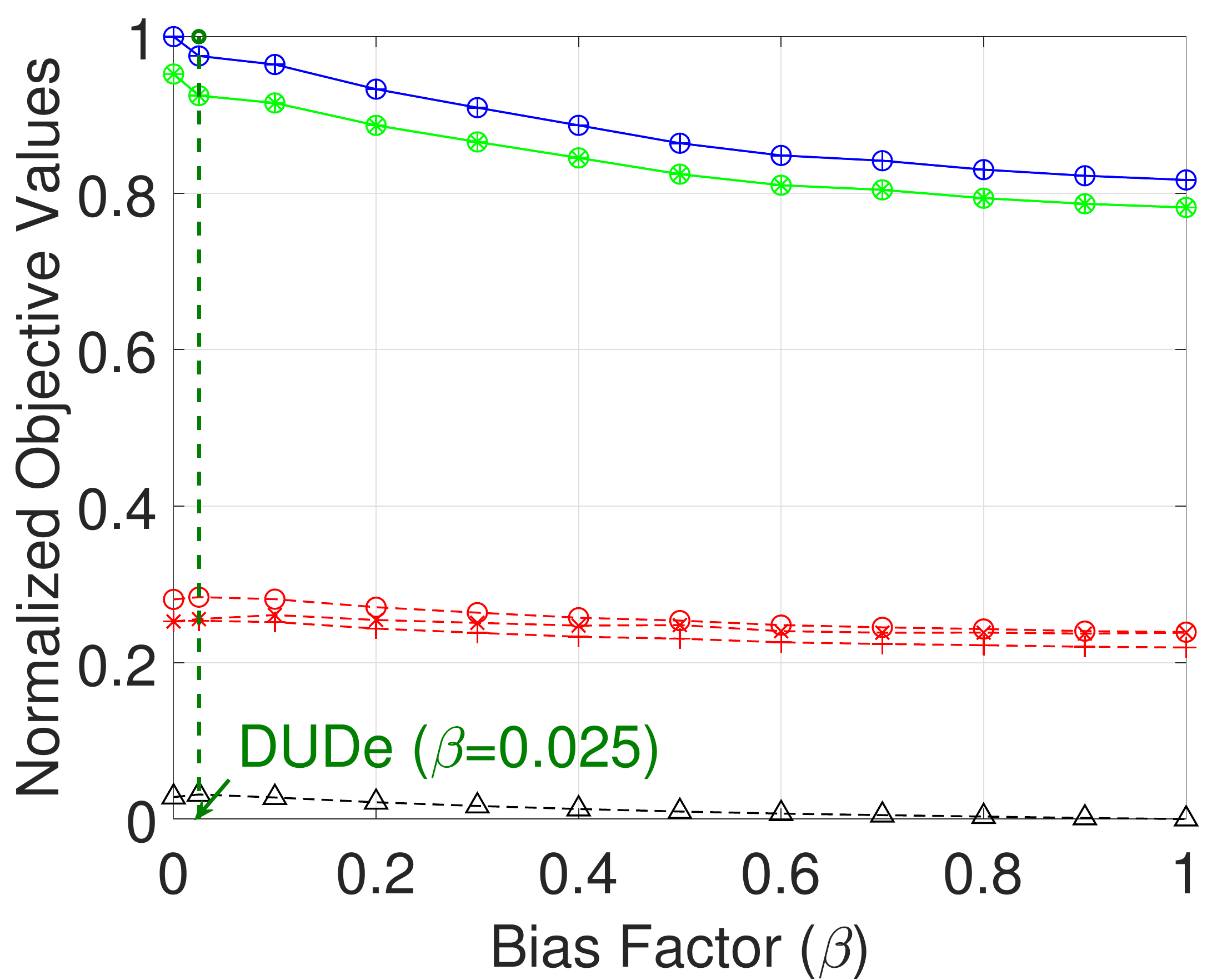}  
        \caption{$\alpha=0.5$}
 \label{fig:beta_50}
     \end{subfigure}
         \begin{subfigure}[b]{0.24 \textwidth}
\includegraphics[width=\columnwidth]{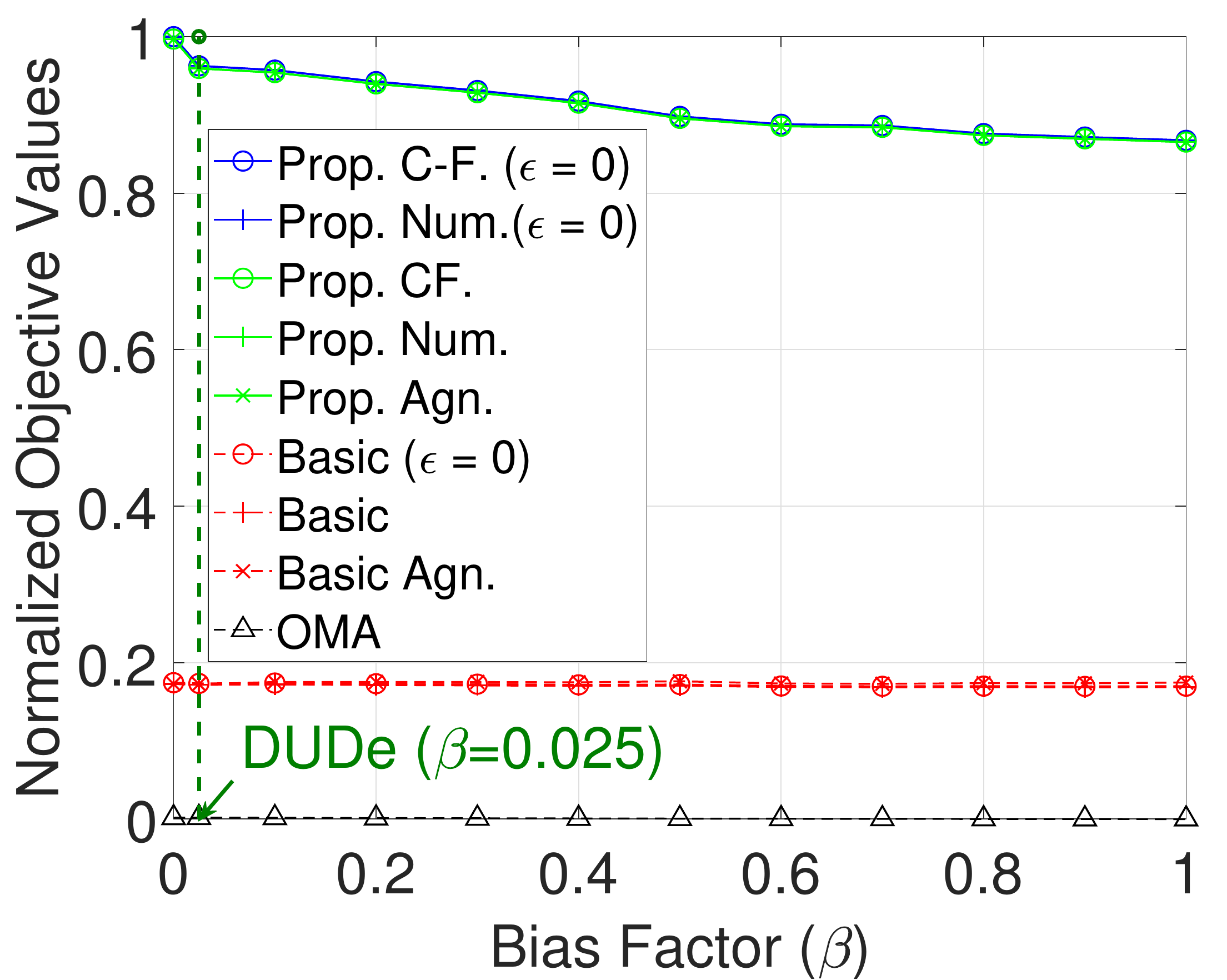}  
        \caption{$\alpha \to 1$ (Prop. Fairness)}
 \label{fig:beta_1}
            \end{subfigure}
    \caption{Normalized network sum-rate vs. Bias factor $\beta$}
          \label{fig:beta}    
\end{figure}
\begin{figure}[!t]
    \centering
        \begin{subfigure}[b]{0.24 \textwidth}
\includegraphics[width=\columnwidth]{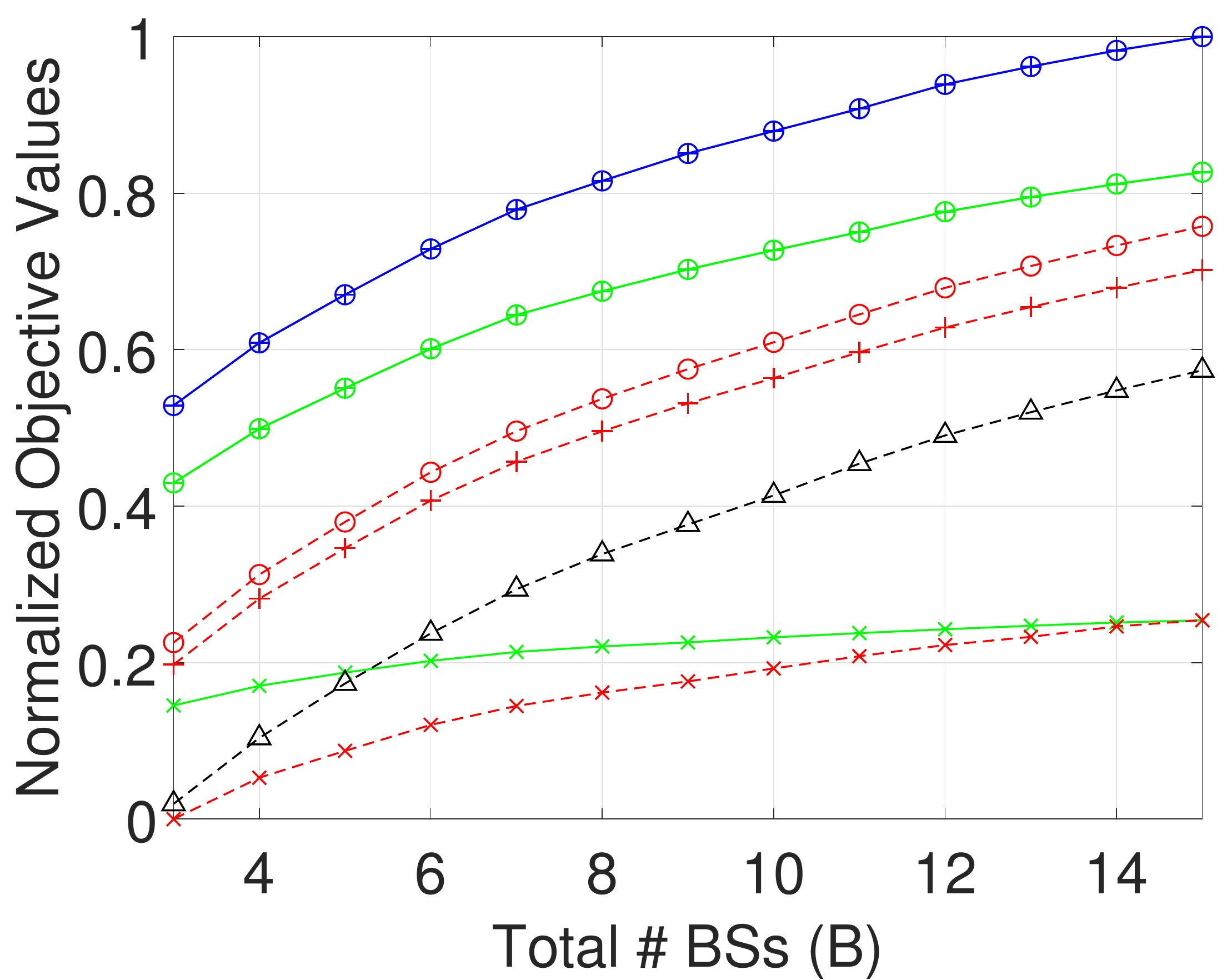}  
        \caption{$\alpha=0$ (Max. Throughput)}
 \label{fig:b_0}
    \end{subfigure}
        \begin{subfigure}[b]{0.24 \textwidth}
\includegraphics[width=\columnwidth]{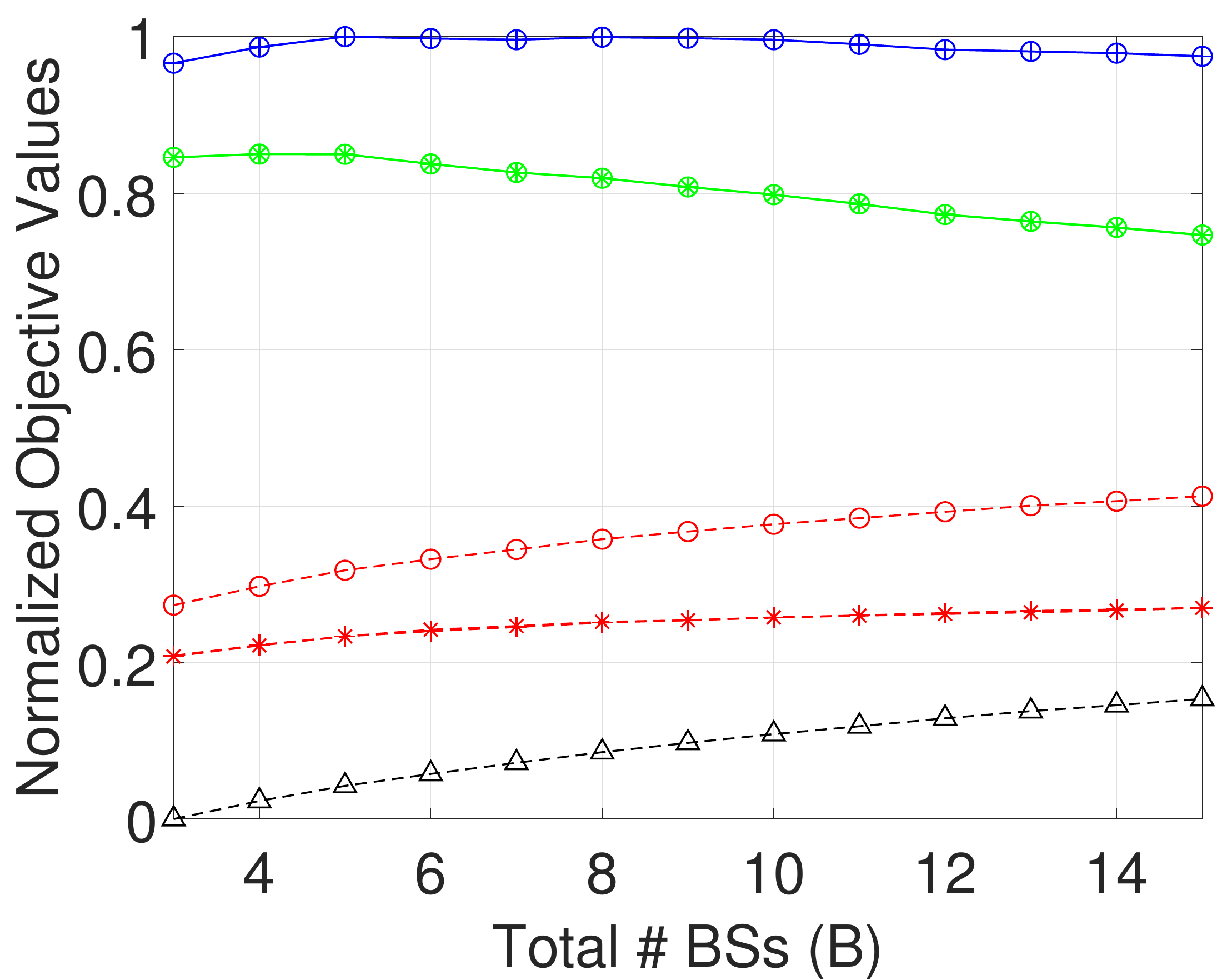}  
        \caption{$\alpha=0.25$}
 \label{fig:b_25}
    \end{subfigure}
            \begin{subfigure}[b]{0.24 \textwidth}
\includegraphics[width=\columnwidth]{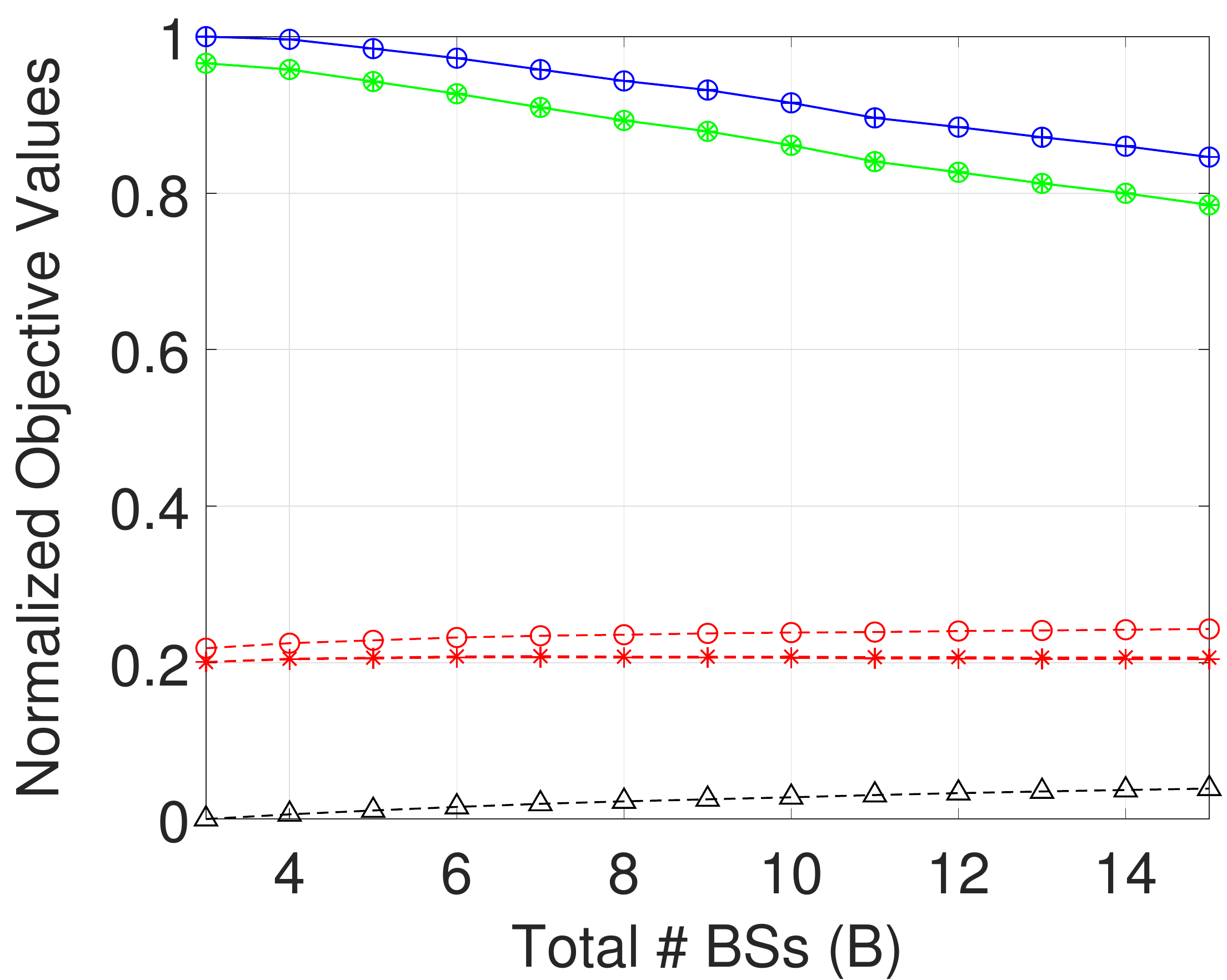}  
        \caption{$\alpha=0.5$}
 \label{fig:b_50}
     \end{subfigure}
         \begin{subfigure}[b]{0.24 \textwidth}
\includegraphics[width=\columnwidth]{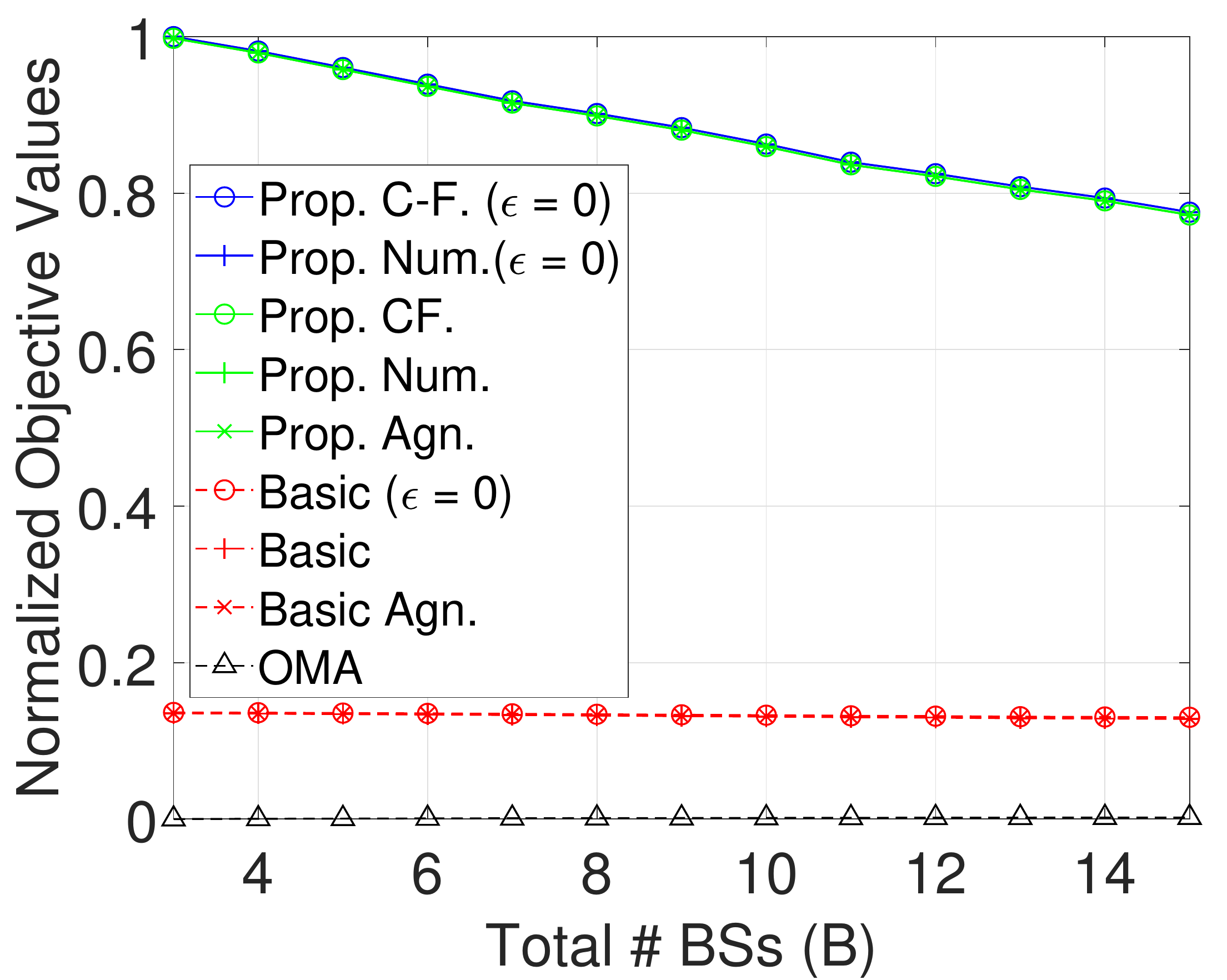}  
        \caption{$\alpha \to 1$ (Prop. Fairness)}
 \label{fig:b_1}
            \end{subfigure}
    \caption{Normalized network sum-rate vs. total number of SBSs $B$}
          \label{fig:b}    
\end{figure}
\begin{figure}[!t]
    \centering
        \begin{subfigure}[b]{0.24 \textwidth}
\includegraphics[width=\columnwidth]{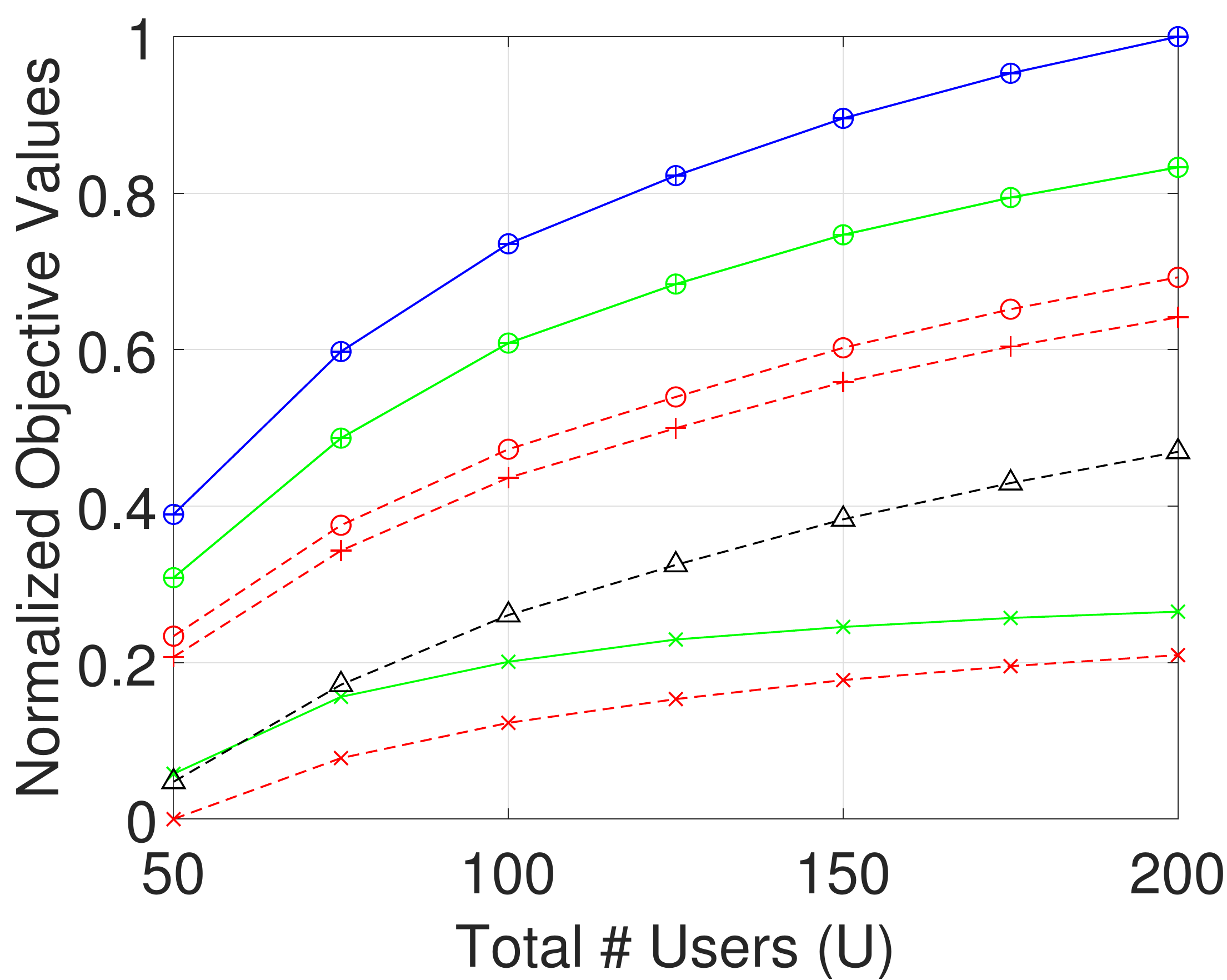}  
        \caption{$\alpha=0$ (Max. Throughput)}
 \label{fig:u_0}
    \end{subfigure}
        \begin{subfigure}[b]{0.24 \textwidth}
\includegraphics[width=\columnwidth]{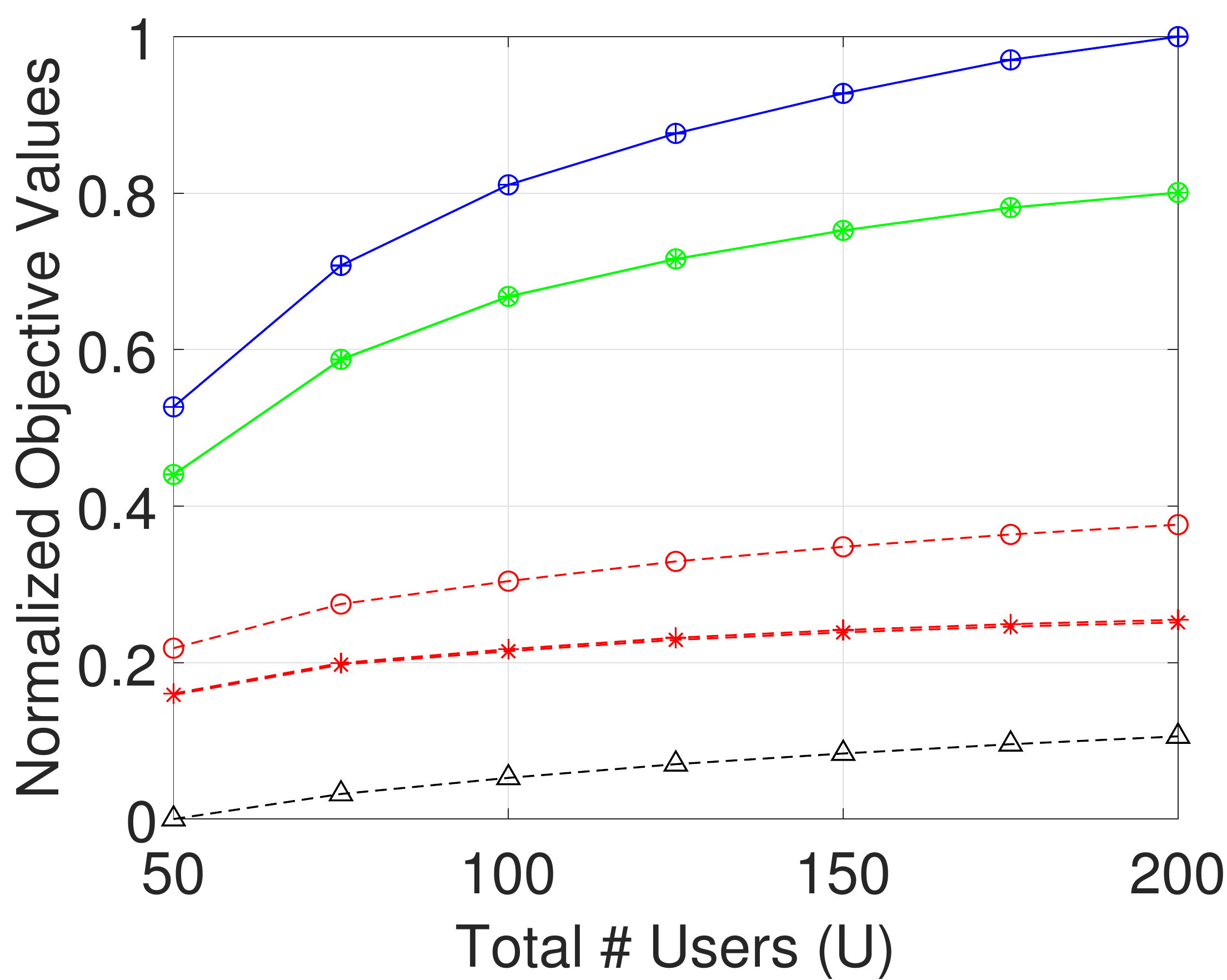}  
        \caption{$\alpha=0.25$}
 \label{fig:u_25}
    \end{subfigure}
            \begin{subfigure}[b]{0.24 \textwidth}
\includegraphics[width=\columnwidth]{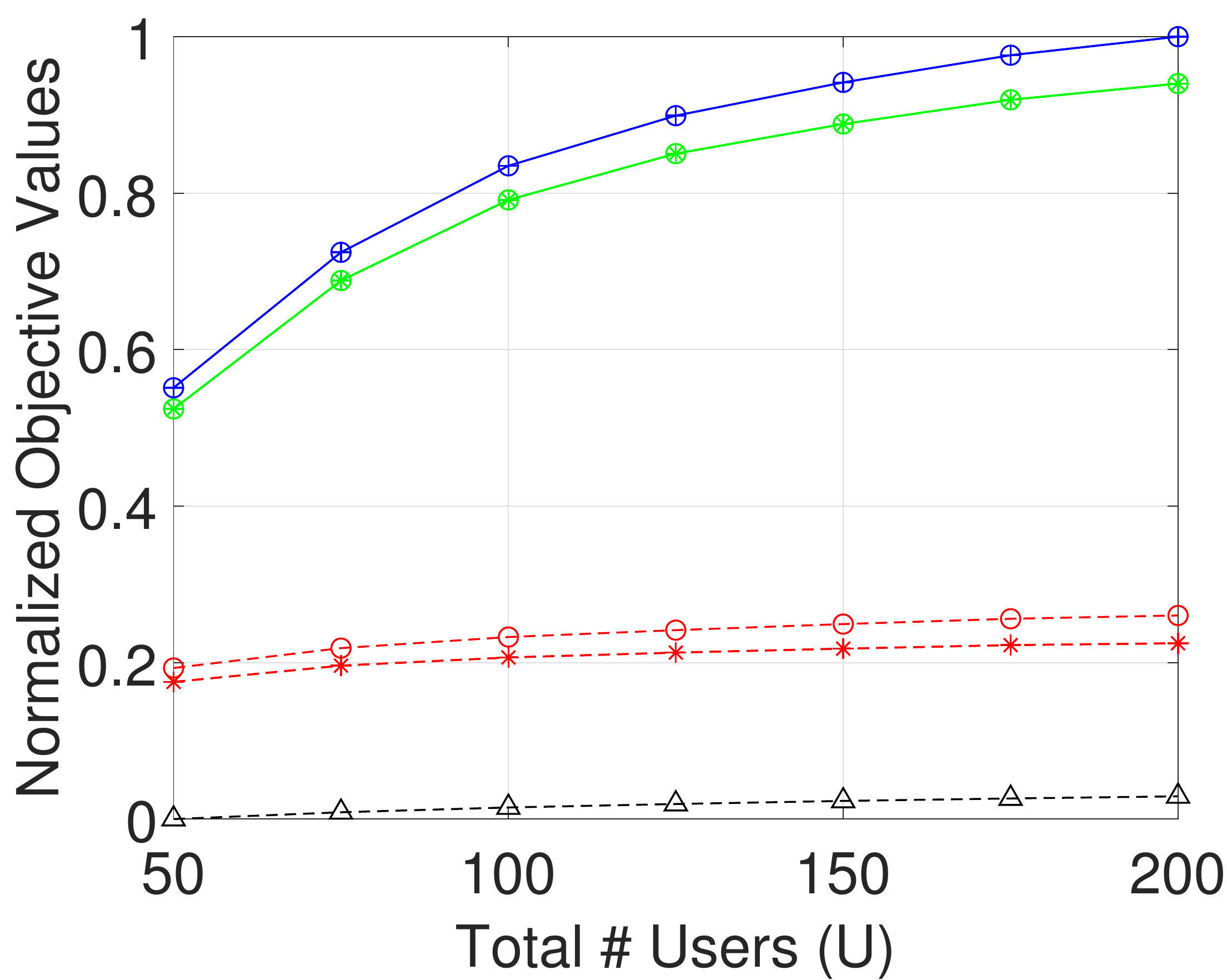}  
        \caption{$\alpha=0.5$}
 \label{fig:u_50}
     \end{subfigure}
         \begin{subfigure}[b]{0.24 \textwidth}
\includegraphics[width=\columnwidth]{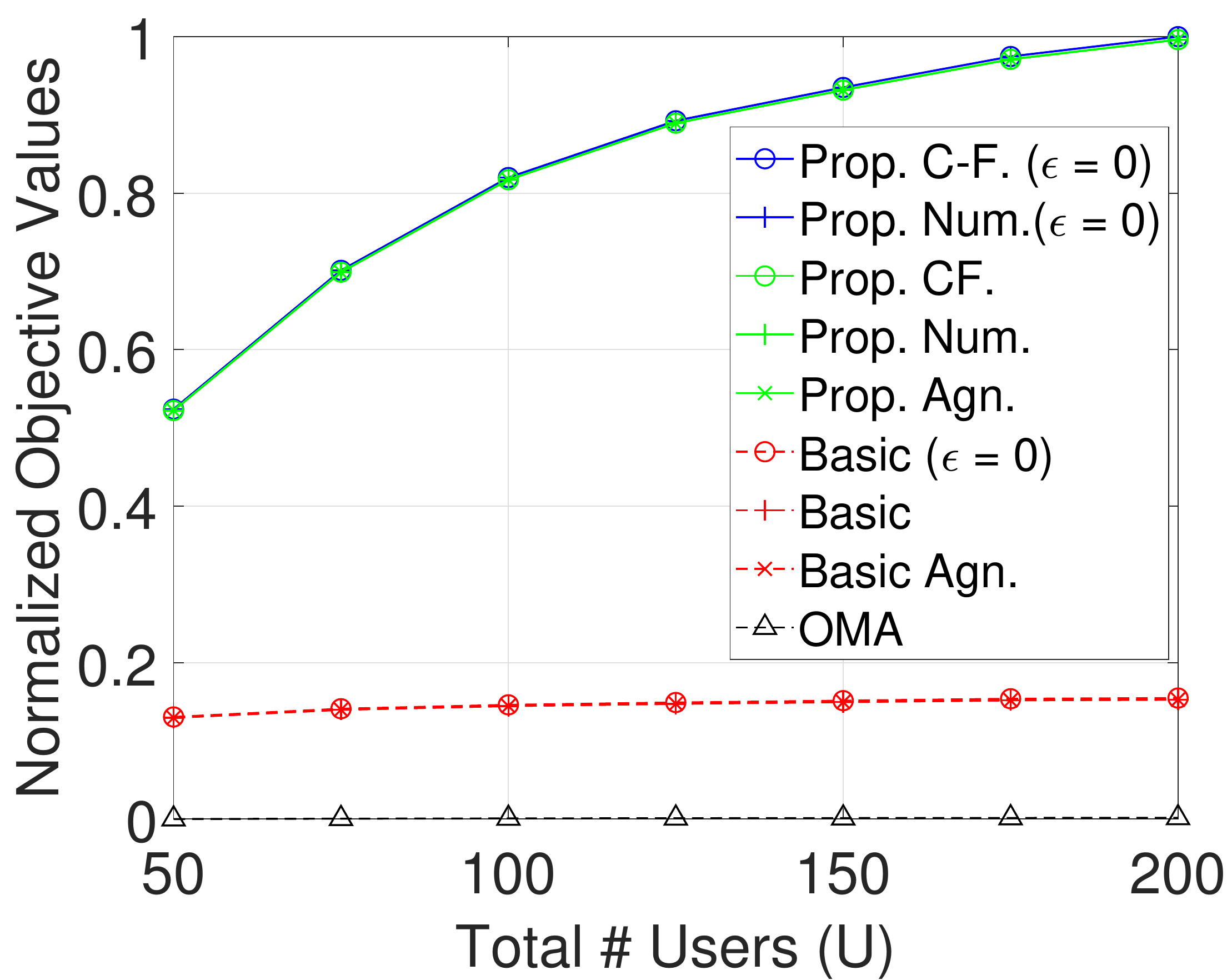}  
        \caption{$\alpha \to 1$ (Prop. Fairness)}
 \label{fig:u_1}
            \end{subfigure}
    \caption{Normalized network sum-rate vs. total number of UEs $U$.}
          \label{fig:u}    
\end{figure}

Let us start  our investigation with the influence of FEF levels on the network performance under different $\alpha$ scenarios as shown in Fig. \ref{fig:fef}. We involve ourselves in FEF effects since it is quite decisive on the pattern observed in the rest of the parameters. The severe performance degradation depicted in Fig. \ref{fig:fef_0} points out that NOMA cannot always deliver a better performance than OMA, thus, SIC receivers should have a desirable efficiency (i.e., $1-\epsilon$) in order to reduce the negative effects of the residual interference on the maximum throughput objective. We must also note for Fig. \ref{fig:fef_0} that a higher cluster size is not beneficial after a certain value of $\epsilon$ since putting more users on the same radio resource causes higher interference due to the increasing residual interference. As $\alpha \to 1$ in Fig. \ref{fig:fef_0}-\ref{fig:fef_1}, we observe the following behaviors: The performance gain between proposed and basic NOMA and that between basic NOMA and OMA increases monotonically. This can clearly be seen from perfect basic NOMA ({\color{red}$-\hspace{-2pt}\circleddash \hspace{-12.5pt}- \hspace{2pt} -$}) and OMA ($- \cdot \hspace{-9.5pt} \bigtriangleup \hspace{-2pt} -$) cases which are around $0.95$/$0.75$/$0.5$/$0.2$ and $0.9$/$0.65$/$0.3$/$0$ for $\alpha$ at $0$/$0.25$/$0.5$/$1$, respectively. This is indeed because of the combination of inherited NOMA fairness and proportional fairness enforced as $\alpha \to 1$.

Moreover, increasing the performance difference between proposed and basic NOMA curves points out that higher cluster sizes more favorable as $\alpha$ reaches to the proportional fairness. Another important pattern to observe is that the undesirable impacts of residual interference diminish since proposed ({\color{green}$-\hspace{-4pt}\ominus \hspace{-4pt}-$}) and basic({\color{red}$-\hspace{-2pt} \circleddash \hspace{-12.5pt}- \hspace{2pt} -$}) NOMA gets closer to corresponding perfect cases as $\alpha \to 1$. Because the UE that contributes the total cluster sumrate is protected no more against the negative impact of the FEF as $\alpha \to 1$ optimal scheme seeks for proportional fairness not only among the clusters but also among the members of a cluster. 

Fig. \ref{fig:k} clearly demonstrates the full benefit of allowing larger NOMA clusters. It is quite interesting that agnostic case of larger cluster sizes turns in a better performance than the perfect basic NOMA scheme of maximum throughput case in Fig. \ref{fig:k_0}. It is also clear that as $\alpha$ approaches the proportional fairness, the network enjoys a larger cluster size more than the maximum throughput. For instance, the ratio between the proposed and the basic NOMA is $1.5$ and $1.9$ for $\bar{K}=3$ and $\bar{K}=10$ under the maximum throughput case, respectively. On the other hand, the ratio between the proposed and the basic NOMA is $3$ and $5$ for $\bar{K}=3$ and $\bar{K}=10$ under the proportional fair objective, respectively.

The impact of UE association scheme on the network performance is demonstrated in Fig. \ref{fig:beta}. As shown in Fig. \ref{fig:beta_0}, network throughput hits a peak when users are associated as per DUDe $\left(\beta=\frac{P_s}{P_m}\right)$, that monotonically degrades as $\beta \rightarrow 0$ and $\beta \rightarrow 1$ in the DUCo scheme. In particular, $\beta = 1$ loads the MBS down with the entire traffic, thus, deliver the worst performance mainly because of the deteriorated cell-edge performance and its inevitable consequence of uncancelled or residual interference to other users. Except for the proposed case, this trend also applies for other $\alpha$ cases. Another important pattern to observe is that negative influence of $\beta$ in the network performance diminishes as $\alpha \to 1$.

Fig. \ref{fig:b} presents the performance trend for increasing number of SBS under different $\alpha$ cases. Increasing $B$ helps the maximum throughput case due to more desirable channel gains since DUDe has a better opportunity to associate UEs with nearby BSs. However, increasing $B$ does not show the same trend as $\alpha \to 1$ because a larger cluster size is more preferable for proportional fairness. Similarly, Fig. \ref{fig:u}  exhibits the increasing behavior of the performance as the total number of UEs increases. Apparently, increasing the total number of UEs provide less performance increase as $\alpha \to 1$.

\section{Conclusion}
\label{sec:conc}
In this paper, an $\alpha$-fair resource allocation and cluster formation  problem is studied for DUDe HetNets under the imperfection of NOMA scheme due to the residual interference and SIC constraints. Unlike the traditional basic NOMA cluster of size two, the largest feasible cluster size is derived in the closed-form as a function cluster bandwidth, SINR requirements, and the FEF levels. Numerical results have clearly shown that a larger cluster size provides a better performance thanks to improved spectral efficiency. Furthermore, we develop a distributed cluster formation and power-bandwidth allocation framework which iteratively updates clusters, power allocations, and bandwidths. For a given bandwidth and cluster formation, optimal power control policy is derived in closed form. By extensive simulation results, we have demonstrated that delivered network performance has different trends under various network parameters. 

\appendices
\section{Proofs for Lemma \ref{lem:CS} and Lemma \ref{lem:CSspec}}
\label{app_lem_1}
\begin{proof}[Proof of Lemma \ref{lem:CS}]
This proof follows from the discussion within the paragraph before Lemma \ref{lem:CS}. Exploiting the eigenvalue equation, $\vect{\mathrm{F}}\vect{\mathrm{\nu}}=\lambda_F \vect{\mathrm{\nu}}$, we have the following set of equations
\begin{align}
\label{eq:ev_func1}
 \frac{\lambda_F}{\bar{\Gamma}_1(\theta) }\nu_1 & = \sum_{i=1}^K \nu_i\\
\label{eq:ev_func2}
 \frac{\lambda_F}{\bar{\Gamma}_i(\theta) }\nu_i&=\epsilon \sum_{j=1}^{i-1} \nu_i + \sum_{k=i+1}^{K} \nu_i, i \geq 2,
\end{align}
where $\nu_i$ is the $i^{th}$ element of the eigenvector of $\vect{\mathrm{\nu}}$, which can be obtained recursively as follows
\begin{align}
\label{eq:eigen1}
&\nu_{2}=\nu_{1} \frac{\epsilon+\frac{\lambda_F}{\bar{\Gamma}_{1}(\theta) }}{1+\frac{\lambda_F}{\bar{\Gamma}_{2}(\theta) }}, \:
\nu_{3}=\nu_{2} \frac{\epsilon+\frac{\lambda_F}{\bar{\Gamma}_{2}(\theta) }}{1+\frac{\lambda_F}{\bar{\Gamma}_{3}(\theta) }}= \\
\nonumber  &\nu_{1} \frac{ \left(\epsilon+\frac{\lambda_F}{\bar{\Gamma}_{1}(\theta) }\right) \left(\epsilon+\frac{\lambda_F}{\bar{\Gamma}_{2}(\theta) } \right) }{ \left(1+\frac{\lambda_F}{\bar{\Gamma}_{2}(\theta) }\right) \left(1+\frac{\lambda_F}{\bar{\Gamma}_{3}(\theta) } \right) }, \hdots,
\nu_{k}=\nu_{1}\prod^{k}_{i=2}\frac{\left(\epsilon+\frac{\lambda_F}{\bar{\Gamma}_{i-1}(\theta) }\right)}{\left(1+\frac{\lambda_F}{\bar{\Gamma}_{i}(\theta) }\right)}.
\end{align}
Assuming a non-ideal SIC receiver, $\epsilon>0$, $\vect{\mathrm{H}}$ becomes an irreducible positive matrix. Ensuring $\lambda_F<1$ in \eqref{eq:eigen1}, the \textit{Perron-Frobenius theorem} yields      
\begin{equation}
\label{eq:eigen4}
\sum^{K}_{i=2}\prod^{i}_{j=2}\frac{(\epsilon+\frac{\lambda_F}{\bar{\Gamma}_{j-1}(\theta) })}{(1+\frac{\lambda_F}{\bar{\Gamma}_{j}(\theta) })}=\frac{\lambda_F}{\bar{\Gamma}_{1}(\theta) }
\end{equation}
Accounting for the feasibility condition $\lambda_F < 1$, the largest feasible cluster size falls within the range of $K_{\min}\leq K\leq K_{\max}$ where the bounds can be obtained from \eqref{eq:eigen4} as
\begin{align}
\label{eq:Krange}
& K_{\min}= \left \lceil \frac{\ln(\epsilon)}{ \ln \left(\frac{1+\epsilon \max_i (\bar{\Gamma}_{i}(\theta))}{1+\max_i (\bar{\Gamma}_{i}(\theta))} \right)} \right \rceil \text{, } \\
&K_{\max}=
 \left \lfloor \frac{\ln(\epsilon)}{\ln \left(\frac{1+\epsilon \min_i (\bar{\Gamma}_{i}(\theta))}{1+\min_i (\bar{\Gamma}_{i}(\theta))} \right)} \right \rfloor. 
\end{align}
This range tightens as the composite SINR requirements tightens and finally reduces to an exact cluster size of 
\begin{equation}
\label{eq:K*}
K(\epsilon, \theta )=\left \lfloor \frac{\ln(\epsilon)}{\ln \left(\frac{1+\epsilon \bar{\Gamma}(\theta)}{1+\bar{\Gamma}(\theta)} \right)} \right \rfloor, \text{ if }\bar{\Gamma}_i(\theta)=\bar{\Gamma}(\theta), \forall i
\end{equation}
Reverse engineering of \eqref{eq:K*} yields the attainable feasible SINR for a given cluster size as 
\begin{equation}
\label{eq:Gamma*}
\bar{\Gamma}^*(K(\epsilon, \theta))=\frac{e^{\ln(\epsilon)/K}-1}{\epsilon-e^{\ln(\epsilon)/K}}, \text{ if }\bar{\Gamma}_i(\theta)=\bar{\Gamma}(\theta), \forall i.
\end{equation} 
\end{proof}

\section{Proof for Lemma \ref{lem:CSconstrained}}
\label{app_lem_3}

\begin{proof}
Building upon Appendix \ref{app_lem_1}, optimal power levels can be derived directly by solving \eqref{eq:SINR-matrix} as follows
\begin{align}
\label{eq:p1}
p_1&=\frac{\bar{\Gamma}_1(\theta) \sigma^2}{1-\sum_{i=1}^K \bar{\Gamma}_i (\theta) \prod_{j=2}^{i} \left( \frac{1+\epsilon \bar{\Gamma}_{j-1}(\theta) }{1+\bar{\Gamma}_{j}(\theta)} \right) }\\
\label{eq:pk}
p_k&=p_{1}\frac{\bar{\Gamma}_{2}(\theta)}{\bar{\Gamma}_{1}(\theta)}\prod^{k}_{j=2} \left(\frac{1+\epsilon\bar{\Gamma}_{j-1}(\theta) }{1+\bar{\Gamma}_{j}(\theta)} \right), \: k \geq 2.
\end{align}
which can be simplified for $\bar{\Gamma}_i(\theta)=\bar{\Gamma}(\theta), \forall i,$ or $\max_i \left(\bar{\Gamma}_i(\theta) \right)=\bar{\Gamma}(\theta)$ as  
\begin{align}
\label{eq:pk_special}
\nonumber p_k&=\left( \frac{1+\epsilon \bar{\Gamma}(\theta)}{1+\bar{\Gamma}(\theta)}\right)^{k-1} \times \\
&  \frac{\bar{\Gamma} (\theta) \sigma^2}{1-\left( \frac{1+\epsilon \bar{\Gamma} (\theta)}{1-\epsilon}\right)+\left( \frac{1+\epsilon \bar{\Gamma}(\theta)}{1-\epsilon}\right) \left( \frac{1+\epsilon \bar{\Gamma}(\theta)}{1+\bar{\Gamma}(\theta)}\right)^{K-1}}, \:k \geq 2,
\end{align} 
which follows from the fact that the second term of denominator of \eqref{eq:p1} becomes a geometric series sum by setting $ \bar{\Gamma}_i(\theta)=\bar{\Gamma}(\theta), \forall i.$ For  $\bar{\Gamma}_i(\theta)=\bar{\Gamma}(\theta), \forall i$ or $\max_i \left(\bar{\Gamma}_i (\theta)\right)=\bar{\Gamma}(\theta)$, the largest feasible cluster size range can be obtained from \eqref{eq:pk_special} as $K_{\min}\leq K\leq K_{\max}$ where 
\begin{align}
\label{eq:KminConst}
 K_{\min}&=\left \lfloor 1+ \frac{\ln \left( \frac{\epsilon(1+\bar{\Gamma}(\theta))}{\epsilon-1}\right)-\ln \left( \frac{\bar{\Gamma}(\theta) \sigma^2}{\bar{P}_u g_K}-\frac{1+\epsilon\bar{\Gamma}(\theta)}{1-\epsilon}\right) }{\ln \left( \frac{1+\epsilon\bar{\Gamma}(\theta)}{1+ \bar{\Gamma}(\theta)}\right) } \right \rfloor, \\
\label{eq:KmaxConst}
K_{\max}&=\left \lceil 1+ \frac{\ln \left( \frac{\bar{\Gamma}(\theta) \sigma^2}{\bar{P}_u g_1}+  \frac{\epsilon(1+\bar{\Gamma}(\theta))}{1-\epsilon}\right) -  \ln \left( \frac{1+\epsilon\bar{\Gamma}(\theta)}{1-\epsilon} \right)}{\ln \left( \frac{1+\epsilon\bar{\Gamma}(\theta)}{1+ \bar{\Gamma}(\theta)}\right) } \right \rceil.
\end{align}
Equations \eqref{eq:KminConst} and  \eqref{eq:KmaxConst} are obtained by substituting \eqref{eq:pk_special} into $\bar{P}_u g_1 \geq p_1$ and $\bar{P}_u g_K \geq p_K$, respectively, rewriting for $K$, and taking the natural logarithm from both sides. 
\end{proof}

\section{Proof of Lemma \ref{lem:cfp}}
\label{app_lem_4}

\begin{proof}
This appendix explains how Table \ref{tab:pwr_cond} created and Lemma \ref{lem:cfp} is obtained. First, let us consider the index set of UEs who are active at maximum transmission power constraint, i.e., $\forall i \in \mathcal{I}_\lambda$. Obviously, such UEs set their power weights to unity, i.e., $\omega_i=1, \forall i \in \mathcal{I}_\lambda$, as in the first case of \eqref{eq:cfp}. The optimal power weights of remaining UEs who are active at CSC, i.e., $\forall i \in \mathcal{I}_\mu$, can be directly obtained from CSCs as follows:
\begin{align}
\nonumber w_{i} & \frac{h_i}{\bar{\Gamma}_i} = \epsilon \sum_{ \substack{ 1\leq j< \mathrm{\max_i}\\ j \in \mathcal{I}_{\lambda} } }h_{j}+\sum_{\substack{ \mathrm{\max_i} <j \leq K \\ j \in \mathcal{I}_{\lambda}  }  }h_{j}\\
&\vspace{-5pt}+\epsilon \sum_{ \substack{1\leq j<\mathrm{\max_i}\\j \in \mathcal{I}_{\mu}  }  }w_{j}h_{j}+\sum_{\substack{ \mathrm{\max_i} <j \leq K\\j \in \mathcal{I}_{\mu}  }  }w_{j}h_{j}  + \varrho,  \forall i \in \mathcal{I}_\mu \label{eq:omega_i}
\end{align}
where $\mathrm{max_i}=\argmax\{m \vert m \in \mathcal{I}_\mu , m<i\}$. Since UE$_i, \: \forall i \in \mathcal{I}_\mu,$ are active at CSCs, \eqref{eq:omega_i} is obtained by substituting $\omega_j=1, \forall j \in \mathcal{I}_\lambda$, and rewriting $R_i=\bar{\Gamma}_i$ for $\omega_i, \forall i \in \mathcal{I}_\mu$. Exploiting \eqref{eq:omega_i}, $\omega_i-\omega_{\mathrm{\max_i}}$ can be written as
 \begin{align}
\nonumber w_{i}\frac{h_{i}}{\bar{\Gamma}_{i}}-& w_{\mathrm{\max_i}}\frac{h_{\mathrm{\max_i}}}{\bar{\Gamma}_{\mathrm{\max_i}}} =\epsilon w_{\mathrm{\max_i}}h_{\mathrm{\max_i}}-w_{i}h_{i}\\
& +(\epsilon-1)\sum_{\substack{\mathrm{\max_i}<j<i\\j\in \mathcal{I}_{\lambda}}}h_{j}, ,  \forall i \in \mathcal{I}_\mu  \label{eq:omega_i-omega_i-1}
\end{align}
After some algebraic manipulations on \ref{eq:omega_i-omega_i-1}, the first-order non-homogeneous recurrence relations with variable coefficients can be obtained as
\begin{align}
\label{eq:recur}
 w_{i}=&w_{\mathrm{\max_i}}\frac{h_{\mathrm{\max_i}}\left(\epsilon+\frac{1}{\bar{\Gamma}_{\mathrm{\max_i}}}\right)}{h_{i}\left(1+\frac{1}{\bar{\Gamma}_{i}}\right)} +\frac{(\epsilon-1)\sum_{\substack{\mathrm{\max_i}<j<i\\j\in \mathcal{I}_{\lambda}}}h_{j}}{h_{i}\left(1+\frac{1}{\bar{\Gamma}_{i}}\right)},
\end{align}
$ \forall i \in \mathcal{I}_\mu$, which is apparently in the form of $w_{i}=w_{\mathrm{\max_i}}a_{i} +b_{i}$ where $a_{i}=\frac{h_{\mathrm{\max_i}}\left(\epsilon+\frac{1}{\bar{\Gamma}_{\mathrm{\max_i}}}\right)}{h_{i}\left(1+\frac{1}{\bar{\Gamma}_{i}}\right)}$ and $b_{i}=\frac{(\epsilon-1)\sum_{\substack{\mathrm{\max_i}<j<i\\j\in \mathcal{I}_{\lambda}}}h_{j}}{h_{i}\left(1+\frac{1}{\bar{\Gamma}_{i}}\right)}$. Accordingly, the recurrent relation in \eqref{eq:recur} can be rewritten as in \eqref{eq:cfp}
solution of which can be obtained as in $\omega_{\rm{mind}}$ by following the standard procedure given in \cite[Theorem 4.2]{Arne11}.
\end{proof}

\bibliographystyle{IEEEtran}

\bibliography{Final}

\begin{thebibliography}{10}
\providecommand{\url}[1]{#1}
\csname url@samestyle\endcsname
\providecommand{\newblock}{\relax}
\providecommand{\bibinfo}[2]{#2}
\providecommand{\BIBentrySTDinterwordspacing}{\spaceskip=0pt\relax}
\providecommand{\BIBentryALTinterwordstretchfactor}{4}
\providecommand{\BIBentryALTinterwordspacing}{\spaceskip=\fontdimen2\font plus
\BIBentryALTinterwordstretchfactor\fontdimen3\font minus
  \fontdimen4\font\relax}
\providecommand{\BIBforeignlanguage}[2]{{%
\expandafter\ifx\csname l@#1\endcsname\relax
\typeout{** WARNING: IEEEtran.bst: No hyphenation pattern has been}%
\typeout{** loaded for the language `#1'. Using the pattern for}%
\typeout{** the default language instead.}%
\else
\language=\csname l@#1\endcsname
\fi
#2}}
\providecommand{\BIBdecl}{\relax}
\BIBdecl

\bibitem{celik2017ULNOMA}
A.~Celik, R.~M. Radaydeh, F.~S. Al-Qahtani, A.~H.~A. El-Malek, and M.~S.
  Alouini, ``Resource allocation and cluster formation for imperfect {NOMA} in
  {DL}/{UL} decoupled {H}et{N}ets,'' in \emph{proc. IEEE Global Commun. Conf.
  (GLOBECOM) Workshops (GC Wkshps)}, Dec. 2017, pp. 1--6.

\bibitem{Dai15}
L.~Dai, B.~Wang, Y.~Yuan, S.~Han, C.~l.~I, and Z.~Wang, ``Non-orthogonal
  multiple access for {5G}: solutions, challenges, opportunities, and future
  research trends,'' \emph{IEEE Commun. Mag.}, vol.~53, no.~9, pp. 74--81,
  Sept. 2015.

\bibitem{Andrews2003}
J.~G. Andrews and T.~H. Meng, ``Optimum power control for successive
  interference cancellation with imperfect channel estimation,'' \emph{IEEE
  Trans. Wireless Commun.}, vol.~2, no.~2, pp. 375--383, Mar. 2003.

\bibitem{SIC_OFDMA}
N.~I. Miridakis and D.~D. Vergados, ``A survey on the successive interference
  cancellation performance for single-antenna and multiple-antenna {OFDM}
  systems,'' \emph{IEEE Commun. Surveys \& Tutorials}, vol.~15, no.~1, pp.
  312--335, First 2013.

\bibitem{Ding2015TVT}
Z.~Ding, P.~Fan, and H.~V. Poor, ``Impact of user pairing on {5G} nonorthogonal
  multiple-access downlink transmissions,'' \emph{IEEE Trans. Veh. Technol.},
  vol.~65, no.~8, pp. 6010--6023, Aug. 2016.

\bibitem{Liu2016fairness}
Y.~Liu, M.~Elkashlan, Z.~Ding, and G.~K. Karagiannidis, ``Fairness of user
  clustering in {MIMO} non-orthogonal multiple access systems,'' \emph{IEEE
  Commun. Lett.}, vol.~20, no.~7, pp. 1465--1468, Jul. 2016.

\bibitem{Liu2017Joint}
Z.~{Liu}, L.~{Lei}, N.~{Zhang}, G.~{Kang}, and S.~{Chatzinotas}, ``Joint
  beamforming and power optimization with iterative user clustering for
  {MISO-NOMA} systems,'' \emph{IEEE Access}, vol.~5, pp. 6872--6884, 2017.

\bibitem{Sedaghat2018}
M.~A. {Sedaghat} and R.~R. {Müller}, ``On user pairing in uplink {NOMA},''
  \emph{IEEE Trans. Wireless Commun.}, vol.~17, no.~5, pp. 3474--3486, May
  2018.

\bibitem{Kang2018}
J.~{Kang} and I.~{Kim}, ``Optimal user grouping for downlink {NOMA},''
  \emph{IEEE Wireless Commun. Lett.}, vol.~7, no.~5, pp. 724--727, Oct. 2018.

\bibitem{Elbamby2017}
M.~S. Elbamby, M.~Bennis, W.~Saad, M.~Debbah, and M.~Latva-aho, ``Resource
  optimization and power allocation in in-band full duplex-enabled
  non-orthogonal multiple access networks,'' \emph{IEEE J. Sel. Areas Commun.},
  vol.~35, no.~12, pp. 2860--2873, Dec. 2017.

\bibitem{Ali2017beam}
S.~Ali, E.~Hossain, and D.~I. Kim, ``Non-orthogonal multiple access ({NOMA})
  for downlink multiuser {MIMO} systems: User clustering, beamforming, and
  power allocation,'' \emph{IEEE Access}, vol.~5, pp. 565--577, 2017.

\bibitem{Ali2016}
M.~S. Ali, H.~Tabassum, and E.~Hossain, ``Dynamic user clustering and power
  allocation for uplink and downlink non-orthogonal multiple access ({NOMA})
  systems,'' \emph{IEEE Access}, vol.~4, pp. 6325--6343, 2016.

\bibitem{Ding2019Improved}
J.~{Ding}, J.~{Cai}, and C.~{Yi}, ``An improved coalition game approach for
  mimo-noma clustering integrating beamforming and power allocation,''
  \emph{IEEE Trans.Vehicular Technol.}, vol.~68, no.~2, pp. 1672--1687, Feb.
  2019.

\bibitem{Cui2018Unsupervised}
J.~{Cui}, Z.~{Ding}, P.~{Fan}, and N.~{Al-Dhahir}, ``Unsupervised machine
  learning-based user clustering in millimeter-wave-noma systems,'' \emph{IEEE
  Trans. Wireless Commun.}, vol.~17, no.~11, pp. 7425--7440, No.v 2018.

\bibitem{Tsai2018Quality}
Y.~{Tsai} and H.~{Wei}, ``Quality-balanced user clustering schemes for
  non-orthogonal multiple access systems,'' \emph{IEEE Commun. Lett.}, vol.~22,
  no.~1, pp. 113--116, Jan. 2018.

\bibitem{Fang2017Joint}
F.~{Fang}, H.~{Zhang}, J.~{Cheng}, S.~{Roy}, and V.~C.~M. {Leung}, ``Joint user
  scheduling and power allocation optimization for energy-efficient {NOMA}
  systems with imperfect {CSI},'' \emph{IEEE J. Sel. Areas Commun.}, vol.~35,
  no.~12, pp. 2874--2885, Dec. 2017.

\bibitem{Wei2017Optimal}
Z.~{Wei}, D.~W.~K. {Ng}, J.~{Yuan}, and H.~{Wang}, ``Optimal resource
  allocation for power-efficient {MC-NOMA} with imperfect channel state
  information,'' \emph{IEEE Trans. Commun.}, vol.~65, no.~9, pp. 3944--3961,
  Sep. 2017.

\bibitem{Xiao2018Reinforcement}
L.~{Xiao}, Y.~{Li}, C.~{Dai}, H.~{Dai}, and H.~V. {Poor}, ``Reinforcement
  learning-based {NOMA} power allocation in the presence of smart jamming,''
  \emph{IEEE Trans. Vehicular Technol.}, vol.~67, no.~4, pp. 3377--3389, Apr.
  2018.

\bibitem{Tabassum2017modeling}
H.~Tabassum, E.~Hossain, and J.~Hossain, ``Modeling and analysis of uplink
  non-orthogonal multiple access in large-scale cellular networks using poisson
  cluster processes,'' \emph{IEEE Trans. Commun.}, vol.~65, no.~8, pp.
  3555--3570, Aug. 2017.

\bibitem{Lei2016}
L.~Lei, D.~Yuan, C.~K. Ho, and S.~Sun, ``Power and channel allocation for
  non-orthogonal multiple access in {5G} systems: Tractability and
  computation,'' \emph{IEEE Trans. Wireless Commun.}, vol.~15, no.~12, pp.
  8580--8594, Dec. 2016.

\bibitem{Zhang2016}
N.~Zhang, J.~Wang, G.~Kang, and Y.~Liu, ``Uplink nonorthogonal multiple access
  in {5G} systems,'' \emph{IEEE Commun. Lett.}, vol.~20, no.~3, pp. 458--461,
  Mar. 2016.

\bibitem{Choi2018}
J.~Choi, ``On power and rate allocation for coded uplink {NOMA} in a
  multicarrier system,'' \emph{IEEE Trans. Commun.}, vol.~66, no.~6, pp.
  2762--2772, Jun. 2018.

\bibitem{Pappi2017}
K.~N. Pappi, P.~D. Diamantoulakis, and G.~K. Karagiannidis, ``Distributed
  uplink-{NOMA} for cloud radio access networks,'' \emph{IEEE Commun. Lett.},
  vol.~21, no.~10, pp. 2274--2277, Oct. 2017.

\bibitem{Yang2016}
Z.~Yang, Z.~Ding, P.~Fan, and N.~Al-Dhahir, ``A general power allocation scheme
  to guarantee quality of service in downlink and uplink {NOMA} systems,''
  \emph{IEEE Trans. Wireless Commun.}, vol.~15, no.~11, pp. 7244--7257, Nov.
  2016.

\bibitem{Celik2017DLNOMA}
A.~Celik, F.~S. Al-Qahtani, R.~M. Radaydeh, and M.~S. Alouini, ``Cluster
  formation and joint power-bandwidth allocation for imperfect {NOMA} in
  {DL}-{H}et{N}ets,'' in \emph{proc. IEEE Global Commun. Conf. (GLOBECOM)},
  Dec. 2017, pp. 1--6.

\bibitem{Celik2019Distributed}
A.~{Celik}, M.~{Tsai}, R.~M. {Radaydeh}, F.~S. {Al-Qahtani}, and M.~{Alouini},
  ``Distributed cluster formation and power-bandwidth allocation for imperfect
  {NOMA} in {DL}-{H}et{N}ets,'' \emph{IEEE Trans. Commun.}, vol.~67, no.~2, pp.
  1677--1692, Feb. 2019.

\bibitem{boccardi2016decouple}
F.~Boccardi, J.~Andrews, H.~Elshaer, M.~Dohler, S.~Parkvall, P.~Popovski, and
  S.~Singh, ``Why to decouple the uplink and downlink in cellular networks and
  how to do it,'' \emph{IEEE Commun. Mag.}, vol.~54, no.~3, pp. 110--117, Mar.
  2016.

\bibitem{celik2017joint}
A.~Celik, R.~M. Radaydeh, F.~S. Al-Qahtani, and M.-S. Alouini, ``Joint
  interference management and resource allocation for device-to-device ({D2D})
  communications underlying downlink/uplink decoupled ({DUD}e) heterogeneous
  networks,'' in \emph{proc. IEEE Intl. Conf. Commun.(ICC)}, May 2017.

\bibitem{Hasan2004cancellation}
A.~{Hasan} and J.~{Andrews}, ``Cancellation error statistics in a
  power-controlled cdma system using successive interference cancellation,'' in
  \emph{8th IEEE International Symposium on Spread Spectrum Techniques and
  Applications}, Aug. 2004, pp. 419--423.

\bibitem{Andrews2005SIC}
J.~G. Andrews, ``Interference cancellation for cellular systems: a contemporary
  overview,'' \emph{IEEE Wireless Commun.}, vol.~12, no.~2, pp. 19--29, Apr.
  2005.

\bibitem{origin_alpha}
J.~Mo and J.~Walrand, ``Fair end-to-end window-based congestion control,''
  \emph{IEEE/ACM Trans. Netw.}, vol.~8, no.~5, pp. 556--567, Oct. 2000.

\bibitem{general_alpha}
E.~Altman, K.~Avrachenkov, and A.~Garnaev, ``Generalized $\alpha$-fair resource
  allocation in wireless networks,'' in \emph{2008 47th IEEE Conference on
  Decision and Control}, Dec. 2008, pp. 2414--2419.

\bibitem{Andrews2005iterative}
A.~Agrawal, J.~G. Andrews, J.~M. Cioffi, and T.~Meng, ``Iterative power control
  for imperfect successive interference cancellation,'' \emph{IEEE Trans.
  Wireless Commun.}, vol.~4, no.~3, pp. 878--884, May 2005.

\bibitem{kahan1996ieee}
W.~Kahan, ``Ieee standard 754 for binary floating-point arithmetic,''
  \emph{Lecture Notes on the Status of IEEE}, vol. 754, no. 94720-1776, p.~11,
  1996.

\bibitem{horn1990matrix}
R.~A. Horn and C.~R. Johnson, \emph{Matrix analysis}.\hskip 1em plus 0.5em
  minus 0.4em\relax Cambridge university press, 1990.

\bibitem{chong2013introduction}
E.~K. Chong and S.~H. Zak, \emph{An introduction to optimization}.\hskip 1em
  plus 0.5em minus 0.4em\relax John Wiley \& Sons, 2013, vol.~76.

\bibitem{Arne11}
A.~Jensen, ``Lecture notes on difference equations,'' Jul. 2011.

\end{thebibliography}

\vspace*{-3\baselineskip}
\begin{IEEEbiography}[{\includegraphics[width=1.1in,height=1.25in]{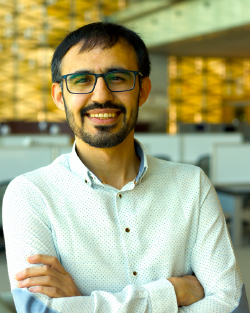}}]{Abdulkadir \c Celik}
(S'14-M'16) received the B.S. degree in electrical-electronics engineering from Sel\c cuk University, Konya, Turkey in 2009, the M.S. degree in electrical engineering in 2013, the M.S. degree in computer engineering in 2015, and the Ph.D. degree in co-majors of electrical engineering and computer engineering in 2016, all from Iowa State University, Ames, IA, USA. He is currently a postdoctoral research fellow at Communication Theory Laboratory of King Abdullah University of Science and Technology (KAUST). His current research interests include but not limited to 5G and beyond, wireless data centers, UAV assisted cellular and IoT networks, and underwater optical wireless communications, networking, and localization.
\end{IEEEbiography}
\vspace*{-3\baselineskip}
\begin{IEEEbiography}[{\includegraphics[width=1.1in,height=1.25in]{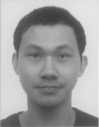}}]{Ming-Cheng Tsai}(S'18) was born in Fujian, China. He received his B.E. degree in electrical engineering from the National Taipei University of Technology, Taipei, Taiwan, in 2015. From 2015 to 2018, he was a student of Communication Engineering at the National Tsinghua University. He is currently pursuing his M.S./Ph.D. degree in Electrical Engineering at King Abdullah University of Science and Technology. His research interests include device to device communication, digital communication, and error correcting codes. 
\end{IEEEbiography}
\vspace*{-3\baselineskip}
\begin{IEEEbiography}[{\includegraphics[width=1.1in,height=1.25in]{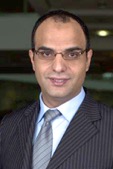}}]{Redha M. Radaydeh}(S'05-M'07-SM'13) was born in Irbid, Jordan, on November 12, 1978. He received the B.S. and M.S. degrees from Jordan University of Science and Technology (JUST), Irbid, in 2001 and 2003, respectively, and the Ph.D. degree from University of Mississippi, Oxford, MS, USA, in 2006, all in electrical engineering. He worked at JUST, King Abdullah University of Science and Technology (KAUST), Texas A\& M University at Qatar (TAMUQ), and Alfaisal University as an associate professor of electrical engineering. He also worked as a remote research scientist with KAUST and was a visiting researcher with Texas A\& M University (TAMU), College Station, TX. Currently, he is a faculty member with the electrical engineering program at Texas
A\& M University-Commerce (TAMUC), Commerce, TX. His research interests include broad topics on wireless communications, and design and performance analysis of wireless networks.
\end{IEEEbiography}
\vspace*{-3\baselineskip}
\begin{IEEEbiography}[{\includegraphics[width=1.1in,height=1.25in]{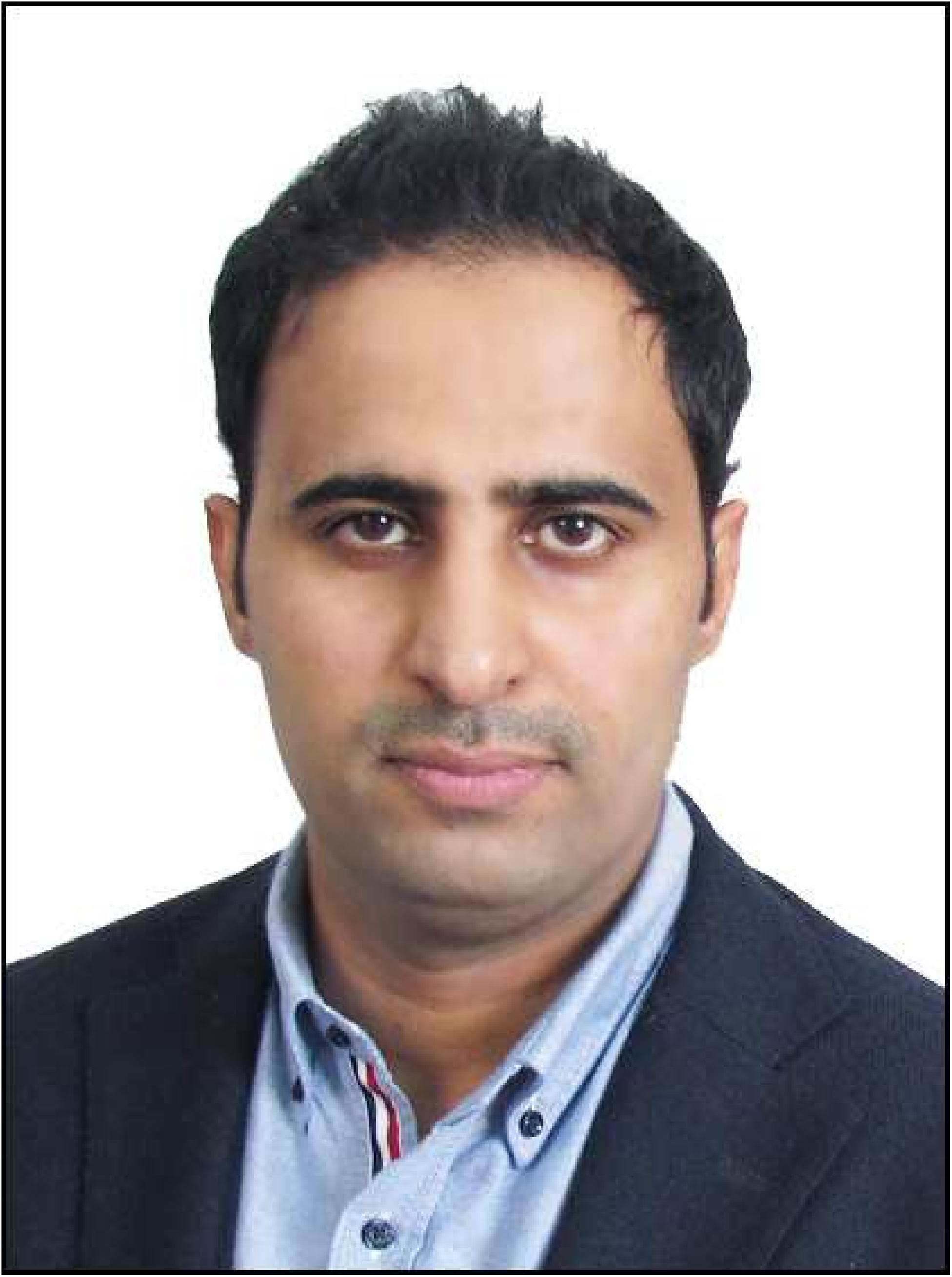}}]{Fawaz S. Al-Qahtani}
(M'10) received the B.Sc. in electrical engineering from King Fahad University of Petroleum and Minerals (KFUPM), Saudi Arabia in 2000 and M.Sc. in Digital Communication Systems from Monash University, Melbourne, Australia in 2005, and Ph.D. degree in Electrical and Computer Engineering, from RMIT University, Australia in December, 2009. He is currently working as IP commercialization manager for ICT portfolio at Research, Development, and Innovation (RDI) in Qatar Foundation, Doha, Qatar. From May 2010 to August 2017, he was research scientist with Texas A\& M University at Qatar, Education City, Doha, Qatar. His research has been sponsored by Qatar National Research Fund (QNRF). He was awarded of JSERP and NPRP projects. His current research interests include channel modeling, applied signal processing, MIMO communication systems, cooperative communications, cognitive radio systems, free space optical, physical layer security, and device-to-device communication. He is the author or co-author of 100 technical papers published in scientific journals and presented at international conferences.
\end{IEEEbiography}
\vspace*{-3\baselineskip}
\begin{IEEEbiography}[{\includegraphics[width=1.1in,height=1.25in]{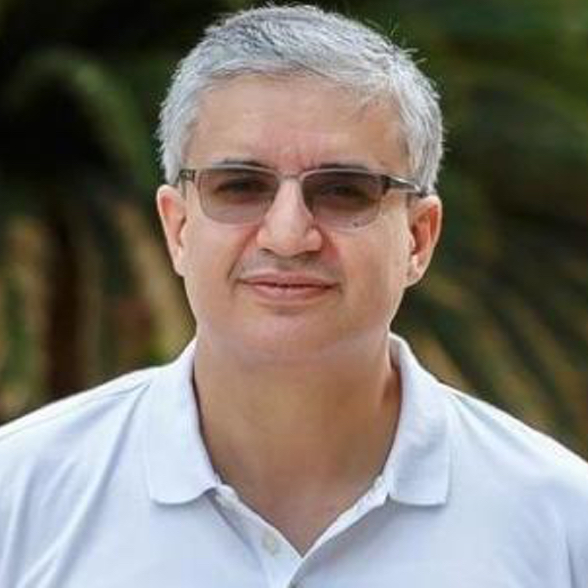}}]{Mohamed-Slim Alouini}
(S'94-M'98-SM'03-F'09) was born in Tunis, Tunisia. He received the Ph.D. degree in Electrical Engineering from the California Institute of Technology (Caltech), Pasadena, CA, USA, in 1998. He served as a faculty member in the University of Minnesota, Minneapolis, MN, USA, then in the Texas A\&M University at Qatar, Education City, Doha, Qatar before joining King Abdullah University of Science and Technology (KAUST), Thuwal, Makkah Province, Saudi Arabia as a Professor of Electrical Engineering in 2009. His current research interests include the modeling, design, and performance analysis of wireless communication systems.
\end{IEEEbiography}

\end{document}